\documentclass[onecolumn]{article}
\pdfoutput=1
\usepackage{graphicx, hyperref}
\usepackage{amsmath,amsthm}
\usepackage{amssymb}
\usepackage{slashed}
\usepackage{stmaryrd}
\usepackage{yfonts}
\usepackage[T1]{fontenc}
\usepackage[utf8]{inputenc}
\usepackage[english]{babel}
\usepackage{tikz-cd}
\usepackage{mathabx}
\usepackage{braket}
\usepackage{physics}
\usepackage{geometry, color}
\usepackage{color}
\usepackage{imakeidx, xparse, mathtools,
            xcolor,  tensor, 
            soul, graphicx ,titlesec,  appendix,}

\theoremstyle{theorem}
\newtheorem{theorem}{Theorem}[section]
\newtheorem{lemma}{Lemma}[section]
\newtheorem{corollary}{Corollary}[section]
\newtheorem{proposition}{Proposition}[section]

\theoremstyle{remark}
\newtheorem{remark}{Remark}[section]

\theoremstyle{example}
\newtheorem{example}{Example}[section]

\theoremstyle{definition}
\newtheorem{definition}{Definition}[section]

\usepackage{mathtools}

\usepackage{mathrsfs}
\numberwithin{equation}{section}

\newcommand{\beq}{\begin{eqnarray}}
\newcommand{\eeq}{\end{eqnarray}}

\hypersetup{
pdfstartview = {FitH},
}
\hypersetup{
	colorlinks=true,         
	linkcolor=brown,          
	citecolor=red,        
	urlcolor=blue            
}

\def\d{\mathfrak d}
\def\g{\mathfrak g}

\def\so{\mathfrak{so}}
\def\an{\mathfrak{an}}
\def\su{\mathfrak{su}}

\def\R{\mathbb{R}}

\usepackage{atbegshi}
\AtBeginDocument{\AtBeginShipoutNext{\AtBeginShipoutDiscard}\addtocounter{page}{-1}}

\usepackage{authblk}
  \allowdisplaybreaks

\begin{document}
\title{(2-)Drinfel'd Double and (2-)BF Theory}

\author[1]{{ \sf Hank Chen}}\thanks{hank.chen@uwaterloo.ca}


\author[1]{{\sf Florian Girelli}}\thanks{fgirelli@uwaterloo.ca}

\affil[1]{\small Department of Applied Mathematics, University of Waterloo, 200 University Avenue West, Waterloo, Ontario, Canada, N2L 3G1}



\maketitle

\begin{abstract}
The gauge symmetry 
and shift/translational symmetry 
of a 3D BF action, which are associated to a pair of dual Lie algebras, can be combined to form the Drinfel'd double.  
This combined symmetry is the gauge symmetry of the Chern-Simons action which is equivalent to the BF action, up to some boundary term. We show that something similar happens in 4D when considering a 2-BF action (aka BFCG action), whose symmetries are specified in terms of a pair of dual strict Lie 2-algebras (ie. crossed-modules). 
Combining these symmetries gives rise to a 2-Drinfel'd double which becomes the gauge symmetry structure of a 4D BF theory, up to a boundary term.   Concretely, we show how using 2-gauge transformations based on dual crossed-modules, the notion of 2-Drinfel'd double defined in \cite{Bai_2013} appears. We also discuss how, similarly to the Lie algebra case, the symmetric contribution of the $r$-matrix of the 2-Drinfel'd double can be interpreted as a quadratic 2-Casimir, which allows to recover the notion of duality. 
\end{abstract}

\tableofcontents

\section{Introduction}
Quantum groups, or Hopf algebras, and their category of representations, provide the symmetry structure to discuss  topological quantum field theories in three dimensions \cite{Witten:1988hc, Crane:1994ty, Baez:2009as, Turaev:2017uxl}. Among these, the specific example of the Drinfel'd double \cite{Drinfeld:1986in, Majid:1996kd} plays a special role. It is the relevant quantum group to discuss the quantization of three dimensional gravity \cite{Noui:2006ku, Meusburger:2008bs}. It is the symmetry inherent to the Kitaev model \cite{Kitaev1997, BMCA} which could be used to build quantum computers (see also \cite{girelli:2017lfn} for some other symmetry than the double). Typically it is built out of a pair of Hopf algebras that are dual to each other.

The deformed notion of symmetry encoded in a quantum group can be derived \cite{Dupuis:2020ndx} or guessed \cite{Fock:1992xy, Fock:1998nu} from an action, such as the Chern-Simons action and (equivalently in some cases) the BF action. When using an action as starting point, the key step to see quantum groups structure to appear involves a discretization procedure to recover some type of (generalized) lattice gauge theory. At the classical level, this generates Poisson-Lie groups which are the classical analogue of the notion of quantum group \cite{Chari:1994pz}.  Poisson-Lie groups are Lie groups equipped with a Poisson bracket which is compatible with the group structure. In certain cases, the Poisson bracket is characterized by an object called the $r$-matrix, which upon quantization becomes the quantum $R$-matrix. The Drinfel'd double is an example of Poisson-Lie group.   Looking at the infinitesimal picture, ie. looking at the Lie algebras, we have the so-called classical (Drinfel'd) double construction (also called Manin triple): we have a pair of bialgebras which are dual to each other, and the cocycle of one is dual to the Lie bracket of the other \cite{Semenov1992, Chari:1994pz}. These cocycles are also the infinitesimal limit of the Poisson brackets compatible with the relevant group structure.  The symmetric part of the $r$-matrix is given in terms of the quadratic Casimir of the Lie algebra behind the Drinfel'd double. This Casimir encodes the duality structure. 

Upon discretization of the BF theory, the Drinfel'd double as a Poisson-Lie group arises naturally as the structure of the symmetry \cite{Dupuis:2020ndx}. This means in particular that one can expect that at the level of the action, somehow the classical double should be present in the structure of the gauge symmetries. This is kind of an understood fact, but here we will show explicitly how this works as it will be useful for the generalization to higher-gauge symmetries. We will also emphasize the different formulation of the action, either invariant under the full double, which would be the Chern-Simons formulation, or in terms of the pair of components of the double, which would amount to the BF formulation.  

\medskip

What happens when we deal with four dimensional topological quantum field theories expressed in terms of BF-type theories \cite{Baez:1995ph}? As discussed in different places, when we go up in dimension, we go up in the categorical ladder \cite{Crane:1994ty,Baez:1995xq, Baez:2009as}. For the four dimensional case, we expect to use the notion of quantum 2-groups (see \cite{majid:2012gy} for a proposal), where a 2-group is a categorification of the notion of group \cite{Baez2004}. Note that there are several options to define the notion of quantum 2-groups. Here we will follow the intuition given by gauge theory and the associated topological actions one can construct. The definition of a 2-group in terms of a crossed-module is likely the most suited to discuss  gauge theory generalizations \cite{Baez:2002jn, Baez:2004in, Baez:2005qu, Girelli:2003ev, Baez:2010ya}. It is not however the best one to discuss the notion of dualization, which is better defined in terms of Lie 2-algebras \cite{Baez:2003fs, Bai_2013}. Indeed this latter formulation, while equivalent, emphasizes better what needs to be dualized --- in particular, the notion of \textit{grading} plays a key role. 

While a general general framework for the notion of quantum 2-group is missing, the notion of Poisson-Lie 2-group has been introduced \cite{Chen:2012gz} and one can expect that one can provide more specific examples of Poisson-Lie 2-groups when dealing with discretization of topological actions. Just as in 3D, there are 2 main types of action for building out a topological model, the BF type \cite{Baez:1995ph}, which is just the higher dimensional version of the 3D one, or the BFCG (aka 2-BF) actions \cite{Girelli:2007tt, Martins:2010ry}. The former is usually interpreted as having  gauge symmetries and shift symmetries, which can be ultimately packaged in terms of a 2-gauge symmetry \cite{Baez:2010ya}. The latter one is explicitly built on 2-gauge symmetries, as well as some shift symmetries \cite{Girelli:2007tt, Martins:2010ry}.    
Similarly to Chern-Simons and BF theory, some examples of 2-BF models on 2-groups have been repackaged into a new BF action  where the variables are grouped together \cite{Mikovic:2011si}. 

\smallskip

These similarities between the 3D and 4D cases beg the question, of whether a notion of 2-Drinfel'd classical double might be at play when dealing with 2-gauge theories. Roughly speaking, we would expect a 2-Drinfel'd classical double to consist in a pair of crossed-modules (more precisely strict Lie 2-bialgebras \cite{ Bai_2013,Chen:2013}) that are dual in the appropriate sense to each other, and such that the cocycle information of one would specify the structure of the other. This object has already been defined in Ref. \cite{Bai_2013}. Here we would like to use the gauge theoretic formulation and re-derive the notion of 2-Drinfel'd classical double. This is one of the main results of the paper. We will then show that the 2-BF theories can be repackaged into a "larger" BF theory with symmetry  specified by the 2-Drinfel'd double. This will mean conversely that 4D BF theories can be specified in terms of 2-BF theories, and that there is an underlying 2-gauge symmetry at play. This has some implications in terms of the discretization of these theories; see for example \cite{Girelli:2021zmt} in the case of the Poincar\'e 2-group. 
As a specific case, we will show that the standard Lorentz $BF$ theory     can also be seen as a kind of 2-BF theory, in terms of a 3D  (quantum) Poincar\'e 2-group. 
Ultimately, this type of considerations could be relevant to discuss quantum gravity models, since 4D gravity can be seen as a constrained topological theory \cite{Barrett:1997gw, Baez:1999in, Baez:1999sr, rovelli2004}.
Indeed the choice of symmetry structure influences the construction of the quantum states which are built from the representations of the symmetry structure. For example\footnote{Note that the choice of symmetry structure already appears in 3D, where one can use the Drinfel'd double $\g\bowtie\g^*$ as the main symmetry structure, as in the Fock-Rosly approach \cite{Fock:1992xy, Fock:1998nu}, and associated combinatorial quantization \cite{Alekseev:1994pa, Alekseev:1994au}, or the gauge symmetry in terms of $\g$ or $\g^*$ \cite{Delcamp:2018sef}.}, if we deal with the  BF theory, we would construct the discrete/quantum theory using representation of the gauge symmetries. On the other hand, if we use instead the 2-BF formulation in terms of 2-gauge, we would use the theory of representations of the crossed-module. 
Since we are talking about the same theory, namely $BF$ theory, the quantum $BF$ amplitude defined in terms of group representations \cite{Baez:1999sr} or 2-group representations should be equivalent, just like in 3D where the Chern-Simons amplitude can be related to the Turaev-Viro/Ponzano-Regge amplitudes \cite{ROBERTS, Freidel:2004nb}.

\smallskip

When defining the 2-Drinfel'd double, we can introduce a generalization of the notion of $r$-matrix \cite{Bai_2013}. We show how something similar to the 3D case, namely whether its symmetric component is encoding some notion of "Casimir" for the Lie 2-algebra behind the 2-Drinfel'd double happens. 

\smallskip

The article is organized as follows. 

In Section \ref{bfdd}, we illustrate our strategy with the standard Drinfel'd double. We show how starting from a pair of dual gauge theories, we can recover the Drinfel'd double. We recall then the standard construction to define a topological action which is invariant under the Drinfel'd double and how it can be broken down into components to recover the $BF$ formulation in terms of the dual gauge theories. 

In Section \ref{algxmod}, we recall the definition of crossed-modules, how to recover a dual crossed-module and the (co-)adjoint action of a crossed-module.  

In Section \ref{2bfdd}, we provide one of the main results of the paper namely, we reconstruct the notion of 2-Drinfeld double by considering dual 2-gauge transformations. Considering that their order in which we apply them should not matter we recover constraints on the dual crossed-modules which amount exactly to the definition of the 2-Drinfeld double provided in \cite{Bai_2013}. 

Identifying the 2-Drinfel'd double symmetries allows in Section \ref{combo2cov} to construct an action invariant under such symmetries, which happens to be some 4D $BF$ theory type. We then explain how such $BF$ theory can also be seen as a pair of dual 2-gauge theories. As a particular example we focus on Lorentz $BF$ theory. This is another of the highlights of the paper.

Finally, in Section \ref{2grclrmat},  we discuss how the $r$-matrix structure of the 2-Drinfeld double also contains the information about the duality used to defined the action, just like in the standard Drinfel'd double case.

\section{Recovering the Drinfel'd double as a gauge symmetry from dual gauge theories }\label{bfdd}
In this section, we begin with a demonstration of our strategy to extract the structure of the Drinfel'd double from gauge theory. 

We consider a pair of dual gauge theories in the sense that their underlying Lie algebras, $\g$, $\g^*$ are dual to each other. Because of the duality, we can induce gauge transformations of $\g$ on $\g^*$ and vice-versa, through the co-adjoint representations.  We can then perform gauge transformations of $\g$ and $\g^*$ on a sector, say $\g$ in different order. Requiring that we get the same result, whatever the order put strong constraints on the shape of the transformations. The constraints are equivalent to the  compatibility relations of a matched pair of Lie algebras $\g\bowtie\g^*$ \cite{Majid:2008iz}.  To have such compatibilitiy relations means that we need to modify accordingly the notion of covariant derivative in each sector $\g$ and $\g^*$. This modification is what we called the antisymmetrization procedure. As a result, the covariant derivative (and therefore the gauge transformation) in a given different sector will then depend on \textit{both} connections.  

These new gauge transformations can  be combined together, these are the Drinfel'd double gauge transformations.

We will then discuss how we can construct a theory invariant under such theory either in terms of connections based on the Drinfel'd double $\g\bowtie\g^*$ as a gauge symmetry (Chern-Simons), or based on the components $\g$ and $\g^*$ (BF theory).

\subsection{Dual gauge theories}

\paragraph{Gauge theory}
Let $P\rightarrow M$ denote a principal $G$-bundle over a $d$-dimensional manifold $M$. 
It is standard lore that the $G$-connection $A\in\Omega^1(M)\otimes\mathfrak{g}$ and its associated curvature
\begin{equation}
    F= dA + \frac{1}{2}[A\wedge A]\in\Omega^2(M)\otimes\mathfrak{g}\nonumber
\end{equation}
undergoes the following gauge transformations
\begin{equation}
    A\rightarrow A^\lambda = A + d_A\lambda,\qquad F\rightarrow F^\lambda = F + [F,\lambda] \label{gauge}
\end{equation}
parameterized by a 0-form $\lambda\in\Omega^0(M)\otimes\mathfrak{g}$, where $\mathfrak{g}=\operatorname{Lie}G$ is the Lie algebra of $G$.

\paragraph{Dual gauge theory}
We can construct another principal bundle, $P'\rightarrow M$ which is a $G^*$-bundle over the $d$-dimensional manifold $M$. For this we introduce $\mathfrak{g}^*$ the dual vector space to $\mathfrak{g}$. 
 $\mathfrak{g}^*$ can also be seen as a Lie algebra. To this aim, we  introduce a Lie bracket $[\cdot,\cdot]_*$ on $\mathfrak{g}^*$. This can be accomplished by the choice of a 1-cochain
\begin{equation}
    \psi: \mathfrak{g}\rightarrow\mathfrak{g}^{2\wedge} \nonumber
\end{equation}
such that $[\cdot,\cdot]_*$ is defined as
\begin{equation}
    [g,g']_*(X) = (g\wedge g')(\psi(X)),\qquad \forall ~g,g'\in\mathfrak{g}^*,X\in\mathfrak{g}.\nonumber
\end{equation}
The Jacobi identity for the bracket is equivalent to the {\bf cocycle condition} \cite{Semenov1992} which is expressed in term of the coadjoint action $\operatorname{ad}^*$ (see below)
\begin{equation}
    d\psi(X,X') = \psi([X,X']) - D(\operatorname{ad}^*_X)\psi(X') + D(\operatorname{ad}^*_{X'})\psi(X)=0,\label{bialgcoh} 
\end{equation}
where $D(X) = X\otimes 1+ 1\otimes X$ is the {\it comultiplication map}. As such, $\psi \in Z^1(\mathfrak{g},\mathfrak{g}^{2\wedge})$ is a {\bf Lie algebra 1-cocycle}, and $(\mathfrak{g};\psi)$ a {\bf Lie bialgebra} \cite{Semenov1992}. $G^*$ would then be the Lie group associated to $\g^*$.

\begin{remark}\label{bialgsymm}
The bialgebra property is symmetric under dualization: if $\mathfrak{g}$ is a Lie bialgebra, so is its dual $\mathfrak{g}^*$, which is equipped with a 1-cocycle $\psi^*\in C^1(\mathfrak{g}^*,(\mathfrak{g}^*)^{2\wedge})$ dual to the Lie bracket of $\mathfrak{g}$. The situation can be summarized below
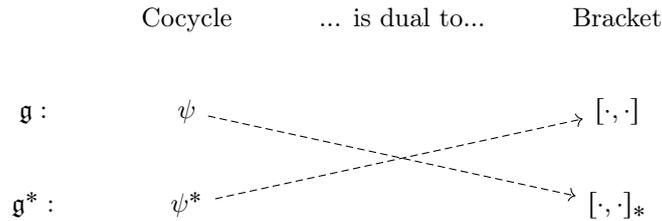
\begin{equation}
    \begin{tikzcd}
                & \text{Cocycle}             & \text{... is dual to...} & \text{Bracket}    \\
\mathfrak{g}:   & \psi \arrow[rrd, dashed]   &                          & {[\cdot,\cdot]}   \\
\mathfrak{g}^*: & \psi^* \arrow[rru, dashed] &                          & {[\cdot,\cdot]_*}
\end{tikzcd}\nonumber
\end{equation}
\end{remark}

Hence to make the construction, we could also start for $G^*$ and its Lie algebra $\g^*$ instead of $G$ and its Lie algebra $\g$. 

\medskip 

Let us now turn  to the gauge structure generated by the dual space $\mathfrak{g}^*$.
If we let $B\in\Omega^1(M)\otimes\mathfrak{g}^*$ denote a $G^*$-connection and 
\begin{equation}
    \tilde F = dB + \frac{1}{2}[B\wedge B]_* \nonumber
\end{equation}
its curvature, then we would also have the following gauge transformation
\begin{equation}
    B\rightarrow B^{\tilde \lambda} = B + d_B\tilde\lambda,\qquad \tilde F\rightarrow \tilde F^{\tilde\lambda} =\tilde F + [\tilde F,\tilde\lambda]_* \label{dualgauge}
\end{equation}
parameterized by a 0-form gauge parameter $\tilde\lambda\in\Omega^0(M)\otimes\mathfrak{g}^*$.

\paragraph{Coadjoint representations.} Because the Lie algebras $\g$ and $\g^*$ are dual to each other, we can dualize their respective adjoint actions, which are respectively 
\begin{equation}
    \operatorname{ad}_X(X') =[X,X'],\qquad \mathfrak{ad}_g(g')=[g,g']_*,\nonumber 
\end{equation}
where $X,X'\in\mathfrak{g},g,g'\in\mathfrak{g}^*$. By dualizing with respect to the canonical evaluation pairing $(g,X) = g(X)$,  a pair of {\bf coadjoint representations} are induced,
\begin{equation}
    \operatorname{ad}^*:\mathfrak{g}\rightarrow\operatorname{End}\mathfrak{g}^*,\qquad \mathfrak{ad}^*:\mathfrak{g}^*\rightarrow\operatorname{End}\mathfrak{g},\nonumber
\end{equation}
defined by
\begin{equation}
    \operatorname{ad}_X^*g(X') = -g(\operatorname{ad}_XX'),\qquad g'(\mathfrak{ad}_g^* X) = - \mathfrak{ad}_gg'(X), \nonumber 
\end{equation}
for which the evaluation pairing is {\it by definition} invariant. This is then inducing an invariant symmetric non-degenerate  bilinear form \cite{Semenov1992}
\begin{equation}
    \langle\langle X + g, X'+g'\rangle\rangle = g(X') + g'(X)\label{pair}
\end{equation}
on the direct sum $\mathfrak{g}\oplus\mathfrak{g}^*$. The connections $A,B$ admit a transformation under $\tilde\lambda,\lambda$, respectively, through these codajoint representations
\begin{equation}
    A\rightarrow A + \mathfrak{ad}_{\tilde\lambda}^*A,\qquad B\rightarrow B+\operatorname{ad}_\lambda^*B.\label{gaugedual}
\end{equation}

To summarize, because we have a pair of dual Lie algebras, we can introduce a gauge transformation of one sector say $\g$ on the dual one $\g^*$ and vice-versa.  Indeed, the transformations Eqs. \eqref{gauge}, \eqref{dualgauge}, \eqref{gaugedual} allow us to study, in each $\mathfrak{g}$- or $\mathfrak{g}^*$-sector, the gauge transformations parameterized by both $\lambda$ and $\tilde\lambda$. In the following, we shall leverage this observation to derive compatibility conditions on the coadjoint representations $\operatorname{ad}^*,\mathfrak{ad}^*$ in order for the direct sum $\mathfrak{g}\oplus\mathfrak{g}^*$ to form a {\bf Manin triple}, which is an equivalent description of the Drinfel'd double  \cite{Chari:1994pz}.

\subsection{The Drinfel'd double as gauge symmetry}\label{dd}
Let us now consider how the connections $A,B$ transform under both $\lambda,\tilde\lambda$. Since the sectors $\mathfrak{g}$ and $\mathfrak{g}^*$ are acting on each other, the gauge parameters will themselves be subject of gauge transformations.
For example, the pure-gauge $B=d\tilde\lambda$ transforms non-trivially under $\lambda$; similarly for the pure-gauge $A=d\lambda$ under $\tilde\lambda$:
\begin{eqnarray}
    d\tilde\lambda &\rightarrow& d\tilde\lambda + \operatorname{ad}_\lambda^*d\tilde\lambda =d\tilde\lambda+d(\mathfrak{ad}_{\tilde\lambda}^*\lambda) -\mathfrak{ad}_{d\tilde\lambda}^*\lambda\nonumber\\
    d\lambda &\rightarrow& d\lambda + \mathfrak{ad}_{\tilde\lambda}^*d\lambda=d\lambda+d(\operatorname{ad}_{\lambda}^*\tilde\lambda) -\operatorname{ad}_{d\lambda}^*\tilde\lambda,\nonumber
\end{eqnarray}
where we have used the Leibniz rule. 

Now if we assume, in the first equation, that $\tilde\lambda$ is slowly-varying (so that $\lambda d\tilde \lambda$ is small), then this induces a transformation of the gauge parameter $\lambda$ itself 
\begin{equation}
    \lambda\rightarrow \lambda + \mathfrak{ad}_{\tilde\lambda}^*\lambda + o(\lambda d\tilde\lambda).\nonumber
\end{equation}
Similarly, if we assume $\lambda$ to be slowly-varying, then we have
\begin{equation}
    \tilde\lambda\rightarrow\tilde\lambda + \operatorname{ad}_\lambda^*\tilde\lambda + o(\tilde\lambda d\lambda)\nonumber
\end{equation}
in the dual sector. 

In the following, we denote by $\boldsymbol\lambda=\lambda+\tilde\lambda$ the combined gauge parameter, and a direction of an arrow over it indicates the order at which the gauge transformation parameterized by which parameter is performed. For instance, $\overrightarrow{\boldsymbol\lambda}$ indicates that we perform a gauge transformation by $\lambda$ first then $\tilde\lambda$, and vice versa for $\overleftarrow{\boldsymbol\lambda}$.

\paragraph{Combined gauge transformation.} 
With this, we may then use Eqs. \eqref{gauge}, \eqref{gaugedual} to compute the transformation law
\begin{eqnarray}
     \overrightarrow{\boldsymbol\lambda} &:& A \xrightarrow{\lambda} A + d_A\lambda\nonumber \\
     &\qquad&\xrightarrow{\tilde\lambda} A + \mathfrak{ad}_{\tilde\lambda}^*A + d_A\lambda + \mathfrak{ad}_{\tilde\lambda}^*d\lambda \nonumber \\
     &\qquad& \qquad + [\mathfrak{ad}_{\tilde\lambda}^*A,\lambda] + [A,\mathfrak{ad}_{\tilde\lambda}^*\lambda]+ [\mathfrak{ad}_{\tilde\lambda}^*A,\mathfrak{ad}_{\tilde\lambda}^*\lambda],\nonumber
\end{eqnarray}
while
\begin{eqnarray}
     \overleftarrow{\boldsymbol\lambda} &:& A\xrightarrow{\tilde\lambda}A + \mathfrak{ad}_{\tilde\lambda}^*A \nonumber \\
     &\qquad &\xrightarrow{\lambda} A+d_A\lambda + \mathfrak{ad}_{\tilde\lambda}^*A + \mathfrak{ad}_{\tilde\lambda}^*d\lambda \nonumber \\
     &\qquad&\qquad + \mathfrak{ad}_{\tilde\lambda}^*[A,\lambda]+ \mathfrak{ad}_{\operatorname{ad}_\lambda^*\tilde\lambda}^*A+ \mathfrak{ad}_{\operatorname{ad}_\lambda^*\tilde\lambda}^*d_A\lambda.\nonumber
\end{eqnarray}
If we neglect quadratic terms in the small gauge parameter $\lambda$, we obtain the difference
\begin{equation}
    [\mathfrak{ad}_{\tilde\lambda}^*A,\lambda] + [A,\mathfrak{ad}^*_{\tilde\lambda}\lambda] - \mathfrak{ad}_{\tilde\lambda}^*[A,\lambda] - \mathfrak{ad}_{\operatorname{ad}_\lambda^*\tilde\lambda}^*A\nonumber
\end{equation}
between the two transformations. This quantity is in fact undesirable, and forcing it to vanish yields the following compatibility condition
\begin{equation}
    \mathfrak{ad}_{\tilde\lambda}^*[A,\lambda] = [\mathfrak{ad}_{\tilde\lambda}^*A,\lambda] + [A,\mathfrak{ad}^*_{\tilde\lambda}\lambda]  - \mathfrak{ad}_{\operatorname{ad}_\lambda^*\tilde\lambda}^*A + \mathfrak{ad}_{\operatorname{ad}_A^*\tilde\lambda}^*\lambda,\label{adcoh}
\end{equation}
where we have explicitly antisymmetrized the last term $ - \mathfrak{ad}_{\operatorname{ad}_\lambda^*\tilde\lambda}^*A$ under an interchange of $A,\lambda$ by adding a term $\mathfrak{ad}_{\operatorname{ad}_A^*\tilde\lambda}^*\lambda$ to the difference. This {\it antisymmetrization procedure} has major implications about the gauge structure of the theory --- we shall explain this in full detail in the next section. This condition Eq. \eqref{adcoh} ensures that the gauge transformation by $\boldsymbol\lambda$ on $A$ is {\it unambiguous} --- namely does not depend on the order of the action of $\lambda,\tilde\lambda$.

On the dual side, we may argue analogously by exchanging the original and the dual sectors:
\begin{equation}
    \lambda\leftrightsquigarrow\tilde\lambda,\qquad \operatorname{ad}^*\leftrightsquigarrow\mathfrak{ad}^*,\qquad A\leftrightsquigarrow B.\nonumber
\end{equation}
More explicitly, we compute from \eqref{dualgauge}, \eqref{gaugedual} that
\begin{eqnarray}
     \overrightarrow{\boldsymbol\lambda} &:& B\xrightarrow{\lambda}B + \operatorname{ad}_{\lambda}^*B \nonumber \\
     &\qquad &\xrightarrow{\tilde\lambda} B+d_B\tilde\lambda + \operatorname{ad}_{\lambda}^*B + \operatorname{ad}_{\lambda}^*d\tilde\lambda \nonumber \\
     &\qquad&\qquad + \operatorname{ad}_{\lambda}^*[B,\lambda]_*+ \operatorname{ad}_{\mathfrak{ad}_{\tilde\lambda}^*\lambda}^*B+ \operatorname{ad}_{\mathfrak{ad}_{\tilde\lambda}^*\lambda}^*d_B\tilde\lambda,\nonumber
\end{eqnarray}
while
\begin{eqnarray}
     \overleftarrow{\boldsymbol\lambda} &:& B \xrightarrow{\tilde\lambda} B + d_B\tilde\lambda\nonumber \\
     &\qquad&\xrightarrow{\lambda} B + \operatorname{ad}_{\lambda}^*B + d_b\tilde\lambda + \operatorname{ad}_{\lambda}^*d\tilde\lambda \nonumber \\
     &\qquad&\qquad + [\operatorname{ad}_{\lambda}^*B,\tilde\lambda]_* + [B,\operatorname{ad}_{\lambda}^*\tilde\lambda]_*+ [\operatorname{ad}_{\lambda}^*B,\operatorname{ad}_{\lambda}^*\tilde\lambda]_*.\nonumber
\end{eqnarray}
Neglecting the term of order $\tilde\lambda^2$, we obtain the difference
\begin{equation}
    [\operatorname{ad}_{\lambda}^*B,\tilde\lambda]_* + [B,\operatorname{ad}^*_{\lambda}\tilde\lambda]_* - \operatorname{ad}_{\lambda}^*[B,\tilde\lambda]_* - \operatorname{ad}_{\mathfrak{ad}_{\tilde\lambda}^*\lambda}^*B,\nonumber
\end{equation}
and hence the compatibility condition upon antisymmetrizing with respect to $B,\tilde\lambda$,
\begin{equation}
    \operatorname{ad}_{\lambda}^*[B,\tilde\lambda]_* = [\operatorname{ad}_{\lambda}^*B,\tilde\lambda]_* + [B,\operatorname{ad}^*_{\lambda}\tilde\lambda]_*  - \operatorname{ad}_{\mathfrak{ad}_{\tilde\lambda}^*\lambda}^*B + \operatorname{ad}_{\mathfrak{ad}_B^*\lambda}^*\tilde\lambda,\label{adcohdual}
\end{equation}
which is dual to Eq. \eqref{adcoh}.

\paragraph{The antisymmetrization procedure.}
One crucial point to note is that, to impose the compatibility conditions Eqs. \eqref{adcoh}, \eqref{adcohdual}, we must introduces the terms $\mathfrak{ad}_{\operatorname{ad}_A^*\tilde\lambda}^*\lambda,\operatorname{ad}_{\mathfrak{ad}_B^*\lambda}^*\tilde\lambda$ in order to preserve the explicit antisymmetry in $A,\lambda$ and $B,\tilde\lambda$.

These additional terms, however, cannot just be introduced from thin air. We must induce it from a modification of the gauge transformations we have assumed in Eqs. \eqref{gauge}, \eqref{gaugedual}. In particular, we include the terms $\mathfrak{ad}_B^*\lambda,\operatorname{ad}_A^*\tilde\lambda$ into how $A,B$ transform respectively under $\lambda,\tilde\lambda$:
\begin{equation}
    A\rightarrow A+d_A\lambda -\mathfrak{ad}_B^*\lambda,\qquad B\rightarrow B+d_B\tilde\lambda - \operatorname{ad}_A^*\tilde\lambda.\label{antisymgau}
\end{equation}
We call this the {\it antisymmetrization procedure} of the gauge structure, and it is a key technique that we shall use extensively in the rest of this paper.

It is simple to see why these extra terms are necessary. By construction, the combined gauge transformation $\boldsymbol\lambda$ in either order reads
\begin{eqnarray}
    \overrightarrow{\boldsymbol\lambda}: A\rightarrow \dots -\mathfrak{ad}_B^*\lambda + \mathfrak{ad}_{\operatorname{ad}_A^*\tilde\lambda}^*\lambda,&\quad& \overleftarrow{\boldsymbol\lambda}: A\rightarrow \dots -\mathfrak{ad}_B^*\lambda,\nonumber \\ 
    \overrightarrow{\boldsymbol\lambda}: B\rightarrow \dots -\operatorname{ad}_A^*\tilde\lambda,&\quad& \overleftarrow{\boldsymbol\lambda}: B\rightarrow \dots -\operatorname{ad}_A^*\tilde\lambda+ \operatorname{ad}_{\mathfrak{ad}_B^*\lambda}^*\tilde\lambda,\nonumber 
\end{eqnarray}
where "$\dots$" indicates terms that we have already seen above. Their difference then manifests the desired terms 
\begin{equation}
    \mathfrak{ad}_{\operatorname{ad}_A^*\tilde\lambda}^*\lambda + \operatorname{ad}_{\mathfrak{ad}_B^*\lambda}^*\tilde\lambda \nonumber
\end{equation}
upon performing the antisymmetrization procedure. 

However, Eq. \eqref{antisymgau} implies that the $\mathfrak{g}$- and $\mathfrak{g}^*$ connections contribute both to define the proper notion of covariant derivative: the gauge transformation for the connection $A$ depends on the dual connection $B$, and vice versa. As such, they should come together to define a {\bf combined connection}
\begin{equation}
    {\bf A} = A+B\in \Omega^1(M)\otimes(\mathfrak{g} \oplus\mathfrak{g}^*) \nonumber
\end{equation}
valued in the direct sum $\mathfrak{g}\oplus\mathfrak{g}^*$, upon which the combined gauge transformation $\boldsymbol\lambda=\lambda+\tilde\lambda$ acts. This is a key observation that shall finally allow us to extract the compatibility conditions in a Manin triple.

Note that the antisymmetrization procedure introduces a modified notion of covariant derivative for the matched pair $\g \bowtie\g^*$, 
\begin{align}
D_{\bf A}\lambda &= d_A\lambda- \mathfrak{ad}_B^*\lambda\nonumber\\
D_{\bf A}\tilde\lambda& = d_B\tilde\lambda - \operatorname{ad}_A^*\tilde\lambda. \label{covdev} 
\end{align}
which had also been discussed in \cite{Dupuis:2020ndx, Girelli:2021pol}.

\paragraph{The Manin triple.}  From the above computations, we have derived the following result:
\begin{theorem}\label{mt}
The combined gauge transformation (including antisymmetrized terms induced from Eq. \eqref{antisymgau})
\begin{equation}
    {\bf A}\rightarrow {\bf A}^{\boldsymbol\lambda} = {\bf A} + (D_{\bf A}\lambda + \mathfrak{ad}_{\tilde\lambda}^*A) + (D_{\bf A}\tilde\lambda + \operatorname{ad}_\lambda^* B)
\end{equation}
is {\it unambiguous} (namely independent of the order of applying $\lambda,\tilde\lambda$) modulo terms of order $o(\lambda^2)+o(\tilde\lambda^2) + o(\lambda d\tilde\lambda)+o(\tilde\lambda d\lambda)$ \underline{iff} the compatibility conditions Eqs. \eqref{adcoh}, \eqref{adcohdual} are satisfied.
\end{theorem}

If we distill Eqs. \eqref{adcoh}, \eqref{adcohdual} in terms of the Lie algebra values:
\begin{eqnarray}
    \mathfrak{ad}_g^*[X,X'] =[ \mathfrak{ad}_{g}^*X,X'] + [X,\mathfrak{ad}^*_{g}X']  - \mathfrak{ad}_{\operatorname{ad}_{X'}^*g}^*X + \mathfrak{ad}_{\operatorname{ad}_X^*g}^*X',\nonumber \\
    \operatorname{ad}_{X}^*[g,g']_* = [\operatorname{ad}_{X}^*g,g']_* + [g,\operatorname{ad}^*_Xg']_*  - \operatorname{ad}_{\mathfrak{ad}_{g'}^*X}^*g + \operatorname{ad}_{\mathfrak{ad}_g^*X}^*g'\label{1manin}
\end{eqnarray}
for each $X,X'\in\mathfrak{g},g,g'\in\mathfrak{g}^*$, then we in fact recover the necessary and sufficient conditions for the coadjoint representations $\operatorname{ad}^*,\mathfrak{ad}^*$ to come together and form the {\bf standard Manin triple} \cite{Chari:1994pz}
\begin{equation}
    \mathfrak{d} = \mathfrak{g}~_{\operatorname{ad}^*}\bowtie_{\mathfrak{ad}^*}\mathfrak{g}^*\nonumber
\end{equation}
equipped with the standard Lie bracket
\begin{equation}
    \pmb{[} X+g,X'+g'\pmb{]} = [X,X'] + [g,g']_* + \operatorname{ad}_X^*g' + \mathfrak{ad}_g^*X - \mathfrak{ad}_{g'}^*X' - \operatorname{ad}_{X'}^*g.\label{maninbrac}
\end{equation}
It is known that this double Lie algebra structure $\mathfrak{d}$ integrates to a {\it Poisson-Lie group} $D$ \cite{Semenov1992}, and hence serves as a model for the {\bf Drinfel'd double} $\mathfrak{d}=D(\mathfrak{g})$.

\paragraph{Manin Rigidity}
We have demonstrated that the gauge theory built on the direct sum $\mathfrak{g}\oplus\mathfrak{g}^*$ of the Lie bialgebra $\mathfrak{g}$ and its dual $\mathfrak{g}^*$ manifests canonically the structure of a Manin triple if and only if  the combined gauge transformation under $\boldsymbol\lambda=\lambda+\tilde\lambda$ is unambiguous. In particular, the naturally induced coadjoint representations $\operatorname{ad}^*,\mathfrak{ad}^*$ of the Lie bialgebra and its dual on each other must necessarily satisfy Eq. \eqref{1manin}.

Now if we consider generic algebra action $\mathfrak{g}\rhd\mathfrak{g}^*$ and back-action $\mathfrak{g}\lhd\mathfrak{g}^*$, equipped with the same non-degenerate invariant bilinear form Eq. \eqref{pair}, we can prove the following gauge-enforced rigidity result:
\begin{corollary} \label{maninrigid}
The induced combined gauge transformations $\boldsymbol\lambda=\lambda+\tilde\lambda$ is unambiguous if and only if  the double 
\begin{equation}
    \mathfrak{d}'=\mathfrak{g}\bowtie\mathfrak{g}^*\nonumber
\end{equation}
is isomorphic to the standard Manin triple $\mathfrak{d} = \mathfrak{g}~_{\operatorname{ad}^*}\bowtie_{\mathfrak{ad}^*}\mathfrak{g}$.
\end{corollary}
\begin{proof}
The same argument as in Section \ref{dd}, leading up to {\bf Theorem \ref{mt}}, goes through for the generic action/back-action $\rhd,\lhd$, such that the compatibility condition Eq. \eqref{1manin} must hold, with $\operatorname{ad}^*$ replaced by $\rhd$ and $\mathfrak{ad}^*$ replaced by $\lhd$. The induced Lie bracket must then take the same form as Eq. \eqref{maninbrac}, which forces $\mathfrak{d}' \cong \mathfrak{d}$.
\end{proof}
\noindent The rigidity of the Manin triple is well-known, but here we have demonstrated by gauge theoretic considerations. 

We shall generalize our method in the following to a {\it 2-gauge structure}, based on a Lie algebra crossed-module $\mathfrak{g}$. We shall seek to extract compatibility conditions for the (strict) coadjoint representations necessary for the notion of a {\it 2-Manin triple}. 

\subsection{3D topological action invariant under the Drinfel'd double}\label{combocov}
The fact that the Manin triple $\mathfrak{d}$ forms a Lie algebra implies that we can do gauge theory on it based on the standard bracket $[\cdot,\cdot]_\mathfrak{d} = \pmb{[}\cdot,\cdot\pmb{]}$ in Eq. \eqref{maninbrac}. In particular, given the combined connection ${\bf A}$ we can write a covariant curvature quantity 
\begin{eqnarray}
    {\bf F} &=& d{\bf A} + \frac{1}{2}\pmb{[} {\bf A}\wedge {\bf A}\pmb{]} \nonumber \\
    &=&(dA + \frac{1}{2}[A\wedge A] + \mathfrak{ad}_B^*(\wedge A)) + (dB + \frac{1}{2}[B\wedge B]_* + \operatorname{ad}_A^*(\wedge B))\nonumber \\
    &\equiv& \bar{F} + \bar{F}^*,\label{combocurv}
\end{eqnarray}
called the {\bf combined curvature}, which has also appeared in Ref. \cite{Girelli:2021pol}. The covariance of ${\bf F}$ is manifest; for an explicit proof, see Appendix  \ref{gencov}.

\paragraph{Chern-Simons theory and interacting BF theory.} Given the combined curvature ${\bf F}$, one may then assemble a gauge-invariant action and study its dynamics. Using the natural invariant non-degenerate symmetric bilinear form Eq. \eqref{pair} on $\mathfrak{d}$, we begin with the {\it Chern-Simons action}
\begin{equation}
    S_\text{CS}[{\bf A}] = \frac{1}{2}\int_M \langle\langle{\bf F}\wedge {\bf F}\rangle\rangle =
     \int_M \langle \bar F^*\wedge \bar F\rangle,\label{chernsimons}
\end{equation}
where $M$ is a 4-dimensional manifold. It can be clearly seen that, due to the covariance of the combined curvature ${\bf F}$, the action Eq. \eqref{chernsimons} is invariant under the Drinfel'd double $\mathfrak{d}=\mathfrak{g}\bowtie\mathfrak{g}^*$ gauge transformations.

If we write down the components of $\bar F,\bar F^*$ explicitly, then one may compute 
\begin{eqnarray}
    \langle \bar F^* \wedge \bar F\rangle &=& \frac{1}{2}d\left[\langle B \wedge (dA + \frac{1}{2}[A\wedge A])\rangle + \frac{1}{2}\langle [B\wedge B]_* \wedge A\rangle\right] \nonumber \\
    &\qquad& + \frac{1}{4}\langle [B\wedge B]_* \wedge [A\wedge A]\rangle + \langle \operatorname{ad}_A^*(\wedge B) \wedge \mathfrak{ad}_B^*(\wedge A)\rangle,\nonumber
\end{eqnarray}
where the equality stands up to a boundary term. 
By writing the final term as 
\begin{equation}
    \langle \operatorname{ad}_A^*(\wedge B) \wedge \mathfrak{ad}_B^*(\wedge A)\rangle = -\frac{1}{2} (\langle B\wedge [A\wedge \mathfrak{ad}_B^*(\wedge A)]\rangle + \langle [B\wedge \operatorname{ad}_A^*(\wedge B)]_*\wedge A\rangle),\nonumber
\end{equation}
we can make use of Eq. \eqref{1manin} to have
\begin{equation}
    [A\wedge \mathfrak{ad}_B^*(\wedge A)] = -\frac{1}{4}\mathfrak{ad}_B^*(\wedge [A\wedge A]),\qquad 
    [B\wedge \operatorname{ad}_A^*(\wedge B)]_* =-\frac{1}{4}\operatorname{ad}_A^*(\wedge [B\wedge B]_*) ,\nonumber
\end{equation}
which yields
\begin{eqnarray}
    \langle \operatorname{ad}_A^*(\wedge B) \wedge \mathfrak{ad}_B^*(\wedge A)\rangle &=& -\frac{1}{8} (\langle \mathfrak{ad}_B (\wedge B)\wedge [A\wedge A]\rangle+ \langle [B\wedge B]_* \wedge \operatorname{ad}_A^*(\wedge A)\rangle) \nonumber\\
    &=& -\frac{1}{4}\langle [B\wedge B]_*\wedge [A\wedge A]\rangle.\nonumber
\end{eqnarray}

This allows us to write the Chern-Simons action as
\begin{equation}
    S_\text{CS}[{\bf A}] = \frac{1}{2}\int_Z \langle B \wedge (dA + \frac{1}{2}[A\wedge A])\rangle + \frac{1}{2}\langle [B\wedge B]_* \wedge A\rangle,\nonumber
\end{equation}
where $Z= \partial M$ is the 3-dimensional boundary of $M$. If we recall the ordinary curvature 2-form
\begin{equation}
    F = dA + \frac{1}{2}[A\wedge A], \nonumber
\end{equation}
then we can rewrite
\begin{equation}
    S_\text{CS}[{\bf A}] = \frac{1}{2}\int_Z \langle B\wedge F\rangle - \frac{1}{2}\langle B\wedge \mathfrak{ad}_B^*(\wedge A)\rangle = \frac{1}{2}\int_Z \langle B\wedge F\rangle + \langle \frac{1}{2}[B\wedge B]_*\wedge A\rangle,\nonumber
\end{equation}
which is in fact a rewritting (up to a boundary term) of the standard BF theory with a volume term \cite{Dupuis:2020ndx}.

If we vary $B$, then we achieve the flatness condition 
\begin{equation}
    \delta_AS_\text{CS}=0\implies \bar F = dA +\frac{1}{2}[A\wedge A] - \mathfrak{ad}_B^*(\wedge A) =0\nonumber
\end{equation}
for the $\mathfrak{g}$-sector of the combined curvature Eq. \eqref{combocurv}. On the other hand, if we vary $A$, then we obtain the flatness condition 
\begin{equation}
    \delta_BS_\text{CS}=0\implies \bar F^* = dB + \frac{1}{2}[B\wedge B]_* - \operatorname{ad}_A^*(\wedge B) = 0\nonumber
\end{equation}
for the $\mathfrak{g}^*$-sector of the combined curvature. 

\paragraph{Pairings on the Drinfel'd double; non-standard Chern-Simons theory.}
The form of the Chern-Simons action Eq. \eqref{chernsimons} depends on the choice of the bilinear form $\langle\langle\cdot,\cdot\rangle\rangle$ on $\mathfrak{d}$. From the natural choice Eq. \eqref{pair} given by the evaluation pairing $(\cdot,\cdot):\mathfrak{g}^*\otimes\mathfrak{g}\rightarrow \mathbb{R}$, the above result yields an interacting BF theory. However, more general choices are possible. 

Given a bialgebra $(\mathfrak{g},\psi)$, we can form the structure of a Manin triple on the direct sum $\mathfrak{d}\cong \mathfrak{g}\oplus\mathfrak{g}^*$ by endowing it with a symmetric non-degenerate invariant bilinear form $\langle\langle\cdot,\cdot\rangle\rangle: \mathfrak{d}^{2\otimes}\rightarrow\mathbb{R}$, with respect to which $\mathfrak{d}^*\cong \mathfrak{d}$ is self-dual as bialgebras. It is clear that the natural evaluation pairing Eq. \eqref{pair} is merely one component of $\langle\langle\cdot,\cdot\rangle\rangle$. The other component is given by the alternative pairing \cite{Osei:2017ybk}
\begin{equation}
    \langle\langle X+g,X'+g'\rangle\rangle' = \langle X,X'\rangle+\langle g,g'\rangle \nonumber
\end{equation}
that is "diagonal" in the direct sum $\mathfrak{g}\oplus\mathfrak{g}^*$. It is invariant in the sense that 
\begin{eqnarray}
    \langle \operatorname{ad}_XX',X''\rangle + \langle X',\operatorname{ad}_XX''\rangle=0,&\quad& \langle \mathfrak{ad}_gg',g''\rangle + \langle g',\mathfrak{ad}_gg''\rangle =0\nonumber\\
    \langle \mathfrak{ad}_g^*X',X''\rangle + \langle X',\mathfrak{ad}_g^*X''\rangle=0,&\quad& \langle \operatorname{ad}_X^*g',g''\rangle + \langle g',\operatorname{ad}_X^*g''\rangle =0\nonumber
\end{eqnarray}
for each $X,X',X''\in\mathfrak{g}$ and $g,g',g''\in\mathfrak{g}^*$. Such a pairing $\langle\cdot,\cdot\rangle$ exists if, for instance, $\mathfrak{g}$ is semisimple, whence $\langle\cdot,\cdot\rangle$ is the Killing form. 


The Drinfel'd double $\mathfrak{d}=\mathfrak{g}\bowtie\mathfrak{g}^*$ has the same structure Eq. \eqref{maninbrac} as previously derived when equipped with this non-standard bilinear form  $\langle\langle\cdot,\cdot\rangle\rangle'$, but the Chern-Simons action Eq. \eqref{chernsimons} now reads
\begin{equation}
    S_\text{CS}'[{\bf A}] = \frac{1}{2}\int_M \langle\langle{\bf F}\wedge{\bf F}\rangle\rangle'= \frac{1}{2}\int_M \langle \bar{F}^*\wedge \bar{F}^*\rangle + \langle \bar{F}\wedge\bar{F}\rangle.\nonumber
\end{equation}
Standard results states that the Chern-Simons polynomial $\langle\langle{\bf F}\wedge{\bf F}\rangle\rangle'$ is exact, hence $S'_\text{CS}[{\bf A}]$ reduces to the boundary action
\begin{equation}
    S_\text{CS}'[{\bf A}] = \frac{1}{2}\int_Z \langle\langle {\bf A}\wedge d{\bf A}\rangle\rangle' + \frac{1}{3}\langle\langle {\bf A}\wedge\pmb{[}{\bf A}\wedge{\bf A}\pmb{]}\rangle\rangle'.\nonumber 
\end{equation}

From Eq. \eqref{combocurv}, we have
\begin{equation}
    \frac{1}{2}\pmb{[}{\bf A}\wedge{\bf A}\pmb{]} = (\frac{1}{2}[A\wedge A]+\mathfrak{ad}_B^*(\wedge A)) + (\frac{1}{2}[B\wedge B]_*+\operatorname{ad}_A^*(\wedge B)),\nonumber
\end{equation}
hence we see that
\begin{eqnarray}
     S_\text{CS}'[{\bf A}]&=& \frac{1}{2}\int_Z \langle A\wedge( dA+\frac{1}{3}[A\wedge A])\rangle + \langle B\wedge (dB+ \frac{1}{3}[B\wedge B]_*)\rangle \nonumber \\
     &\qquad& + \int_Z \langle A\wedge \mathfrak{ad}_B^*(\wedge A)\rangle + \langle B\wedge \operatorname{ad}_A^*(\wedge B)\rangle\nonumber\\
     &\equiv& \frac{1}{2}(S_\text{CS}[A]+S_\text{CS}[B])  + \int_Z \langle A\wedge \mathfrak{ad}_B^*(\wedge A)\rangle + \langle B\wedge \operatorname{ad}_A^*(\wedge B)\rangle,\nonumber
\end{eqnarray}
where we have defined the usual Chern-Simons action $S_\text{CS}[A],S_\text{CS}[B]$ in each of the individual $\mathfrak{g},\mathfrak{g}^*$-sectors. The rest are interaction terms. 


One may consider the Chern-Simons theory $S_\text{CS}''$ given by a linear combination of the alternative pairing with the canonical one
\begin{equation}
    \langle\langle\cdot,\cdot\rangle\rangle'' = \alpha\langle\langle \cdot,\cdot\rangle\rangle + \beta\langle\langle\cdot\,\cdot\rangle\rangle',\nonumber
\end{equation}
such that the non-degeneracy requires the real parameters $\alpha,\beta$ to satisfy $\alpha^2+ \beta^2 \neq 0$. We would then have
\begin{equation}
    S_\text{CS}''[{\bf A}] = \alpha S_\text{CS}[{\bf A}]+ \beta S_\text{CS}'[{\bf A}].\nonumber
\end{equation}
Such a theory had been considered in detail in \cite{Osei:2017ybk}.

\section{Lie algebra crossed-module and its dual}\label{algxmod}
In this section, we construct a pair of dual Lie algebra crossed-modules, just like we defined a pair of dual Lie algebras. We also discuss the motion of adjoint and co-adjoint representations in this set up.

\subsection{Lie algebra crossed-module.} 
We first review the notion of a Lie algebra crossed-module. Let $\mathfrak{g}_0,\mathfrak{g}_{-1}$ be a pair of Lie algebras, a {\bf Lie algebra crossed-module} is given in terms of a map 
\begin{equation}
    t: \mathfrak{g}_{-1}\rightarrow\mathfrak{g}_0,\nonumber
\end{equation}
called the "$t$-map" or the {\it differential}, together with an action $\rhd$ of $\mathfrak{g}_0$ on $\mathfrak{g}_{-1}$ such that the {\bf Pfeiffer identities}
\begin{equation}
    t(X\rhd Y) = [X,tY],\qquad tY\rhd Y' = [Y,Y'],\qquad \forall~X\in\mathfrak{g}_0,~Y,Y'\in\mathfrak{g}_{-1} \label{pfeif}
\end{equation}
are satisfied, in addition to the {\bf 2-Jacobi identities}
\begin{eqnarray}
[X,[X',X'']]+[X',[X'',X]]+[X'',[X,X']]=0,\nonumber \\
X\rhd (X' \rhd Y) - X' \rhd (X\rhd Y) - [X,X']\rhd Y = 0\label{2jacob}
\end{eqnarray} 
for all $X,X',X''\in\mathfrak{g}_0$ and $Y\in\mathfrak{g}_{-1}$. Notice here that, by the second Pfeiffer identity Eq. \eqref{pfeif}, the Lie algebra bracket $[\cdot,\cdot]$ on $\mathfrak{g}_{-1}$ is completely determined by the structures $t,\rhd$. A Lie algebra crossed-module is also called a \textit{strict Lie 2-algebra}. 


\begin{definition}
We call a Lie algebra crossed-module {\bf trivial} if $t=\operatorname{id}$, and {\bf skeletal} if $t=0$  \cite{Kim:2019owc,Chen:2013,Chen:2012gz,Bai_2013}. 
\end{definition}


We shall denote a Lie algebra crossed-module by $\mathfrak{g} = (\mathfrak{g}_{-1}\xrightarrow{t}\mathfrak{g}_0,\rhd,[\cdot,\cdot])$, where $[\cdot,\cdot]$ is the Lie bracket on $\mathfrak{g}_0$. This is indeed the minimal information one needs since, the Lie bracket of $\g_{-1}$ is determined from the second Pfeiffer identity in terms of the action and $t$-map.  The subscripts on the Lie algebras are called the {\it degrees of the grading} in $\mathfrak{g}$. It is a very useful tool for keeping track of the crossed-module structures.

Notice that there is a natural adjoint representation of $\mathfrak{g}$ on itself. Since we have two pieces, $\g_0$ and $\g_{-1}$, we can split this adjoint representation into two components, one coming from $\g_0$ and one from $\g_{-1}$. 
Indeed, we have the bracket of $\g_0$ on itself but also the action $\rhd$ of $\g_0$ on $\g_{-1}$. The adjoint representation of  the $\g_{-1}$ sector is also specified by the crossed-module action $\rhd$. 
This is denoted by
\begin{equation}
    \operatorname{ad}=(\operatorname{ad}_0,\operatorname{ad}_{-1}):\mathfrak{g}\rightarrow \operatorname{End}\mathfrak{g},\label{strictad}
\end{equation}
which comes in graded components $\operatorname{ad}_0,\operatorname{ad}_{-1}$ defined by
\begin{equation}
    \operatorname{ad}_0(X) = (\operatorname{ad}_X\equiv [X,\cdot]\, ,\, \chi_X\equiv X\rhd\cdot) \in \operatorname{End}(\mathfrak{g}_0\oplus\mathfrak{g}_{-1}),\qquad \operatorname{ad}_{-1}(Y) = \cdot \rhd Y \in \operatorname{Hom}(\mathfrak{g}_0,\mathfrak{g}_{-1})\nonumber 
\end{equation}
for each $X\in\mathfrak{g}_0,Y\in\mathfrak{g}_{-1}$. 

\begin{example}\label{exp1}
In general, there are two ways of associating a  Lie algebra $\mathfrak{l}$ to a Lie algebra crossed-module.
\begin{enumerate}
    \item {\bf The canonical embedding: trivial crossed-module}: for any 1-algebra $\mathfrak{l}$, we define the crossed-module
    \begin{equation}
        \operatorname{id}_\mathfrak{l}=(\mathfrak{g}_{-1}\xrightarrow{\operatorname{id}}\mathfrak{g}_0,\rhd,[\cdot,\cdot]),\qquad \mathfrak{g}_{-1}=\mathfrak{g}_0=\mathfrak{l} \nonumber
    \end{equation}
    whose $t$-map is the identity. The Pfeiffer identities Eq. \eqref{pfeif} force the action $\rhd=[\cdot,\cdot]$ to be the adjoint action. We call the association $\mathfrak{l}\mapsto\operatorname{id}_\mathfrak{l}$ the {\bf canonical embedding}, and in fact embeds the category of Lie algebras into the category of Lie algebra crossed-modules \cite{Bai_2013,Wag}.
    \item {\bf The suspension embedding: skeletal crossed-module}: given the semidirect product $V\rtimes\mathfrak{u}$, where $\mathfrak{u}$ is a Lie algebra acting on the (abelian) Lie algebra $V$, we can  associate a skeletal crossed-module structure to $V\rtimes\mathfrak{u}$, 
    \begin{equation}
        \mathfrak{g}^0=(\mathfrak{g}_{-1}\xrightarrow{0}\mathfrak{g}_0,\rhd,[\cdot,\cdot]),\qquad \mathfrak{g}_{-1}=V,~\mathfrak{g}_0=\mathfrak{u} \nonumber
    \end{equation}
    to it. The Pfeiffer identities Eq. \eqref{pfeif} enforce $\mathfrak{g}_{-1}=V$ to be Abelian (ie. a vector space), but otherwise imposes no constraints on the action $\rhd$. Provided $V$ is Abelian, the crossed-module structure of $\mathfrak{g}^0$ coincides with the Lie bracket of $V\rtimes\mathfrak{u}$
    \begin{equation}
        [X+Y,X'+Y'] = [X,X] + X\rhd Y' - X'\rhd Y,\qquad X,X'\in\mathfrak{g},~Y,Y'\in V.\label{seminull}
    \end{equation}
    Aside from the grading (ie. degree information) inherent in the crossed-module $\mathfrak{g}^0$, there is no distinction between the algebraic structures of $\mathfrak{g}^0$ and $V\rtimes\mathfrak{u}$. We call the association $V\rtimes\mathfrak{u}\mapsto \mathfrak{g}^0$ the {\bf suspension embedding}. 
\end{enumerate}
\end{example}

\subsection{Dual crossed-module and Lie bialgebra crossed-modules.} In the following, we shall also need the notion of a {\it dual Lie algebra crossed-module}. As we mentioned already, we only need to dualize the minimal data encoding the crossed-module structure.

A crossed-module is defined in terms of the Lie algebra $\g_0$, the space $\g_1$, the action of $\g_0$ on $\g_{-1}$ and the $t$-map. The Lie algebra structure is inherited from the $t$-map and the action. Canonically, the duality relation would be between opposite grading. This means that if $\g_0$ (resp. $\g_{-1}$) has degree 0 (resp. $-1$), then $\g_{0}^*$ will have degree $-1$ (resp. $0$).

Hence the dual crossed-module would be characterized by the Lie algebra $\g^*_{-1}$ of degree 0, the (vector) space $\g_{0}^*$ which would become a Lie algebra if we specify the action $\rhd^*$ and the dual $t$-map. This means that to define the dual crossed-module, we need to define a cocycle on $\g_{-1}$ to define the degree 0 Lie algebra $\g^*_{-1}$. We need a cocycle  which would give the action $\rhd^*$. Finally, dualizing the $t$-map will provide the right $t$-map between $\g_0^*$ and $\g_{-1}$. The Pfeiffer identities for the dual structures are then translated in terms of the cocycle properties.

\medskip 

More explicitly, the dual structures are induced by a pair of Lie algebra cochains
\begin{equation}
    \delta_{-1}: \mathfrak{g}_{-1}\rightarrow \mathfrak{g}_{-1}\wedge \mathfrak{g}_{-1},\qquad \delta_{0}: \mathfrak{g}_0\rightarrow (\mathfrak{g}_{-1}\otimes\mathfrak{g}_0)\oplus(\mathfrak{g}_0\otimes \mathfrak{g}_{-1}) \label{dualcocy} 
\end{equation}
such that we induce the crossed-module structures $[\cdot,\cdot]_*,\rhd^*$ in the following way
\begin{equation}
    [f,f']_*(Y) = (f\wedge f')(\delta_{-1}Y),\qquad (f\rhd ^* g)(X) = (f\wedge g)(\delta_0 X)\nonumber
\end{equation}
for each $X\in\mathfrak{g}_0,Y\in\mathfrak{g}_{-1}$ and $g\in\mathfrak{g}_0^*,f,f'\in\mathfrak{g}_{-1}^*$ in the dual. Here, duality is once again taken with respect to the canonical evaluation pairing
\begin{equation}
    (f+g,X+Y) = g(X) + f(Y),\nonumber
\end{equation}
which induces the natural invariant non-degenerate symmetric bilinear form \cite{Bai_2013}
\begin{equation}
    \langle\langle (X+f) + (Y+g), (X'+f') + (Y'+g')\rangle\rangle = g(X') + g'(X) + f(Y') + f'(Y)\label{2bilin}
\end{equation}
on the direct sum $(\mathfrak{g}_0\oplus\mathfrak{g}_{-1}^*)\oplus(\mathfrak{g}_{-1}\oplus\mathfrak{g}_0^*)$ as a vector space, where we collected together the terms of the same degree.

We now introduce a map $\tilde t: \mathfrak{g}_0^*\rightarrow\mathfrak{g}_{-1}^*$ such that the dual analogue of the Pfeiffer identities Eq. \eqref{pfeif}
\begin{equation}
    \tilde t(f\rhd^* g) =[f,\tilde t g]_*,\qquad \tilde t g\rhd^* g' = [g,g']_*,\qquad \forall~f\in\mathfrak{g}_{-1}^*,~g,g'\in\mathfrak{g}_0^* \nonumber
\end{equation}
are satisfied. When written in terms of the 2-cochains $(\delta_{-1},\delta_0)$, they are equivalent to\footnote{To see this, we for instance evaluate $\delta_0\tilde t^*$ to yield $(f\wedge g)(\delta_0\tilde t^* Y) = (f\rhd^* g)(\tilde t^*(Y)) = (\tilde t(f\rhd^* g))(Y)$, while $(f\wedge g)((\tilde t^*\otimes 1 + 1\otimes \tilde t^*)\delta_{-1}(Y)) = (f\wedge \tilde tg)(\delta_{-1}Y) = [f,\tilde t g]_*(Y)$.}
\begin{equation}
    \delta_0 \tilde t^* = (\tilde t^*\otimes 1 + 1\otimes \tilde t^*)\delta_{-1},\qquad (\tilde t^*\otimes 1 - 1\otimes \tilde t^*) \delta_0 = 0,\label{dualcoh1}
\end{equation}
where $\tilde t^*: \mathfrak{g}_{-1}\rightarrow\mathfrak{g}_0$ is the dual map of $\tilde t$. Notice that, once again, the dual bracket $[\cdot,\cdot]_*$ on $\mathfrak{g}_0^*$ is determined completely by the structures $\tilde t,\rhd^*$ from the second Pfeiffer identity.

Now in addition to the Pfeiffer identities, we must enforce the analogues of the 2-Jacobi identities Eq. \eqref{2jacob}. These read
\begin{eqnarray}
[f,[f',f'']_*]_*+[f',[f'',f]_*]_*+[f'',[f,f']_*]_*=0,\nonumber \\
f\rhd^* (f' \rhd^* g) - f' \rhd^* (f\rhd^* g) - [f,f']_*\rhd^* g = 0,\label{dual2jacob}
\end{eqnarray}
where $f,f',f''\in\mathfrak{g}_{-1}^*,g\in\mathfrak{g}_0^*$, which is equivalent to the equivariance of the structures $\rhd^*,[\cdot,\cdot]_*$ under the coadjoint representation\footnote{Notice the action $\rhd^*$ in the dual crossed-module $\mathfrak{g}^*[1]$ is distinct from the coaction $\chi^*$ of $\mathfrak{g}_0$ on $\mathfrak{g}^*_{-1}$.} $\operatorname{ad}^*=(\operatorname{ad}^*_0,\operatorname{ad}_{-1}^*)$ obtained by dualizing Eq. \eqref{strictad}:
\begin{eqnarray}
    \operatorname{ad}_X^*(f\rhd^* g) &=& (\chi_X^*f)\rhd^*g - f\rhd^*(\operatorname{ad}_X^*g),\nonumber \\ \chi_X^*[f,f']* &=& [\chi_X^*f,f'] + [f,\chi_X^*f'],\nonumber \\
    \operatorname{ad}_{-1}^*(Y)[f,f'] &=&  f\rhd^* \operatorname{ad}_{-1}^*(Y)(f') - f'\rhd^* \operatorname{ad}_{-1}^*(Y)(f).\nonumber
\end{eqnarray}
If written in terms of the 2-cochain $(\delta_{-1},\delta_0)$, we obtain explicitly the compatibility relations
\begin{eqnarray}
    \delta_0([X,X']) &=& (X\rhd \otimes 1 + 1\otimes \operatorname{ad}_X)\delta_0(X') - (X'\rhd \otimes 1 + 1\otimes\operatorname{ad}_{X'})\delta_0(X),\nonumber \\
    \delta_{-1}(X\rhd Y) &=& (X\rhd \otimes 1 + 1\otimes X\rhd)\delta_{-1}(Y) + \delta_0(X)(\rhd Y\otimes 1 + 1\otimes \rhd Y).\label{dualcoh2}
\end{eqnarray}
If Eqs. \eqref{dualcoh1} and \eqref{dualcoh2} are satisfied, then $\mathfrak{g}^*[1]=(\mathfrak{g}_0^*\xrightarrow{\tilde t}\mathfrak{g}_{-1}^*,[\cdot,\cdot]_*,\rhd^*)$ is a bona fide crossed-module\footnote{In keeping with the mathematical literature, we have kept the notation "$[1]$", which indicates a degree-shift in the grading.}. We call such a 2-cochain $(\delta_{-1},\delta_0)$ a {\bf 2-cocycle}, and $(\mathfrak{g};\delta_{-1},\delta_0)$ a {\bf Lie bialgebra crossed-module}.

\begin{remark}\label{2bialgsym}
Similar to bialgebras, the bialgebra crossed-module structure is symmetric under dualization: if $\mathfrak{g}$ is a Lie bialgebra crossed-module, then so is its dual $\mathfrak{g}^*[1]$, which is equipped with the 2-cocycle
\begin{equation}
    \delta_{-1}^*:\mathfrak{g}_0^*\rightarrow\mathfrak{g}_0^*\wedge\mathfrak{g}_0^*,\qquad \delta_0^*:\mathfrak{g}_{-1}^*\rightarrow (\mathfrak{g}_{-1}^*\otimes\mathfrak{g}_0^*) \oplus(\mathfrak{g}_0^*\otimes\mathfrak{g}_{-1}^*) \nonumber
\end{equation}
dual to the crossed-module structure $([\cdot,\cdot],\rhd)$ of $\mathfrak{g}$, satisfying analogues of Eqs. \eqref{dualcoh1} and \eqref{dualcoh2}. However, this means that $\tilde t^*$ must be related to the original crossed-module map $t$; in fact, we {\it must} have $\tilde t = t^*$ in a bialgebra crossed-module \cite{Chen:2013,Bai_2013}. The situation is summarized below:
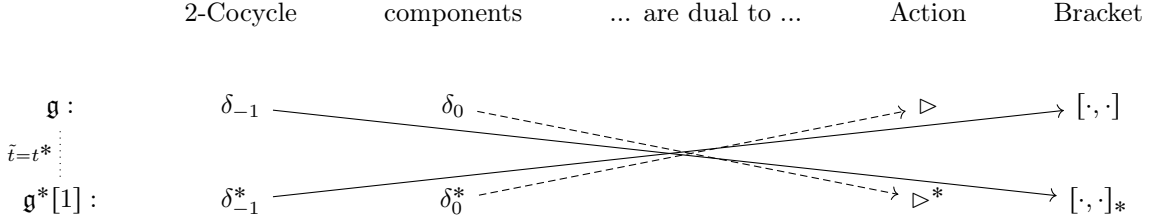
\begin{equation}
    \begin{tikzcd}
                                                                   & \text{2-Cocycle}            & \text{components}              & \text{... are dual to ...} & \text{Action} & \text{Bracket}    \\
\mathfrak{g}:                                                      & \delta_{-1} \arrow[rrrrd]   & \delta_0 \arrow[rrd, dashed]   &                          & \rhd          & {[\cdot,\cdot]}   \\
{\mathfrak{g}^*[1]:} \arrow[u, "\tilde t = t^*", no head, dotted] & \delta_{-1}^* \arrow[rrrru] & \delta_0^* \arrow[rru, dashed] &                          & \rhd^*        & {[\cdot,\cdot]_*}
\end{tikzcd}\nonumber
\end{equation}
\end{remark}

\begin{example}\label{exp2}
Consider the examples given in {\bf Example \ref{exp1}}.
\begin{enumerate}
    \item {\bf The canonical bicrossed-module:} It is clear that any Lie 2-bialgebra structure on $\operatorname{id}_\mathfrak{l}$ collapses to just the data on the Lie algebra $\mathfrak{l}$: indeed, Eq. \eqref{dualcoh1} implies that the 2-cocycle $(\delta_{-1},\delta_0)$ is merely two copies of the usual Lie algebra 1-cocycle $\psi\in Z^1(\mathfrak{l},\mathfrak{l}^{2\wedge})$.
    
%
Due to the relation $\tilde t=t^*=\operatorname{id}$, the dual 2-algebra $\operatorname{id}_\mathfrak{l}^*[1]$ is also trivial, whose crossed-module structure is determined by the 2-cocycle $(\delta_{-1},\delta_0)=(\psi,\psi)$. Moreover, the coherence condition Eq. \eqref{dualcoh1} implies that the dual bracket $[\cdot,\cdot]_*$ induced by $\psi$ constitutes the crossed-module structures in $\operatorname{id}_{\mathfrak{l}}^*[1]$.

As 2-bialgebras, we thus have $\operatorname{id}_\mathfrak{l}^*[1]=\operatorname{id}_{\mathfrak{l}^*}$; namely duality commutes with the canonical embedding. We call $(\operatorname{id}_\mathfrak{l};\psi,\psi)$ the {\bf canonical bialgebra crossed-module} associated to the bialgebra $(\mathfrak{l};\psi)$.

\item {\bf The suspension bialgebra crossed-module:} Given the semidirect product $V\rtimes\mathfrak{u}$, taking the dual in general "flips" the action $(V\rtimes\mathfrak{u})^* = \mathfrak{u}^*\rtimes V^*$. Take a 1-cocycle $\psi\in Z^1(V\rtimes\mathfrak{u},(V\rtimes\mathfrak{u})^{2\wedge})$ on the semidirect product such that
\begin{eqnarray}
    ((f+g)\wedge (f'+g'))(\psi(X+Y))&=& [f+g,f'+g']_*(X+Y)\nonumber \\
    &=&[f,f']_*(Y)+ [g,g']_*(X) + f\rhd^*g'(X) - f'\rhd^* g(X) \nonumber
\end{eqnarray}
for each $f,f'\in V^*,g,g'\in\mathfrak{u}^*$ and $X\in\mathfrak{u},Y\in V$, we see that components of the 2-cochain $(\delta_{-1},\delta_0)$ in fact make an appearance:
\begin{equation}
   [f,f'](Y) = (f\wedge f')( \delta_{-1}Y),\qquad (f\rhd^* g)(X) = (f\wedge g)(\delta_0X).\nonumber
\end{equation}
However, the term involving $[g,g']_*$ is missing, which suggests that $\mathfrak{u}^*$ {\it must} be Abelian. This is consistent with the fact that the $t$-map $\tilde t = t^*=0$ on the dual crossed-module $(\mathfrak{g}^0)^*[1]$ is also trivial. When this is the case, Eq. \eqref{dualcoh1} becomes vacuous while Eq. \eqref{dualcoh2} is equivalent to the cocycle condition for $\psi$. As a consequence, we have the following proposition.

\begin{proposition}\label{suscocy}
If the 1-cocycle $\psi|_\mathfrak{u}$ restricted to $\mathfrak{u}\subset V\rtimes\mathfrak{u}$ has trivial image in $\mathfrak{u}^{2\wedge}$ (hence $\mathfrak{u}^*$ is Abelian), then $\psi=\delta_{-1}+\delta_0$ determines a 2-cocycle $(\delta_{-1},\delta_0)$ for its skeletal suspension $\mathfrak{g}^0$.  
\end{proposition}
\noindent This implies that $(\mathfrak{g}^*)^0 = (\mathfrak{g}^0)^*[1]$; namely the duality commutes with the suspension embedding, provided the hypothesis of the proposition is true. We call $(\mathfrak{g}^0;\delta_{-1},\delta_0)$ the {\bf suspension bialgebra crossed-module} associated to the bialgebra $(V\rtimes\mathfrak{u};\psi)$.
\end{enumerate}
\end{example}

Similar to the original crossed-module, there is also a natural adjoint representation of $\mathfrak{g}^*[1]$ on itself. This is denoted 
\begin{equation}
    \mathfrak{ad}=(\mathfrak{ad}_0,\mathfrak{ad}_{-1}):\mathfrak{g}^*[1]\rightarrow\operatorname{End}\mathfrak{g}^*[1].\label{dualstrictad}
\end{equation}
It also comes in graded components given by 
\begin{equation}
    \mathfrak{ad}_0(f) = ([f,\cdot]_*\equiv \mathfrak{ad}_f\, , \, f\rhd\cdot \equiv \eta_f) \in \operatorname{End}(\mathfrak{g}_{-1}^*\oplus\mathfrak{g}_0^*),\qquad \mathfrak{ad}_{-1}(g) = \cdot \rhd^* g \in \operatorname{Hom}(\mathfrak{g}_{-1}^*,\mathfrak{g}_0^*),\nonumber 
\end{equation}
where $f\in\mathfrak{g}_{-1}^*,g\in\mathfrak{g}_0^*$. 

\subsection{Coadjoint representations}\label{seccoadj}
By dualizing the adjoint representations Eqs. \eqref{strictad} and \eqref{dualstrictad}, we obtain the {\bf coadjoint representations}
\begin{eqnarray}
    (\operatorname{ad}^*_0,\operatorname{ad}^*_{-1}):\mathfrak{g}\rightarrow \operatorname{End}\mathfrak{g}^*[1],&\quad&
    (\mathfrak{ad}_0^*,\mathfrak{ad}_{-1}^*):\mathfrak{g}^*[1]\rightarrow\operatorname{End}\mathfrak{g},\nonumber\\
    \operatorname{ad}_0^*=(\operatorname{ad}^*,\chi^*):\g_0 \rightarrow \operatorname{End}(\g_0^*\oplus\g_{-1}^*),&\quad& \mathfrak{ad}_0^*=(\mathfrak{ad}^*,\eta^*): \g_{-1}^*\rightarrow \operatorname{End}(\g_{-1}\oplus\g_0),\nonumber\\
    \operatorname{ad}_{-1}^*\equiv\Delta:\mathfrak{g}_{-1}\rightarrow\operatorname{Hom}(\mathfrak{g}_{-1}^*,\mathfrak{g}_0^*),&\quad& \mathfrak{ad}_{-1}^*\equiv \tilde\Delta:\mathfrak{g}_0^*\rightarrow\operatorname{Hom}(\mathfrak{g}_0,\mathfrak{g}_{-1}).\label{rep}
\end{eqnarray}
Explicitly for each $X,X'\in\mathfrak{g}_0,Y\in\mathfrak{g}_{-1}$ and $g,g'\in\mathfrak{g}_0^*,f\in\mathfrak{g}_{-1}^*$, they are defined in graded components by
\begin{eqnarray}
    (\operatorname{ad}^*_Xg)(X') = -g([X,X']),&\quad& f'(\mathfrak{ad}^*_fY) = -[f,f']_*(Y),\nonumber \\
    (\chi_X^*f)(Y) = -f(X\rhd Y), &\quad& g(\eta_f^*X) =- (f\rhd^* g)(X),\nonumber \\
    (\Delta_Y(f))(X) = - f(X\rhd Y),&\quad& f(\tilde\Delta_g(X)) =- (f\rhd^* g)(X).\nonumber
\end{eqnarray}
It can then be seen that the canonical evaluation pairing $(f+g,X+Y) = g(X) + f(Y)$ --- and hence its induced pairing Eq. \eqref{2bilin} --- is {\it by definition} invariant under the coadjoint representations Eq. \eqref{rep}. The first Pfeiffer identities $t\chi = \operatorname{ad}t, \tilde t\eta=\mathfrak{ad}\tilde t$ then lead to
\begin{eqnarray}
    \chi^*_Xt^*=t^*\operatorname{ad}^*_X,&\quad& \eta^*_f\tilde t^* = \tilde t^*\mathfrak{ad}^*_f,\label{comm}\\
    \Delta_Y\circ \tilde t = \operatorname{ad}_{\tilde t^*Y}^*,&\quad& \tilde\Delta_g\circ t = \mathfrak{ad}_{t^* g}^*.\label{abhom} 
\end{eqnarray}
If $\operatorname{ad}^*,\mathfrak{ad}^*$ satisfy Eqs. \eqref{comm}, \eqref{abhom}, then Eq. \eqref{rep} define {\bf strict coadjoint representations} of $\mathfrak{g}$ and $\mathfrak{g}^*[1]$ on each other. These shall play a very important role in the following.

\section{Recovering the 2-Drinfel'd double as a  2-gauge symmetry from dual 2-gauge theories}\label{2bfdd}
Similarly to the 1-gauge case, we are going to define a pair of 2-gauge theories, based on dual crossed-modules. Because they are dual, we will be able to define an action of the dual 2-gauge on the 2-gauge data and vice-versa. Then demanding that the order does not matter will put constraints which are the higher gauge equivalent of the matched pair of Lie algebras. Since now we have to deal with both 1- and 2-gauge transformations, there are more compatibility rules to consider.  

\subsection{Dual 2-gauge theories}
Given the principal 2-bundle $\mathcal{P}\rightarrow M$ with structure 2-group $\mathcal{G}$, it is equipped a {\it 2-connection} $(A,\Sigma)$ consisting of a $\mathfrak{g}_0$-valued 1-form $A$ and $\mathfrak{g}_{-1}$-valued 2-form $\Sigma$.

The 2-gauge (infinitesimal) transformations are given by \cite{Baez:2002jn,Girelli:2003ev,Mikovic:2016xmo} (recall $[\cdot,\cdot]=\operatorname{ad}$)
\begin{equation}
    \lambda:\begin{cases}A\rightarrow A^\lambda = A + d_A\lambda \\ \Sigma\rightarrow \Sigma^\lambda = \Sigma+ \lambda\rhd\Sigma\end{cases},\qquad L:\begin{cases}A\rightarrow A^L = A + tL\\ \Sigma\rightarrow \Sigma^L = \Sigma + d_AL + \frac{1}{2}[L\wedge L]\end{cases},\nonumber 
\end{equation}
where $d_A\lambda = d\lambda + [A,\lambda]$ and $\lambda\in\Omega^0(M)\otimes\mathfrak{g}_0$ and $L\in\Omega^1(M)\otimes\mathfrak{g}_{-1}$ are gauge parameters. Given the following quantities
\begin{equation}
    \mathcal{F}=F - t\Sigma,\qquad \mathcal{G}=d_A\Sigma\label{2flat}
\end{equation}
where $F=dA+\frac{1}{2}[A\wedge A]$ is the curvature of $A$, one can then compute that they in fact transform covariantly \cite{Martins:2010ry,Radenkovic:2019qme}:
\begin{equation}
 \lambda : \begin{cases}\mathcal{F}\rightarrow \mathcal{F}^\lambda = \mathcal{F} + [\lambda,\mathcal{F}] \\ \mathcal{G} \rightarrow \mathcal{G}^\lambda = \mathcal{G} + \lambda\rhd \mathcal{G} \end{cases}, \qquad
 L: \begin{cases} \mathcal{F}\rightarrow \mathcal{F}^L= \mathcal{F}  \\ 
  \mathcal{G}\rightarrow \mathcal{G}^L= \mathcal{G} + \mathcal{F}\wedge^\rhd L
 \end{cases}. \label{2BFgaugetrans}
\end{equation}
The quantity $\mathcal{F}$ is known as {\bf fake-flatness}, while $\mathcal{G}$ is the {\bf 2-curvature}. 

\paragraph{Dual 2-gauge structure.} Let us consider now a dual principal 2-bundle $\mathcal{P}^*\rightarrow X$, which is equipped with a dual 2-connection $(C,B)$ valued in the strict 2-algebra $\mathfrak{g}^*[1]$. We thus have the corresponding 2-gauge transformations (recall $[\cdot,\cdot]_*=\mathfrak{ad}$)
\begin{equation}
    \tilde\lambda: \begin{cases}C\rightarrow C^{\tilde\lambda} = C + d_C\tilde\lambda \\ B\rightarrow B^{\tilde\lambda} = B + \tilde\lambda \rhd^* B\end{cases},\qquad \tilde L:\begin{cases}C\rightarrow C^{\tilde L} = C + \tilde t \tilde L \\ B\rightarrow B^{\tilde L} = B + d_C\tilde L + \frac{1}{2}[\tilde L\wedge \tilde L]_*\end{cases},\nonumber
\end{equation}
where $d_C\tilde\lambda=d\tilde\lambda + [C,\tilde\lambda]_*$ and $\tilde\lambda\in\Omega^0(M)\otimes\mathfrak{g}_{-1}^*$ and $\tilde L\in\Omega^1(M)\otimes\mathfrak{g}_0^*$ are gauge parameters.

The fake-flatness and 2-curvature quantities are given by
\begin{equation}
    \tilde{\mathcal{F}} = \tilde F - \tilde tB,\qquad \tilde{\mathcal{G}} = d_CB, \label{2dualflat} 
\end{equation}
where $\tilde F= d_CC=dC+\frac{1}{2}[C\wedge C]_*$ is the curvature of $C$. One may compute that
\begin{equation}
 \tilde\lambda : \begin{cases}\tilde{\mathcal{F}}\rightarrow \tilde{\mathcal{F}}^{\tilde\lambda} = \tilde{\mathcal{F}} + [\tilde\lambda,\mathcal{F}]_* \\ \tilde{\mathcal{G}} \rightarrow \tilde{\mathcal{G}}^{\tilde\lambda} = \tilde{\mathcal{G}} + \tilde\lambda\rhd^* \mathcal{G} \end{cases}, \qquad
 \tilde L: \begin{cases} \tilde{\mathcal{F}}\rightarrow \tilde{\mathcal{F}}^{\tilde L}= \tilde{\mathcal{F}}  \\ 
 \tilde{ \mathcal{G}}\rightarrow \tilde{\mathcal{G}}^{\tilde{L}}= \tilde{\mathcal{G}} + \tilde{\mathcal{F}}\wedge^{\rhd^*} \tilde L
 \end{cases}, \label{dual2BFgaugetrans}
\end{equation}
form which we see that the quantities $\tilde{\mathcal{F}},\tilde{\mathcal{G}}$ are indeed covariant.

\paragraph{The coadjoint action and back-action.} Now recall that the crossed-module $\mathfrak{g}$ and its dual $\mathfrak{g}^*[1]$ define natural representations on each other given in Eq. \eqref{rep}. With these, the 2-connections $(A,\Sigma)$ shall transform under the dual gauge parameters $\tilde\lambda,\tilde L$, as follows:
\begin{eqnarray}
    \lambda: \begin{cases}B\rightarrow B^\lambda = B + \operatorname{ad}_\lambda^* B \\ C\rightarrow C^\lambda = C + \chi_\lambda^* C \end{cases},&\qquad& L: \begin{cases} B\rightarrow B^L= B + \Delta(C\wedge L) \\ C\rightarrow C^L= C\end{cases}, \label{2BFdualgaugetrans}
\end{eqnarray}
and similarly for the dual 2-connections $(C,B)$ as
\begin{eqnarray}
    \tilde\lambda: \begin{cases}\Sigma\rightarrow \Sigma^{\tilde\lambda} = \Sigma + \mathfrak{ad}^*_{\tilde\lambda}\Sigma \\ A\rightarrow A^{\tilde\lambda} = A + \eta_{\tilde\lambda}^* A \end{cases},&\qquad& \tilde L: \begin{cases} \Sigma\rightarrow \Sigma^{\tilde L}= \Sigma + \tilde\Delta(A\wedge \tilde L) \\ A\rightarrow A^{\tilde L}= A\end{cases}. \label{dual2BFdualgaugetrans}
\end{eqnarray}
These transformations have also appeared previously in the study of 2-BF theories \cite{Martins:2010ry, Radenkovic:2019qme}, which are also sometimes known as "$BFCG$-theory" \cite{Asante:2019lki}.

As in Section \ref{bfdd}, we posit that the gauge transformations given in Eqs. \eqref{2BFgaugetrans}, \eqref{dual2BFgaugetrans}, \eqref{2BFdualgaugetrans}, \eqref{dual2BFdualgaugetrans} should allow us to derive conditions such that the direct sum $\mathfrak{g} \oplus \mathfrak{g}^*[1]$ forms the {\bf 2-Manin triple}
\begin{equation}
    \mathfrak{d} = \mathfrak{g}~_{\operatorname{ad}^*}\bowtie_{\mathfrak{ad}^*}\mathfrak{g}^*[1].\nonumber
\end{equation}
In fact, we shall recover Eqs. \eqref{comm}, \eqref{abhom} as well, purely from the considerations of the 2-gauge symmetries present in each of the 2-algebra sectors.

\subsection{The 2-Drinfel'd double as 2-gauge symmetry}\label{2dd}
We follow the same strategy as in Section \ref{dd}, and use Eqs. \eqref{2BFgaugetrans}, \eqref{dual2BFgaugetrans}, \eqref{2BFdualgaugetrans}, \eqref{dual2BFdualgaugetrans} to perform the gauge transformations.

We begin by defining the combined 2-gauge parameter by
\begin{equation}
    \boldsymbol\lambda = \lambda+\tilde\lambda\in\Omega^0(M)\otimes(\mathfrak{g}_0\oplus \mathfrak{g}_{-1}^*),\qquad {\bf L}=L+\tilde L\in\Omega^1(M)\otimes(\mathfrak{g}_{-1}\oplus\mathfrak{g}_0^*). \nonumber
\end{equation}
We grouped them so that they have the same degree, as forms but also in terms of the grading.  

We are going to consider the different combinations we can make in different order and impose that the order does not matter. This will impose a set of constraints which will be equivalent to the definition of the 2-Drinfel'd double.

\subsubsection{Combined 1-gauge transformations} 
As in Section \ref{dd}, the 1-connection $A$, say, transforms not only under $\lambda$ but also $\tilde\lambda$, and similarly for $C$. By considering pure gauges $A = d\lambda,C=d\tilde\lambda$, the induced gauge transformations are
\begin{equation}
    \lambda \rightarrow\lambda+ \eta_{\tilde\lambda}^*\lambda+ o(\lambda d\tilde\lambda),\qquad \tilde\lambda\rightarrow \tilde\lambda+ \chi_\lambda^*\tilde\lambda + o(\tilde\lambda d\lambda), \nonumber
\end{equation}
modulo variations in the 1-gauge parameters. 

\paragraph{1-gauge transformations of the  1-form connections.}
We can follow an analogous treatment as previously, but with $\mathfrak{ad}^*,\operatorname{ad}^*$ replaced with $\eta^*,\chi^*$ in accordance with Eqs. \eqref{2BFgaugetrans}, \eqref{dual2BFgaugetrans}, \eqref{2BFdualgaugetrans}, \eqref{dual2BFdualgaugetrans}. Running the same arguments as in Section \ref{dd}, we introduce here $C$-dependent gauge transformations on $A$ and vice versa, in order to properly antisymmetrize the compatibility conditions. 
\begin{equation}
    A\rightarrow A+d_A\lambda - \eta_C^*\lambda,\qquad C\rightarrow C+d_C\tilde\lambda - \chi_A^*\tilde\lambda.\nonumber
\end{equation}
This leads us to consider the {\bf combined 1-connection}
\begin{equation}
    {\bf A} = A+C\in \Omega^1(M)\otimes(\mathfrak{g}_0\bowtie\mathfrak{g}_{-1}^*)\nonumber
\end{equation}
valued in the {\it degree-$0$} components of the direct sum $\mathfrak{g}\bowtie\mathfrak{g}^*[1]$.

Moreover, this antisymmetrization procedure introduces the modified covariant derivatives 
\begin{equation}
    D_{\bf A}^{(0)}\lambda = d_A\lambda - \eta_C^*\lambda,\qquad D_{\bf A}^{(0)} \tilde\lambda = d_C\tilde\lambda - \chi_A^*\tilde\lambda,\label{covdev1}
\end{equation}
for the combined 1-connection ${\bf A}=A+C$, in analogy with Eq. \eqref{covdev}. We therefore have the analogous result:
\begin{lemma}
The gauge transformation (including antisymmetrized terms)
\begin{equation}
    {\bf A}\rightarrow {\bf A}^{\boldsymbol\lambda} = {\bf A} + (D_{\bf A}^{(0)}\lambda +\eta_{\tilde\lambda}^*A)+( D_{\bf A}^{(0)}\tilde\lambda \chi_\lambda^*C ) \nonumber 
\end{equation}
is unambiguous (independent of the order at which $\lambda,\tilde\lambda$ acts) modulo terms of order $o(\lambda^2)+o(\tilde\lambda^2)+o(\tilde\lambda d\lambda)+o(\lambda d\tilde\lambda)$ \underline{iff} the {\bf first compatibility conditions}
\begin{eqnarray}
     \eta^*_{\tilde\lambda}[A,\lambda] &=& [\eta^*_{\tilde\lambda}A,\lambda]+ [A,\eta^*_{\tilde\lambda}\lambda] - \eta_{\chi_A^* \tilde\lambda}^*\lambda + \eta_{\chi_\lambda^*\tilde\lambda}^*A,\nonumber \\
    \chi^*_\lambda[C,\tilde\lambda]_* &=&[\chi_\lambda^*C,\tilde\lambda]_*+[C,\chi_\lambda^*\tilde\lambda]_* - \chi^*_{\eta_C^*\lambda}\tilde\lambda + \chi^*_{\eta_{\tilde\lambda}^*\lambda}C\label{match1}
\end{eqnarray}
are satisfied.
\end{lemma}

\paragraph{1-gauge transformations of the 2-form connections.}
Less trivially, the 2-connection $\boldsymbol\Sigma$ also transforms under the 1-gauge $\lambda+\tilde\lambda$. Consider first $\Sigma$. To study how it transforms, we compute the action of consecutive transformations in $\lambda$ and $\tilde\lambda$, with Eqs. \eqref{2BFgaugetrans},  \eqref{dual2BFdualgaugetrans} 
\begin{eqnarray}
    \overrightarrow{\boldsymbol\lambda}&:& \Sigma\xrightarrow{\lambda}\Sigma+\lambda\rhd \Sigma \nonumber \\
    &\qquad& \xrightarrow{\tilde\lambda} \Sigma + \mathfrak{ad}_{\tilde\lambda}^*\Sigma + \lambda\rhd \Sigma \nonumber \\
    &\qquad& \qquad + \lambda \rhd \mathfrak{ad}_{\tilde\lambda}^*\Sigma + (\eta_{\tilde\lambda}^*\lambda)\rhd \Sigma + \underbrace{(\eta_{\tilde\lambda}^*\lambda)\rhd (\mathfrak{ad}_{\tilde\lambda}^*\Sigma)}_{\sim o(\tilde\lambda^2)},\nonumber
\end{eqnarray}
while the other order gives
\begin{eqnarray}
    \overleftarrow{\boldsymbol\lambda} &:& \Sigma\xrightarrow{\tilde\lambda} \Sigma + \mathfrak{ad}_{\tilde\lambda}^*\Sigma \nonumber \\
    &\qquad& \xrightarrow{\lambda} \Sigma + \lambda\rhd\Sigma + \mathfrak{ad}_{\tilde\lambda}^*\Sigma \nonumber \\
    &\qquad&\qquad \mathfrak{ad}_{\tilde\lambda}^*(\lambda\rhd\Sigma) + \mathfrak{ad}_{\chi_\lambda^*\tilde\lambda}^*\Sigma + \underbrace{\mathfrak{ad}_{\chi_\lambda^*\tilde\lambda}^*(\lambda\rhd \Sigma)}_{\sim o(\lambda^2)}.\nonumber
\end{eqnarray}
The difference between them, modulo quadratic terms in the small gauge parameters $\lambda,\tilde\lambda$, then reads
\begin{equation}
    \mathfrak{ad}_{\tilde\lambda}^*(\lambda\rhd\Sigma) + \mathfrak{ad}_{\chi_\lambda^*\tilde\lambda}^*\Sigma - \lambda \rhd \mathfrak{ad}_{\tilde\lambda}^*\Sigma - (\eta_{\tilde\lambda}^*\lambda)\rhd \Sigma,\label{1diff2}
\end{equation}
and we must force it to vanish.

This argument can be repeated for the dual sector, with the replacements
\begin{equation}
    \Sigma\leftrightsquigarrow B,\qquad \chi^*\leftrightsquigarrow\eta^*,\qquad \mathfrak{ad}^*\leftrightsquigarrow\operatorname{ad}^*,\qquad \Delta\leftrightsquigarrow\tilde\Delta. \nonumber
\end{equation}
More precisely, we see from Eqs.  \eqref{dual2BFgaugetrans}, \eqref{2BFdualgaugetrans} that $B$ transforms as
\begin{eqnarray}
    \overrightarrow{\boldsymbol\lambda} &:& \Sigma\xrightarrow{\lambda} B + \operatorname{ad}_{\tilde\lambda}^*B \nonumber \\
    &\qquad& \xrightarrow{\tilde\lambda} B + \tilde\lambda\rhd^*B + \operatorname{ad}_{\lambda}^*B \nonumber \\
    &\qquad&\qquad \operatorname{ad}_{\lambda}^*(\tilde\lambda\rhd^*B) + \operatorname{ad}_{\eta_{\tilde\lambda}^*\lambda}^*B + \underbrace{\mathfrak{ad}_{\chi_\lambda^*\tilde\lambda}^*(\tilde\lambda\rhd^*B)}_{\sim o(\tilde\lambda^2)},\nonumber
\end{eqnarray}
while
\begin{eqnarray}
    \overleftarrow{\boldsymbol\lambda}&:& B\xrightarrow{\tilde\lambda}B+\tilde\lambda\rhd^*B  \nonumber \\
    &\qquad& \xrightarrow{\lambda} B + \operatorname{ad}_{\lambda}^*B + \tilde\lambda\rhd^*B \nonumber \\
    &\qquad& \qquad + \tilde\lambda \rhd^* \operatorname{ad}_{\lambda}^*B + (\chi_{\lambda}^*\tilde\lambda)\rhd^* B + \underbrace{(\chi_{\lambda}^*\tilde\lambda)\rhd^* (\operatorname{ad}_{\lambda}^*B)}_{\sim o(\lambda^2)}.\nonumber
\end{eqnarray}
Their difference modulo $\lambda^2,\tilde\lambda^2$ then reads
\begin{equation}
    \operatorname{ad}_{\lambda}^*(\tilde\lambda\rhd^*B) + \operatorname{ad}_{\eta_{\tilde\lambda}^*\lambda}^*B -  \tilde\lambda \rhd^* \operatorname{ad}_{\lambda}^*B - (\chi_{\lambda}^*\tilde\lambda)\rhd^* B,\label{2diff2}
\end{equation}
and we also require this to vanish. 

\paragraph{Antisymmetrization procedure for the 2-form connection.}
Similar to the 1-gauge case, we must antisymmetrize the last term $\mathfrak{ad}^*_{\chi_\lambda^*\tilde\lambda}\Sigma$ with respect to $\lambda,\Sigma$ on Eq. \eqref{1diff2}, and similarly with respect to $\tilde\lambda,B$ on Eq. \eqref{1diff2}. As such, we must modify how the 2-connections $\Sigma,B$ transform under the 1-gauge parameters $\lambda,\tilde\lambda$, respectively:
\begin{equation}
    \Sigma \rightarrow \Sigma + \lambda\rhd\Sigma -\mathfrak{ad}_B^*\lambda,\qquad B\rightarrow B+ \tilde\lambda\rhd^* B - \operatorname{ad}_\Sigma^*\tilde\lambda.\label{antisym2gau} 
\end{equation}
In contrast to the 1-group case, however, it is a bit subtle to see why such terms are necessary. 

The subtlety here is that the coadjoint action in the degree-(-1) sectors $\mathfrak{g}_{-1},\mathfrak{g}_0^*$ are determined by the second Pfeiffer identity, whence terms such as $\mathfrak{ad}_B^*\lambda,\operatorname{ad}_\Sigma^*\tilde\lambda$ can be written in terms of the maps $\Delta,\tilde\Delta$ by Eq. \eqref{abhom}. To be explicit, the second Pfeiffer identity Eq. \eqref{pfeif} gives
\begin{equation}
    \operatorname{ad}_YY' = [Y,Y'] = tY\rhd Y',\nonumber
\end{equation}
whence pairing against $f\in\mathfrak{g}^*_{-1}$ gives
\begin{equation}
    (\operatorname{ad}_Y^*f)(Y') = f([Y',Y]) = f(tY'\rhd Y) = (-\Delta_{Y}(f))(tY')\nonumber
\end{equation}
for any $Y,Y'\in\mathfrak{g}_{-1}$. Similarly in the dual sector, we have
\begin{equation}
    g'(\mathfrak{ad}_g^*X) = (\tilde t g')(-\tilde\Delta_g(X))\nonumber
\end{equation}
for each $g,g'\in\mathfrak{g}_0^*$ and $X\in\mathfrak{g}_0$. This allows us to rewrite
\begin{equation}
    \mathfrak{ad}_B^*\lambda = -\tilde\Delta_B(\lambda),\qquad \operatorname{ad}_\Sigma^*\tilde\lambda = -\Delta_\Sigma(\tilde\lambda).\label{2pfeifabhom}
\end{equation}

Now, if we compute from Eq. \eqref{antisym2gau} that
\begin{eqnarray}
    \overrightarrow{\boldsymbol\lambda} : \Sigma\rightarrow \dots - \mathfrak{ad}_B^*\lambda + \mathfrak{ad}_{\operatorname{ad}_\Sigma^*\tilde\lambda}^*\lambda,&\quad& \overleftarrow{\boldsymbol\lambda}: \Sigma\rightarrow \dots -\mathfrak{ad}_B^*\lambda,\nonumber \\
    \overrightarrow{\boldsymbol\lambda}: B\rightarrow \dots - \operatorname{ad}_\Sigma^*\tilde\lambda,&\quad& \overleftarrow{\boldsymbol\lambda}: B\rightarrow \dots - \operatorname{ad}_\Sigma^*\tilde\lambda + \operatorname{ad}_{\mathfrak{ad}_\Sigma^*\tilde\lambda},\nonumber
\end{eqnarray}
where "$\dots$" indicate terms we have encountered already previously. Then, by Eq. \eqref{2pfeifabhom}, the difference
\begin{equation}
    \mathfrak{ad}_{\operatorname{ad}_\Sigma^*\tilde\lambda}^*\lambda +  \operatorname{ad}_{\mathfrak{ad}_\Sigma^*\tilde\lambda}= \tilde\Delta_{\Delta_\Sigma(\tilde\lambda)}(\lambda) + \Delta_{\tilde\Delta_B(\lambda)}(\tilde\lambda)\nonumber
\end{equation}
in fact yields the desired antisymmetrization terms. Appending these terms to Eqs. \eqref{1diff2}, \eqref{2diff2} leads to the desired compatibility conditions 
\begin{eqnarray}
    \mathfrak{ad}_{\tilde\lambda}^*(\lambda\rhd \Sigma) &=& \lambda\rhd(\mathfrak{ad}^*_{\tilde\lambda}\Sigma)+ (\eta_{\tilde\lambda}^*\lambda)\rhd\Sigma  -\mathfrak{ad}^*_{\chi_\lambda^*\tilde\lambda}\Sigma + \tilde\Delta_{\Delta_{\Sigma}(\tilde\lambda)}(\lambda),\label{match11}\\
     \operatorname{ad}_\lambda^*(\tilde\lambda\rhd^*B)&=&\tilde\lambda\rhd^* (\operatorname{ad}_\lambda^* B)+ (\chi_\lambda^*\tilde\lambda)\rhd^* B- \operatorname{ad}^*_{\eta_{\tilde\lambda}^*\lambda}B  + \Delta_{\tilde\Delta_B(\lambda)} (\tilde\lambda)\label{match12}.
\end{eqnarray}
The $B$-dependent 1-gauge transformation on $\Sigma$ (and vice versa) suggest we should define a {\bf combined 2-connection}
\begin{equation}
    \boldsymbol\Sigma = \Sigma + B \in\Omega^2(M)\otimes(\mathfrak{g}_{-1}\oplus\mathfrak{g}_0^*)\nonumber
\end{equation}
valued in the {\it degree-$(-1)$} components of the direct sum $\mathfrak{g}\oplus\mathfrak{g}^*[1]$.

Moreover, a modified covariant action is also introduced by the combined 2-connection $\boldsymbol\Sigma=\Sigma+B$ via the antisymmetrization procedure:
\begin{eqnarray}
    D_{\boldsymbol\Sigma}^{(0)}\lambda &=& \lambda\rhd \Sigma -\mathfrak{ad}_B^*\lambda = \operatorname{ad}_{-1}(\Sigma)(\lambda)- \mathfrak{ad}_B^*\lambda,\nonumber\\ D_{\boldsymbol\Sigma}^{(0)}\tilde\lambda &=& \tilde\lambda\rhd^*B - \operatorname{ad}_\Sigma^*\tilde\lambda=\mathfrak{ad}_{-1}(B)(\tilde\lambda) - \operatorname{ad}_\Sigma^*\tilde\lambda.\label{covdev2}
\end{eqnarray}
This yields the following result:
\begin{lemma}
The 1-gauge transformation (including antisymmetrized terms in Eq. \eqref{antisym2gau})
\begin{equation}
    \boldsymbol\Sigma\rightarrow\boldsymbol\Sigma^{\boldsymbol\lambda} =\boldsymbol\Sigma + (D_{\boldsymbol\Sigma}^{(0)} \lambda + \mathfrak{ad}_{\tilde\lambda}^*\Sigma) + (D_{\boldsymbol\Sigma}^{(0)} \tilde\lambda +\operatorname{ad}_\lambda^*B) \nonumber
\end{equation}
is unambiguous modulo terms of the order $o(\lambda^2) + o(\tilde\lambda^2)+o(\tilde\lambda d\lambda)+ o(\lambda d\tilde\lambda)$ \underline{iff} the {\bf second compatibility conditions} Eqs. \eqref{match11}, \eqref{match12} are satisfied.
\end{lemma}

\subsubsection{Combined 2-gauge transformations} Let us now study the 2-gauge transformations under the  2-gauge parameter ${\bf L}= L + \tilde L$. As $A,C$ do not transform respectively under $\tilde L,L$, we see that there is no ambiguity in ${\bf A}$ under ${\bf L}$:
\begin{equation}
    {\bf A} \rightarrow {\bf A}^{\bf L} = {\bf A} + tL + \tilde t \tilde L = {\bf A} + T{\bf L},\nonumber
\end{equation}
where we introduced the total $t$-map $T=t\otimes 1+1\otimes \tilde t$. 
On the other hand, the 2-connections $\Sigma,B$ do transform under $\tilde L$ and $L$ respectively, but only as shifts by $\tilde\Delta_{\tilde L}(\wedge A),\Delta_L(\wedge C)$, respectively, that do not depend on $\Sigma,B$ themselves. Therefore, there is no induced transformations of $L$ and $\tilde L$ on each other, and terms of cubic order $o(L^2\tilde L)+o(\tilde L^2 L)$ do not occur. Consequently, the quadratic terms $L^2,\tilde L^2$ do not effect our results, and we shall neglect them in the following.

Now we consider the 2-gauge shift transformations on the 2-connections $\Sigma,B$, and demand the difference in the different order to be zero. Similar to before, this leads to a modification of the 2-gauge transformation, because we must antisymmetrize the terms in these differences with respect to $\tilde L,A$ and $L,C$, respectively. This can be accomplished by inserting the terms $-\operatorname{ad}_A^*\tilde L$ into how $B$ transforms under $\tilde L$ in Eq. \eqref{2BFdualgaugetrans}. Similarly, we insert the term $-\mathfrak{ad}_C^*L$ into how $\Sigma$ transforms under $L$ in Eq. \eqref{2BFgaugetrans}. In other words, we introduce the 2-gauge covariant derivatives
\begin{equation}
    D^{(1)}_{\bf A}L = d_AL - \mathfrak{ad}_C^*(\wedge L),\qquad D^{(1)}_{\bf A}\tilde L = d_C\tilde L - \operatorname{ad}_A^*(\wedge \tilde L) \label{2covdev2}
\end{equation}
of the matched pair $\mathfrak{g}_{-1}\bowtie\mathfrak{g}_0^*$ in degree-(-1). 

We can check explicitly now that the order of the transformations does not matter provided we have some consistency conditions that we shall now derive. We compute directly with Eqs. \eqref{2BFgaugetrans}, \eqref{dual2BFgaugetrans} and the modified 2-gauge covariant derivatives Eq. \eqref{2covdev2}:
\begin{eqnarray}
    \overrightarrow{\bf L}&:&\Sigma\xrightarrow{L} \Sigma+d_AL- \mathfrak{ad}_C^*(\wedge L) \nonumber \\
    &\qquad& \xrightarrow{\tilde L}\Sigma + d_AL + \tilde\Delta_{\tilde L}(\wedge A) - \mathfrak{ad}_C^*(\wedge L)\nonumber \\
    &\qquad&\qquad - \mathfrak{ad}_{\tilde t\tilde L}^*(\wedge L).\nonumber\\
    \overleftarrow{\bf L} &:& \Sigma\xrightarrow{\tilde L}\Sigma + \tilde\Delta_{\tilde L}(\wedge A) \nonumber \\
    &\qquad&\xrightarrow{L} \Sigma + d_AL+ \tilde\Delta_{\tilde L}(\wedge A) - \mathfrak{ad}_C^*(\wedge L)\nonumber \\
    &\qquad&\qquad + \tilde\Delta_{\tilde L}(\wedge tL).\nonumber
\end{eqnarray}
Similarly, we may compute with Eqs.\eqref{2BFdualgaugetrans}, \eqref{dual2BFdualgaugetrans} that
\begin{eqnarray}
    \overleftarrow{\bf L}&:&B\xrightarrow{\tilde L} B+d_C\tilde L - \operatorname{ad}_A^*(\wedge \tilde L) \nonumber \\
    &\qquad& \xrightarrow{L}B + d_C\tilde L + \Delta_{L}(\wedge C) - \operatorname{ad}_A^*(\wedge \tilde L)\nonumber \\
    &\qquad&\qquad - \operatorname{ad}_{tL}^*(\wedge \tilde L).\nonumber\\
    \overrightarrow{\bf L} &:& B\xrightarrow{L}B + \Delta_{L}(\wedge C) \nonumber \\
    &\qquad&\xrightarrow{\tilde L} B + d_C\tilde L+ \Delta_{L}(\wedge C) - \operatorname{ad}_A^*(\wedge \tilde L)\nonumber \\
    &\qquad&\qquad +\Delta_{L}(\wedge \tilde t\tilde L).\nonumber
\end{eqnarray}

Now the difference in these transformations reads
\begin{equation}
    (\Delta_L(\wedge \tilde t \tilde L) +  \operatorname{ad}_{tL}^*(\wedge \tilde L))+( \tilde\Delta_{\tilde L}(\wedge tL)+\mathfrak{ad}_{\tilde t\tilde L}^*(\wedge L) ),\nonumber
\end{equation}
from which we see that the condition $\tilde t=t^*$ implies
\begin{equation}
    (\Delta_L(\wedge \tilde t \tilde L) -  \operatorname{ad}_{\tilde t^* L}^*(\wedge \tilde L))+( \tilde\Delta_{\tilde L}(\wedge tL)-\mathfrak{ad}_{ t^*\tilde L}^*(\wedge L) )=0,\nonumber
\end{equation}
which vanishes by Eq. \eqref{abhom}.

We therefore have the following result:
\begin{lemma}
\label{strictrep1}
The combined gauge transformation (including antisymmetrized terms)
\begin{equation}
    \boldsymbol\Sigma \xrightarrow{\bf L} \boldsymbol\Sigma^{\bf L} = \boldsymbol\Sigma + (D_{\bf A}^{(1)}L + \tilde\Delta_{\tilde L}(\wedge A) ) + (D_{\bf A}^{(1)}\tilde L +\Delta_L(\wedge C))\nonumber
\end{equation}
is unambiguous {\bf exactly} \underline{iff} $\tilde t=t^*$ and  $\Delta=\operatorname{ad}_{-1}^*, \tilde\Delta=\mathfrak{ad}_{-1}^*$ satisfy Eq. \eqref{abhom}.
\end{lemma}

\subsubsection{Mixed 1- and 2-gauge transformations}
Let us now study the combined gauge transformations under the mixed parameter $(\tilde \lambda,L)$, and dually under $(\lambda,\tilde L)$. Toward this, we must study how these gauge parameters interact. 

First, as $C$ is invariant under $L$, so is the pure-gauge 1-connection $C=d\tilde\lambda$ and hence there are no terms involving the action of $L$ on $\tilde\lambda$. On the other hand, a pure-gauge 2-connection $\Sigma =d_AL$ admits a gauge transformation
\begin{equation}
    d_AL \rightarrow d_AL + \mathfrak{ad}_{\tilde\lambda}^*d_AL . 
\end{equation}
By Leibniz rule, we see that we have
\begin{equation}
    dL \rightarrow dL + \mathfrak{ad}_{\tilde\lambda}^*dL = dL + d(\mathfrak{ad}_{\tilde\lambda}^*L) - \mathfrak{ad}_{d\tilde\lambda}^*\wedge L,\nonumber
\end{equation}
hence this induces an action of $\tilde\lambda$ on $L$:
\begin{equation}
    L\rightarrow L + \mathfrak{ad}_{\tilde\lambda}^*L + o(d\tilde\lambda L) \nonumber
\end{equation}
modulo terms proportional to $d\tilde\lambda$. 

On the dual side, we similarly have no action of $\tilde L$ on $\lambda$, as $A$ is invariant under $\tilde L$. On the other hand, we have
\begin{equation}
    d\tilde L \rightarrow d\tilde L + \operatorname{ad}_{\lambda}^*d\tilde L = d\tilde L + d(\operatorname{ad}_{\lambda}^*\tilde L) - \operatorname{ad}_{d\lambda}^*\wedge\tilde L,\nonumber
\end{equation}
which induces the action of $\lambda$ on $\tilde L$:
\begin{equation}
    \tilde L\rightarrow \tilde L + \operatorname{ad}_\lambda^*L + o(d\lambda \tilde L) \nonumber
\end{equation}
modulo terms of order $o(d\lambda \tilde L)$.

\paragraph{On 1-connections.}
With this acquired, we now consider how the 1-connections transform under these mixed gauge transformations. First as $\tilde L$ acts trivially on $A$, and similarly $L$ acts trivially on $C$, we see that there is in fact no ambiguity 
\begin{equation}
    A \xrightarrow{(\lambda,\tilde L)}A +d_A\lambda,\qquad C\xrightarrow{(\tilde\lambda,L)} C + d_C\tilde\lambda.\nonumber 
\end{equation}
On the other hand, we see that we have
\begin{eqnarray}
    \overrightarrow{(\tilde\lambda,L)}&:& A\xrightarrow{\tilde\lambda} A + \eta_{\tilde\lambda}^*A \xrightarrow{L} A + tL + \eta_{\tilde\lambda}^*A+ \eta_{\tilde\lambda}^*tL,\nonumber \\
    \overleftarrow{(\tilde\lambda,L)}&:& A\xrightarrow{L} A + tL \xrightarrow{\tilde\lambda} A + \eta_{\tilde\lambda}^*A + tL + t\mathfrak{ad}_{\tilde\lambda}^*L,\nonumber
\end{eqnarray}
and similarly
\begin{eqnarray}
    \overrightarrow{(\lambda,\tilde L)}&:& C\xrightarrow{\lambda} C + \chi_{\lambda}^*C \xrightarrow{\tilde L} C + \tilde t\tilde L + \chi_{\lambda}^*\tilde t\tilde L,\nonumber \\
    \overleftarrow{(\lambda,\tilde L)}&:& C\xrightarrow{\tilde L} C + \tilde t\tilde L \xrightarrow{\lambda} C + \chi_{\lambda}^*C + \tilde t\tilde L + \tilde t\operatorname{ad}_{\lambda}^*L.\nonumber
\end{eqnarray}
The difference in these gauge transformations is
\begin{equation}
    (\eta_{\tilde\lambda}^*tL - t\mathfrak{ad}_{\tilde\lambda}^*L) + (\chi_{\lambda}^*\tilde t\tilde L -  \tilde t\operatorname{ad}_{\lambda}^*\tilde L),\nonumber
\end{equation}
and the condition for them to vanish --- with $\tilde t = t^*$ understood --- is in fact nothing but the Pfeiffer identities
\begin{equation}
    \eta_{\tilde\lambda}^*\tilde t^*L = \tilde t^*\mathfrak{ad}_{\tilde\lambda}^*L,\qquad \chi_\lambda^*t^*\tilde L = t^* \operatorname{ad}_\lambda^*\tilde L\nonumber
\end{equation}
in Eq. \eqref{comm}. In other words, we have the following result:
\begin{lemma}
\label{strictrep2}
The mixed gauge transformations
\begin{equation}
    A \xrightarrow{(\tilde\lambda, L)} A + tL + \mathfrak{ad}_{\tilde\lambda}^*A + \mathfrak{ad}_{\tilde\lambda}^*tL,\qquad C \xrightarrow{(\lambda, \tilde L)} C + \tilde t\tilde L + \operatorname{ad}_{\lambda}^*C + \operatorname{ad}_{\lambda}^*\tilde t\tilde L \nonumber
\end{equation}
are unambiguous modulo the order $o(d\lambda \tilde L) + o(d\tilde\lambda L)+o(\tilde\lambda^2)+o(\lambda^2)$ \underline{iff} $\tilde t=t^*$ and $\operatorname{ad}_0^*=(\operatorname{ad}^*,\chi^*),\mathfrak{ad}_0^*=(\mathfrak{ad}^*,\eta^*)$ satisfy Eq. \eqref{comm}.
\end{lemma}
\noindent We therefore see that {\bf Lemmas \ref{strictrep1}, \ref{strictrep2}} make Eqs. \eqref{rep} define bona fide strict coadjoint representations, as anticipated.

\paragraph{On 2-connections.} We now turn to the mixed gauge transformations of the 2-connections $\Sigma,B$. Let us begin with $(\lambda,\tilde L)$. First, we have
\begin{eqnarray}
    \overrightarrow{(\lambda,\tilde L)} &:& \Sigma \xrightarrow{\lambda}\Sigma + \lambda\rhd \Sigma \xrightarrow{\tilde L}\Sigma + \lambda \rhd\Sigma+\tilde\Delta_{\tilde L}(\wedge A)  \nonumber \\
    &\qquad& \qquad +\lambda\rhd \tilde\Delta_{\tilde L}(\wedge A),\nonumber \\
    \overleftarrow{(\lambda,\tilde L)} &:& \Sigma\xrightarrow{\tilde L}\Sigma+\tilde\Delta_{\tilde L}(\wedge A)\nonumber \\
    &\qquad & \xrightarrow{\lambda} \Sigma + \lambda\rhd \Sigma + \tilde\Delta_{\tilde L}(\wedge A)\nonumber \\
    &\qquad&\qquad + \tilde\Delta_{\operatorname{ad}_{\lambda}^*\tilde L}(\wedge A) +  \tilde\Delta_{\tilde L}(\wedge d_A\lambda) + \underbrace{\tilde\Delta_{\operatorname{ad}_\lambda^*\tilde L}(\wedge d_A\lambda)}_{\sim o(\lambda^2)}.\nonumber
\end{eqnarray}
Modulo terms of order $o(d\lambda \tilde L)+ o(\lambda^2)$, the difference is the quantity
\begin{equation}
    \tilde\Delta_{\tilde L}(\wedge [A,\lambda]) + \tilde\Delta_{\operatorname{ad}_{\lambda}^*\tilde L}(\wedge A)- \lambda\rhd \tilde\Delta_{\tilde L}(\wedge A).\nonumber
\end{equation}
In order to force this quantity to vanish, the last two terms must be antisymmetrized with respect to $A$ and $\lambda$. This can be accomplished by the covariant derivative Eq. \eqref{2covdev2}, from which we acquire the terms
\begin{equation}
    - \tilde\Delta_{\operatorname{ad}_A^*(\wedge \tilde L)}(\lambda) + A\wedge^\rhd \tilde\Delta_{\tilde L}(\lambda),\nonumber
\end{equation}
and achieve the compatibility condition
\begin{equation}
    \tilde\Delta_{\tilde L}(\wedge [A,\lambda]) =\lambda\rhd \tilde\Delta_{\tilde L}(\wedge A) - A\wedge^\rhd \tilde\Delta_{\tilde L}(\lambda) -\tilde\Delta_{\operatorname{ad}_{\lambda}^*\tilde L}(\wedge A)+\tilde\Delta_{\operatorname{ad}_A^*(\wedge \tilde L)}(\lambda).\label{match21}
\end{equation}

Next, we consider the dual transformation $(\tilde\lambda,L)$ on the 2-connection $B$. We notice that the above argument can be run through analogously; indeed, we have
\begin{eqnarray}
    \overrightarrow{(\tilde\lambda,L)} &:& B \xrightarrow{\tilde \lambda}B + \tilde\lambda\rhd^* B \xrightarrow{L}B + \tilde\lambda \rhd^*B+\Delta_{L}(\wedge C)  \nonumber \\
    &\qquad& \qquad +\tilde\lambda\rhd^* \Delta_{L}(\wedge C),\nonumber \\
    \overleftarrow{(\tilde\lambda,L)} &:& B\xrightarrow{ L}B+ \Delta_{ L}(\wedge C)\nonumber \\
    &\qquad & \xrightarrow{\tilde\lambda} B + \tilde\lambda\rhd^* B + \Delta_{ L}(\wedge C)\nonumber \\
    &\qquad&\qquad + \Delta_{\mathfrak{ad}_{\tilde\lambda}^* L}(\wedge C) +  \Delta_{ L}(\wedge d_C\tilde\lambda) + \underbrace{\Delta_{\mathfrak{ad}_{\tilde\lambda}^*L}(\wedge d_C\tilde\lambda)}_{\sim o(\tilde\lambda^2)}.\nonumber
\end{eqnarray}
Modulo terms of order $o(d\tilde\lambda L)+o(\tilde\lambda^2)$, the difference reads
\begin{equation}
    \Delta_{ L}(\wedge [C,\tilde\lambda]_*) + \Delta_{\mathfrak{ad}_{\tilde\lambda}^* L}(\wedge C)- \tilde\lambda\rhd^* \Delta_{ L}(\wedge C).\nonumber
\end{equation}
Antisymmetrizing the last two terms with the covariant derivative Eq. \eqref{covdev2}, and setting this difference to zero yields the dual set of Eq. \eqref{match21}:
\begin{equation}
    \Delta_{ L}(\wedge [C,\tilde\lambda]_*) =\tilde\lambda\rhd^* \Delta_{ L}(\wedge C) - C\wedge^{\rhd^*} \Delta_{L}(\tilde\lambda) -\Delta_{\mathfrak{ad}_{\tilde\lambda}^* L}(\wedge C)+\Delta_{\mathfrak{ad}_C^*(\wedge L)}(\tilde\lambda).\label{match22}
\end{equation}

Finally, the question now is what would one yield by considering the mixed gauge $(\tilde \lambda,L)$ on $\Sigma$, or dually the mixed gauge $(\lambda,\tilde L)$ on $B$? We in fact do not get anything new; indeed, we directly compute 
\begin{eqnarray}
    \overrightarrow{(\tilde\lambda,L)} &:& \Sigma\xrightarrow{\tilde\lambda} \Sigma + \mathfrak{ad}_{\tilde\lambda}^*\Sigma \nonumber \\
    &\qquad& \xrightarrow{L} \Sigma + d_AL + - \mathfrak{ad}_C^*(\wedge L) + \mathfrak{ad}_{\tilde\lambda}^*\Sigma \nonumber \\
    &\qquad&\qquad + \mathfrak{ad}_{\tilde\lambda}^*d_AL \nonumber \\    \overleftarrow{(\tilde\lambda,L)}&:& \Sigma\xrightarrow{L} \Sigma + d_AL -\mathfrak{ad}_C^*(\wedge L)\nonumber \\
    &\qquad&\xrightarrow{\tilde\lambda} \Sigma + d_AL  + \mathfrak{ad}_{\tilde\lambda}^*\Sigma-\mathfrak{ad}_C^*(\wedge L) \nonumber \\
    &\qquad& \qquad + d_A\mathfrak{ad}_{\tilde\lambda}^*L (\eta_{\tilde\lambda}^*A)\wedge^\rhd L + \nonumber \\
    &\qquad& \qquad +\underbrace{(\eta_{\tilde\lambda}^*A)\wedge^\rhd \mathfrak{ad}_{\tilde\lambda}^*L}_{o(\tilde\lambda^2)}.\nonumber
\end{eqnarray}
By neglecting terms of order $o(d\tilde\lambda L)+ o(\tilde\lambda^2)$, the difference reads
\begin{equation}
    \mathfrak{ad}_{\tilde\lambda}^*(A\wedge^\rhd L) - A\wedge^\rhd \mathfrak{ad}_{\tilde\lambda}^*L - (\eta_{\tilde\lambda}^*A)\wedge^\rhd L \nonumber .\nonumber
\end{equation}
Setting this quantity to zero requires us to antisymmetrize the last two terms with respect to $A$ and $L$.

To do this, we introduce the terms in the modified covariant derivative Eq. \eqref{covdev2}. This yields the compatibility condition
\begin{equation}
     \mathfrak{ad}_{\tilde\lambda}^*(A\wedge^\rhd L) = A\wedge^\rhd \mathfrak{ad}_{\tilde\lambda}^*L-L\wedge^\rhd \mathfrak{ad}_{\tilde\lambda}^*A - (\eta_{\tilde\lambda}^*A)\wedge^\rhd L + (\eta_{\tilde\lambda}^*L)\wedge^\rhd A.\nonumber
\end{equation}
Dually, the difference in how $B$ transforms under $(\lambda,\tilde L)$ gives rise to the compatibility condition
\begin{equation}
    \operatorname{ad}_{\lambda}^*(C\wedge^{\rhd^*}\tilde L) = C\wedge^{\rhd^*} \operatorname{ad}_{\lambda}^*\tilde L-\tilde L\wedge^{\rhd^*} \operatorname{ad}_{\lambda}^*C - (\chi_{\lambda}^*C)\wedge^{\rhd^*} \tilde L + (\chi_{\lambda}^*\tilde L)\wedge^{\rhd^*} C.\nonumber
\end{equation}
These conditions bear striking resemblance to the second compatibility conditions Eqs. \eqref{match11}, \eqref{match12}; they in fact lead to the same condition (cf. Eq. \eqref{2manin2}).


\subsubsection{The 2-Manin triple} We have derived the following result:
\begin{theorem}\label{2mt}
The combined 2-gauge transformations (including antisymmetrized terms) of the combined 2-connection $({\bf A},\boldsymbol\Sigma)=(A+C,\Sigma+B)$
\begin{eqnarray}
    \boldsymbol\lambda &:&\begin{cases}{\bf A} \rightarrow {\bf A}^{\boldsymbol\lambda}={\bf A} + (D_{\bf A}^{(0)}\lambda  + \eta_{\tilde\lambda}^*A) + (D_{\bf A}^{(0)}\tilde\lambda +\chi_{\lambda}^*C) \\ \boldsymbol\Sigma \rightarrow \boldsymbol\Sigma^{\boldsymbol\lambda} =\boldsymbol\Sigma + (D_{\boldsymbol\Sigma}^{(0)}\lambda+ \mathfrak{ad}_{\tilde\lambda}^*\Sigma) + (D_{\boldsymbol\Sigma}^{(0)}\tilde\lambda+ \operatorname{ad}_\lambda^*B)\end{cases},\nonumber \\ 
    {\bf L} &:&\begin{cases} {\bf A}\rightarrow {\bf A}^{\bf L} = {\bf A} + (tL) + (\tilde t\tilde L) \\ \boldsymbol\Sigma \rightarrow\boldsymbol\Sigma^{\bf L} = \boldsymbol\Sigma + (D_{\bf A}^{(1)}L+ \tilde\Delta_{\tilde L}(\wedge A)) + (D_{\bf A}^{(1)}\tilde L+\Delta_L(\wedge C)) \end{cases}, \nonumber
\end{eqnarray}
in terms of the modified covariant derivatives Eqs. \eqref{covdev1}, \eqref{covdev2}, \eqref{2covdev2}, are unambiguous (ie. independent of the order at which \underline{any} of the gauge parameters $\lambda,\tilde\lambda,L,\tilde L$ act) modulo terms of the order 
\begin{equation}
    o(d\lambda(\tilde\lambda+\tilde L)) + o(d\tilde\lambda (\lambda+L)) + o(d\tilde L L) + o(dL \tilde L)+o(\lambda^2) + o(\tilde\lambda^2)\nonumber
\end{equation}
\underline{iff} 
\begin{enumerate}
    \item the maps in Eq. \eqref{rep},
    \begin{equation}
    \operatorname{ad}^*:\mathfrak{g}\rightarrow\operatorname{End}\mathfrak{g}^*[1],\qquad \mathfrak{ad}^*:\mathfrak{g}^*[1]\rightarrow\operatorname{End}\mathfrak{g},\nonumber
\end{equation}
satisfy Eqs. \eqref{comm}, \eqref{abhom} and hence define strict coadjoint representations, and
\item the {\bf compatibility conditions} Eqs. \eqref{match1}, \eqref{match11}, \eqref{match12}, \eqref{match21}, \eqref{match22} are satisfied.
\end{enumerate}
\end{theorem}

These are in fact the necessary and sufficient conditions derived in Ref. \cite{Bai_2013} for a strict Lie 2-bialgebra $\mathfrak{g}$ to form the {\bf standard 2-Manin triple}\footnote{As a Drinfel'd double is equivalent to a Manin triple \cite{Chari:1994pz}, by abuse of language, we will also freely interchange the name 2-Drinfel'd double with 2-Manin triple. }
\begin{equation}
    \mathfrak{d} = \mathfrak{g}~_{\operatorname{ad}^*}\bowtie_{\mathfrak{ad}^*}\mathfrak{g}^*[1] \nonumber
\end{equation}
with its dual $\mathfrak{g}^*[1]$. Indeed, if we distill all of the above compatibility conditions to their Lie algebra values, then we obtain\footnote{Note certain signs are different due to the antisymmetry of the wedge product on forms.}
\begin{eqnarray}
    \eta_f^*[X,X'] &=& [\eta_f^*X,X']+[X,\eta_f^*X'] - \eta_{\chi^*_X f}^*X' + \eta_{\chi^*_{X'}f}^*X,\nonumber \\
    \chi_X^*[f,f']_* &=& [\chi_X^*f,f']_* + [f,\chi_X^*f]_* - \chi_{\eta_f^*X}^*f' + \chi_{\eta_{f'}^*X}f,\label{2manin1} \\
    \mathfrak{ad}_f^*(X\rhd Y) &=& X\rhd (\mathfrak{ad}_f^*Y) + (\eta_f^*X)\rhd Y - \mathfrak{ad}_{\chi_X^*f}Y + \tilde\Delta_{\Delta_Y(f)}(X),\nonumber \\
    \operatorname{ad}_X^*(f\rhd^* g)&=& f\rhd^* (\operatorname{ad}_X^* g) + (\chi_X^*f)\rhd^* g - \operatorname{ad}_{\eta_f^*X}g + \Delta_{\tilde\Delta_g(X)}(f),\label{2manin2} \\
    \tilde\Delta_g([X,X']) &=& X\rhd \tilde\Delta_g(X') + X'\rhd \tilde\Delta_g(X) - \tilde\Delta_{\operatorname{ad}_X^*g}(X') + \tilde\Delta_{\operatorname{ad}_{X'}^*g}(X),\nonumber \\
    \Delta_Y([f,f']_*) &=& f\rhd^* \Delta_Y(f') + f'\rhd^*\Delta_Y(f) - \Delta_{\mathfrak{ad}_f^*Y}(f') + \Delta_{\mathfrak{ad}_{f'}^*Y}(f)\label{2manin3}
\end{eqnarray}
for each $X,X'\in\mathfrak{g}_0,f,f'\in\mathfrak{g}_{-1}^*,Y\in\mathfrak{g}_{-1},g\in\mathfrak{g}_0^*$.

Given the compatibility conditions Eqs. \eqref{2manin1}, \eqref{2manin2}, \eqref{2manin3}, the 2-Manin triple $\mathfrak{d}$ forms Lie algebra crossed-module. The $t$-map is given by
\begin{equation}
    T=t+\tilde t = t + t^*: \mathfrak{g}_{-1}\oplus\mathfrak{g}_0^* \rightarrow \mathfrak{g}_0\oplus\mathfrak{g}_{-1}^*,\nonumber
\end{equation}
and the crossed-module structures
\begin{eqnarray}
\pmb{[} X+f,X'+f' \pmb{]} &=& [X,X']+[f,f']_*+ \chi_X^*f' - \chi_{X'}^*f + \eta_{f}^*X' - \eta_{f'}^*X, \nonumber \\
(X+f)~\mathrlap{\triangleright}{\rhd}~ (Y+g) &=&   X \rhd Y  + f\rhd^* g + \tilde\Delta_g(X) - \operatorname{ad}_X^*g  + \Delta_Y(f)- \mathfrak{ad}_{f}^*Y\label{2maninbrac}
\end{eqnarray}
for each $X,X'\in\mathfrak{g}_0,Y\in\mathfrak{g}_{-1},f,f'\in\mathfrak{g}_{-1}^*,g\in\mathfrak{g}_0^*$. The degree-(-1) bracket 
\begin{eqnarray}
   \pmb{[} Y+g,Y'+g'\pmb{]} &=& T(Y+g)~\mathrlap{\triangleright}{\rhd} ~(Y'+g') \nonumber \\
   &=& [Y, Y']  + [g, g']_* + \tilde\Delta_{g'}(tY) - \operatorname{ad}_{tY}^*g' + \Delta_{Y'}(t^* g) - \mathfrak{ad}_{t^* g}^*Y'\nonumber
\end{eqnarray}
is given by the second Pfeiffer identity from Eq. \eqref{2maninbrac}.

Crucially, the canonical crossed-module structure on the 2-Manin triple $\mathfrak{d} = \mathfrak{g}\bowtie\mathfrak{g}^*[1]$ is defined in accordance with the grading. When written in terms of the components, we see that the "double" presentation and its crossed-module structure are related $$(\mathfrak{g}_{-1},\mathfrak{g}_0) \bowtie (\mathfrak{g}_0^*,\mathfrak{g}_{-1}^*) \leftrightsquigarrow (\mathfrak{g}_{-1},\mathfrak{g}_0^*) \xrightarrow{T} (\mathfrak{g}_0,\mathfrak{g}_{-1}^*),$$ by an exchange $\mathfrak{g}_0\leftrightsquigarrow \mathfrak{g}^*$ of a component at degree-0. 


\begin{example}\label{exp3}
Consider the bialgebra crossed-modules $(\operatorname{id}_\mathfrak{l};\psi,\psi)$ and $(\mathfrak{g}^0;\delta_{-1},\delta_0)$ given in {\bf Example \ref{exp2}}.
\begin{enumerate}
    \item {\bf Canonical 2-Manin triple:} Fix the 1-cocycle $\psi=\delta_{-1}=\delta_0$ on $\mathfrak{l}$. Due to the Pfeiffer identities, all of the components of Eq. \eqref{rep} collapse to merely the coadjoint representation:
    \begin{equation}
        \chi^* = \operatorname{ad}^* = -\Delta,\qquad \eta^* = \mathfrak{ad}^* = -\tilde\Delta.\nonumber
    \end{equation}
    Denote by $\mathfrak{1}_\mathfrak{l}$ the 2-Manin triple obtained from the canonical bialgebra crossed-module $(\operatorname{id}_\mathfrak{l};\psi,\psi)$ associated to the bialgebra $(\mathfrak{l};\psi)$, then we see that the compatibility conditions Eqs. \eqref{2manin1}, \eqref{2manin2}, \eqref{2manin3} collapse to merely three copies of the compatibility condition Eq. \eqref{1manin} for $\mathfrak{d}$.
    
    This means that the bialgebra crossed-modules $\mathfrak{1}_\mathfrak{l}= \operatorname{id}_\mathfrak{d}$ in fact coincide, as they have the same crossed-module structures given by Eq. \eqref{2maninbrac}. In essence, the canonical embedding commutes with the "Manin-ization" $\mathfrak{l}\mapsto \mathfrak{d}$; this fact had also been noted in Ref. \cite{Bai_2013}. We call $\mathfrak{1}_\mathfrak{l}=\operatorname{id}_\mathfrak{d}$ the {\bf canonical 2-Manin triple} associated to the bialgebra $(\mathfrak{l};\psi)$.

    \item {\bf Suspension 2-Manin triple:} Fix the 1-cocyce $\psi = \delta_{-1} + \delta_0$ on $V\rtimes\mathfrak{u}$ as given in {\bf Proposition \ref{suscocy}}. Summing the crossed-module structure Eq. \eqref{2maninbrac} on the 2-Manin triple $\mathfrak{d}^0$ of the suspension $\mathfrak{g}^0$ yields 
    \begin{eqnarray}
        \pmb{[}(X+Y)+(f+g),(X'+Y')+(f'+g')\pmb{]}_{\mathfrak{d}} &=&\pmb{[}X+f,X'+f'\pmb{]}_{\mathfrak{d}^0}  \nonumber \\
        &\qquad& + (X+f) ~\mathrlap{\triangleright}{\rhd}~ (Y'+g') \nonumber \\
        &\qquad& - (X'+f')~\mathrlap{\triangleright}{\rhd}~ (Y+g),\label{semi2manin}
    \end{eqnarray}
    where $X\in\mathfrak{u},g\in\mathfrak{u}^*,Y\in V,f\in V^*$. This coicides with the Lie bracket Eq. \eqref{maninbrac} on the Manin triple $\mathfrak{d}=\mathfrak{d}(V\rtimes\mathfrak{u})$ of the semidirect product $V\rtimes \mathfrak{u}$. 

    If we denote by $\mu \mathfrak{d}^0$ the 2-Manin triple $\mathfrak{d}^0$ with the grading ignored, then we see that $\mathfrak{d}=\mu\mathfrak{d}^0$ have identical algebraic structures. We call $\mathfrak{d}^0$ the {\bf suspension 2-Manin triple} associated to the bialgebra $(V\rtimes\mathfrak{u};\psi)$.
\end{enumerate}
\end{example}

It is clear from the Lie brackets Eq. \eqref{2maninbrac} that the components of $\mathfrak{d}$ with the same degree, namely $\mathfrak{g}_0\bowtie\mathfrak{g}_{-1}^*$ and $\mathfrak{g}_{-1}\bowtie\mathfrak{g}_0^*$, in fact themselves form Drinfel'd doubles. In other words, a Lie bialgebra crossed-module $(\mathfrak{g},\mathfrak{g}^*[1])$ defines two Lie bialgebras $(\mathfrak{g}_0,\mathfrak{g}_{-1}^*),(\mathfrak{g}_{-1},\mathfrak{g}_0^*)$ \cite{Chen:2013,Bai_2013}.

\subsection{Strict Lie 2-bialgebras and 2-Manin rigidity}
So far, the results we have obtained were all expressed in terms of Lie algebra crossed-modules. However, it is well-known that Lie algebra crossed-modules are equivalent to {\it strict Lie 2-algebras} \cite{Wag,Chen:2013,Bai_2013}. We shall see that {\it all} of our results carry over to the known 2-algebra scenario.

Recall a (strict) {\bf Lie 2-algebra} is a two-term graded vector space $\mathfrak{g}=\mathfrak{g}_{-1}\oplus \mathfrak{g}_0$ equipped with a degree-1 differential $l_1:\mathfrak{g}_{-1}\rightarrow\mathfrak{g}_0$ and a graded Lie bracket $l_2=[\cdot,\cdot]:\mathfrak{g}_i\wedge \mathfrak{g}_j\rightarrow\mathfrak{g}_{i+j}$ (where $i,j=0,-1$), such that the {\bf Koszul relations}
\begin{eqnarray}
    [X,X'] = -[X',X],\qquad [X,Y] = -[Y,X],\qquad [Y,Y']=0,\nonumber \\
    l_1 [X,Y] = [X,l_1Y],\qquad [l_1Y,Y']=[Y,l_1Y'],\nonumber \\
    {[[X,X'],X''] }+ [[X',X''],X] + [[X'',X],X']=0,\nonumber \\
    {[[X,X'],Y]} + [[X',Y],X] + [[Y,X],X']=0\nonumber
\end{eqnarray}
are satisfied for each $X,X',X''\in \mathfrak{g}_0$ and $Y,Y'\in\mathfrak{g}_{-1}$.

If we redefine $l_1=t:\mathfrak{g}_{-1}\rightarrow\mathfrak{g}_0$, and put
\begin{equation}
l_2|_{\Lambda^2\mathfrak{g}_0} = [\cdot,\cdot]_{\mathfrak{g}_0},\qquad l_2|_{\mathfrak{g}_0\wedge\mathfrak{g}_{-1}} = \mathfrak{g}_{0}\rhd \mathfrak{g}_{-1}, \label{translation}
\end{equation}
then we see that the above Koszul relations above reproduce precisely the Pfeiffer identities Eq. \eqref{pfeif}, as well as the 2-Jacobi identities Eqs. \eqref{2jacob}. Moreover, if the 2-cochain $(\delta_{-1},\delta_0)$ Eq. \eqref{dualcocy} introduced in Section \ref{algxmod} satisfy the coherence conditions Eqs. \eqref{dualcoh1}, \eqref{dualcoh2}, then $(\delta_{-1},\delta_0)$ is in fact a {\it Lie 2-algebra 2-cocycle}.

\begin{proposition}\label{2algxmod}
Lie algebra crossed-modules (resp. bialgebra crossed-modules) coincide with strict Lie 2-algebras (resp. Lie 2-bialgebras) as defined in Ref. \cite{Bai_2013}. 
\end{proposition}

Now, given the natural coadjoint representations Eq. \eqref{rep} of $\mathfrak{g}$ and its dual $\mathfrak{g}^*[1]$ on each other, the  conditions Eqs. \eqref{2manin1}, \eqref{2manin2}, \eqref{2manin3} in fact coincide with those derived in Ref. \cite{Bai_2013}. These conditions are necessary for $\mathfrak{g}\oplus\mathfrak{g}^*[1]$ to form a {\bf matched pair} of 2-algebras, hence our notion of the standard 2-Manin triple coincides with that given in Ref. \cite{Bai_2013} as well. 

\paragraph{2-Manin rigidity.}
Suppose, given a Lie 2-bialgebra $\mathfrak{g}$ and its dual $\mathfrak{g}^*[1]$, we introduce a pair of generic graded action/back-action
\begin{equation}
    \bar\rhd: \mathfrak{g}\rightarrow\operatorname{End}\mathfrak{g}^*[1],\qquad \bar\lhd:\mathfrak{g}^*[1]\rightarrow\operatorname{End}\mathfrak{g}.\nonumber
\end{equation}
Then we in fact have the following gauge-enforced result
\begin{corollary}\label{2maninrigid}
Given $\tilde t=t^*$, the induced combined 2-gauge transformations $(\boldsymbol\lambda,{\bf L}) = (\lambda+\tilde\lambda,L+\tilde L)$ are unambiguous up to the order given in {\bf Theorem \ref{2mt}} iff
\begin{enumerate}
    \item the action/back-action pair $\bar\rhd,\bar\lhd$ form strict coadjoint representations, and 
    \item the double
    \begin{equation}
        \mathfrak{d}' = \mathfrak{g}\overline{\bowtie}\mathfrak{g}^*[1] \nonumber
    \end{equation}
    is isomorphic to the standard 2-Manin triple $\mathfrak{d} = \mathfrak{g}~_{\operatorname{ad}^*}\bowtie_{\mathfrak{ad}^*}\mathfrak{g}^*[1]$.
\end{enumerate}
\end{corollary}
\begin{proof}
The same argument as in Section \ref{2dd}, leading up to {\bf Theorem \ref{2mt}}, goes through for the generic action/back-action $\bar\rhd,\bar\lhd$ such that, given $\tilde t = t^*$,
\begin{enumerate}
    \item they satisfy the analogues of Eqs. \eqref{comm} and \eqref{abhom}, and
    \item they satisfy the analogues of Eqs. \eqref{match1}, \eqref{match11}, \eqref{match12}, \eqref{match21}, \eqref{match22},
\end{enumerate}
with $\operatorname{ad}$ replaced by $\bar\rhd$ and $\mathfrak{ad}$ replaced by $\bar\lhd$. The induced crossed-module structure for $\mathfrak{d}'$ must then take the same form Eq. \eqref{2maninbrac}, which forces $\mathfrak{d}'\cong\mathfrak{d}$.
\end{proof}
\noindent This rigidity of 2-Manin triples was proven in Ref. \cite{Bai_2013}, and here we have demonstrated that this is enforced by gauge-theoretic considerations. In contrast to Corollary \ref{maninrigid}, however, the 2-gauge theory not only enforces this rigidity, but also the conditions Eqs. \eqref{comm}, \eqref{abhom} necessary for $\bar\rhd,\bar\lhd$ to form strict coadjoint representations.

\section{4D topological action invariant under the 2-Drinfel'd double}\label{combo2cov}
As $\mathfrak{d}$ is itself a (strict) Lie 2-algebra, we can study the 2-gauge theory associated to it. Here, we seek to construct curvature quantities that are covariant under the combined 2-gauge transformation $(\boldsymbol\lambda,{\bf L})$, briefly {\it 2-covariant}. In analogy with Section \ref{combocov}, we can define a combined curvature that reads
\begin{eqnarray}
    {\bf F} &=& d{\bf A} + \frac{1}{2}\pmb{[}{\bf A}\wedge{\bf A}\pmb{]} \nonumber\\
    &=&(dA + \frac{1}{2}[A\wedge A] + \eta_C^*(\wedge A))\nonumber \\
    &\qquad& + (dC + \frac{1}{2}[C\wedge C]_* + \chi_A^*(\wedge C)) \nonumber \\
    &\equiv& \bar F+ \bar F^*,\label{combo1curv}
\end{eqnarray}
which yields the {\bf combined fake-flatness}
\begin{equation}
    \boldsymbol{\mathcal{F}} = {\bf F} - T\boldsymbol\Sigma = (\bar F - t\Sigma) + (\bar F^* - t^*B) \equiv \bar{\mathcal{F}} + \bar{\mathcal{F}}^*,\label{comboff}
\end{equation}
by making use of the crossed-module map $T = t+ \tilde t$ for $\mathfrak{d}$. The bracket can be computed as
\begin{equation}
    \llbracket \boldsymbol{\mathcal{F}},\boldsymbol\lambda\rrbracket = ([\bar{\mathcal{F}},\lambda]+ \eta_{\tilde\lambda}^*\bar{\mathcal{F}} - \eta_{\bar{\mathcal{F}}^*}^*\lambda)+([\bar{\mathcal{F}}^*,\tilde\lambda]_*+ \chi_\lambda^*\bar{\mathcal{F}}^* - \chi_{\bar{\mathcal{F}}}^*\tilde\lambda).\nonumber
\end{equation}
Analogously, we define the {\bf combined 2-curvature}
\begin{eqnarray}
    \boldsymbol{\mathcal{G}} &=& d\boldsymbol\Sigma + {\bf A}\wedge^{\mathrlap{\triangleright}{\rhd}}\boldsymbol\Sigma \nonumber\\
    &=& (d\Sigma + A\wedge^\rhd \Sigma + \tilde\Delta_B(\wedge A)-\mathfrak{ad}_C^*(\wedge \Sigma)) \nonumber \\
    &\qquad& + (dB + C\wedge^{\rhd^*} + \Delta_\Sigma(\wedge C) -\operatorname{ad}_A^*(\wedge B) )\nonumber \\
    &\equiv & \bar{\mathcal{G}}+ \bar{\mathcal{G}}^*.\label{combo2c} 
\end{eqnarray}
The 2-covariance of $\boldsymbol{\mathcal{F}},\boldsymbol{\mathcal{G}}$ can once again be inferred immediately from the crossed-module structure of the standard 2-Manin triple $\mathfrak{d}$; for an explicit proof, see Appendix \ref{gencov}. It in fact allows us to rewrite the combined 2-gauge transformations in {\bf Theorem \ref{2mt}} as 
\begin{equation}
        \boldsymbol\lambda :\begin{cases}{\bf A}\rightarrow{\bf A}^{\boldsymbol\lambda} = {\bf A} + d_{\bf A}\boldsymbol\lambda \\ \boldsymbol\Sigma \rightarrow\boldsymbol\Sigma^{\boldsymbol\lambda} = \boldsymbol\Sigma + \boldsymbol\lambda ~\mathrlap{\triangleright}{\rhd}~\boldsymbol
    \Sigma\end{cases},\qquad
    {\bf L}:\begin{cases}{\bf A}\rightarrow{\bf A}^{\bf L} = {\bf A} + T{\bf L} \\ \boldsymbol\Sigma^{\bf L} = \boldsymbol\Sigma + d_{\bf A}{\bf L}\end{cases}.\label{combo2gau}
\end{equation}
This is a very useful and compact way of organizing the combined 2-gauge transformations.


\begin{example}\label{exp4}
Recall the canonical and suspension 2-Manin triples considered in {\bf Example \ref{exp3}}.
\begin{enumerate}
    \item  Pick a 1-cocycle $\psi$ on $\mathfrak{l}$. We demonstrate the gauge-theoretic consequences of $\mathfrak{1}_\mathfrak{l}=\operatorname{id}_\mathfrak{d}$. First, the 2-covariant quantities in the 2-Manin triple $\mathfrak{1}_\mathfrak{l}$ are given by Eq. \eqref{comboff} and \eqref{combo2c}. Recalling that all components of Eq. \eqref{rep} collapse to merely the coadjoint representation $\operatorname{ad}^*$, the sum of these 2-covariant quantities can in fact be written as
\begin{eqnarray}
    \bar{\mathcal{G}} + \bar{\mathcal{G}}^* &=& d(\Sigma+B) + \pmb{[}(A+C)\wedge(\Sigma+B)\pmb{]},\nonumber \\ \bar{\mathcal{F}}+\bar{\mathcal{F}}^*& =& d(A+C) + \frac{1}{2}\pmb{[}(A+C)\wedge (A+C)\pmb{]}-(\Sigma +B), \nonumber
\end{eqnarray}
where $\pmb{[}\cdot,\cdot\pmb{]}$ denotes the bracket Eq. \eqref{maninbrac} of the Manin triple $\mathfrak{d}$ of $\mathfrak{l}$. By interpreting $(A+C,\Sigma+B)$ as the combined 2-connection on $\operatorname{id}_\mathfrak{d}$, we see that $\mathfrak{1}_\mathfrak{l}$ and $\operatorname{id}_\mathfrak{d}$ host identical 2-covariant quantities, which is consistent with the fact that $\mathfrak{1}_\mathfrak{l}=\operatorname{id}_\mathfrak{d}$. 

\item Pick a 1-cocycle $\psi$ on $\mathfrak{h}=V\rtimes\mathfrak{u}$. As in {\bf Example \ref{exp3}}, let $\mathfrak{d}(\mathfrak{h})$ denote the Drinfel'd double associated to the bialgebra $(\mathfrak{h};\psi)$, and let $\mathfrak{d}^0$ denote the 2-Drinfel'd double associated to the skeletal suspension $\mathfrak{d}^0$ of $\mathfrak{h}$. We demonstrate the gauge-theoretic consequences of $\mathfrak{d}(\mathfrak{h})=\mu\mathfrak{d}^0$, where $\mu$ is the map that forgets the grading in the crossed-module. As the form-degree of the fields depend on the grading, there is a procedure that "promotes" the covariant quantities on $\mathfrak{d}(\mathfrak{h})$ to 2-covariant quantities based on $\mathfrak{d}^0$. 

The components $\bar F \in \mathfrak{h},\bar F^*\in\mathfrak{h}^*$ in the combined curvature ${\bf F} = \bar F+ \bar F^* \in\mathfrak{d}(\mathfrak{h})$  in Eq. \eqref{combocurv} can be written as %
\begin{equation}
    \bar F= dA + \frac{1}{2}[A,A]+\mathfrak{ad}_b^*(\wedge A),\qquad \bar F^* = db+ \frac{1}{2}[b\wedge b]_* + \operatorname{ad}_A^*(\wedge b). \nonumber
\end{equation}
By decomposing into components of the semidirect product $\mathfrak{h}=V\rtimes\mathfrak{u}$, we have
\begin{eqnarray}
    \bar F &=&dA|_\mathfrak{u}+ \frac{1}{2}[A|_{\mathfrak{u}}\wedge A|_\mathfrak{u}]+\eta_{b|_{V^*}}^*(\wedge A|_\mathfrak{u})\nonumber \\ 
    &\qquad& + dA|_V  + A|_\mathfrak{u}\rhd A|_V  +  \mathfrak{ad}_{b|_{V^*}}^*(\wedge A|_V) - \tilde\Delta_{b|_{\mathfrak{u}^*}}(\wedge A|_{\mathfrak{u}}),\nonumber \\
    \bar F^*&=& db|_{V^*}+ \frac{1}{2}[b|_{V^*}\wedge b|_{V^*}]_* + \chi_{A|_{\mathfrak{u}}}^*(\wedge b|_{V^*}) \nonumber\\
    &\qquad& + db|_{\mathfrak{u}^*}  + b|_{V^*}\rhd^* b|_{\mathfrak{u}^*} + \operatorname{ad}_{A|_{\mathfrak{u}}}^*(\wedge b|_{\mathfrak{u}^*}) - \Delta_{A|_{V}}(\wedge b|_{V^*}),\label{semidirectcurv}
\end{eqnarray}
where we have defined, for each $X\in\mathfrak{u},Y\in V,f\in V^*,g\in\mathfrak{u}^*$, the notations
\begin{equation}
    (\chi_X^*f)(Y) = \Delta_Y(f)(X) = -f(X\rhd Y),\qquad g(\eta^*_fX) = f(\tilde\Delta_g(X)) = -(f\rhd^*g)(X),\nonumber
\end{equation}
as per convention. This notation is convenient, as it allows us to identify
\begin{eqnarray}
    A|_{\mathfrak{u}} = A^0,&\quad& b|_{V^*} = C^0\nonumber \\ 
     \bar F|_{\mathfrak{u}} = \bar F^0,&\quad& \bar F^*|_{V^*} = (\bar F^0)^*, \nonumber
\end{eqnarray}
where ${\bf F}^0=\bar F^0 + (\bar F^0)^*$ the 2-covariant 1-curvature defined in Eq. \eqref{combo1curv}. 

Now what of the other components $A|_V,b|_{\mathfrak{u}^*}$ of the connections? If we bump up the form-degrees of these quantities to 2-forms, and rename them $\Sigma^0,B^0$, respectively, then the curvature quantities $\bar F|_V,\bar F^*|_{\mathfrak{u}^*}$ become 3-forms. From the formulas Eq. \eqref{semidirectcurv} above, these 3-forms are given by 
\begin{eqnarray}
    d\Sigma^0  + A|_\mathfrak{u}\rhd \Sigma^0  +  \mathfrak{ad}_{b|_{V^*}}^*(\wedge \Sigma^0) - \tilde\Delta_{B^0}(\wedge A|_{\mathfrak{u}}) &=& d\Sigma^0  + A^0\rhd \Sigma^0  +  \mathfrak{ad}_{C^0}^*(\wedge \Sigma^0) - \tilde\Delta_{B^0}(\wedge A^0),\nonumber\\
    dB^0  + b_{V^*}\rhd^* B^0 + \operatorname{ad}_{A|_{\mathfrak{u}}}^*(\wedge B^0) + - \Delta_{\Sigma^0}(\wedge b|_{V^*})&=& dB^0  + C^0\rhd^* B^0 + \operatorname{ad}_{A^0}^*(\wedge B^0) - \Delta_{\Sigma^0}(\wedge V^0),\nonumber
\end{eqnarray}
which is precisely the 2-covariant 2-curvature $\boldsymbol{\mathcal{G}}^0=\bar{\mathcal{G}}^0 + (\bar{\mathcal{G}}^0)^*$ given in Eq. \eqref{combo2c} for the 2-Drinfel'd double $\mathfrak{d}^0$, with the combined 2-connection given by
\begin{equation}
    ({\bf A}^0,\boldsymbol\Sigma^0) = (A^0+C^0,\Sigma^0+B^0).\nonumber
\end{equation}
In this sense, the suspension of the Drinfel'd double to make it a 2-Drinfel'd double can be seen equivalently as considering the bumping in dimensions of the forms $A|_V,b|_{\mathfrak{u}^*}\rightsquigarrow \Sigma^0,B^0$. 
\end{enumerate}
\end{example}

\subsection{BF theory and interacting 2-BF theory.}
Using the 2-covariant quantities $\boldsymbol{\mathcal{F}},\boldsymbol{\mathcal{G}}$ in Eq. \eqref{comboff}, \eqref{combo2c}, we can consider a 2-gauge theory based on the 2-Manin triple $\mathfrak{d} = \mathfrak{g}\bowtie\mathfrak{g}^*[1]$. First, given the bilinear form Eq. \eqref{2bilin} on $\mathfrak{d}$, we can write down the {\bf BF theory}
\begin{eqnarray}
    S_\textbf{BF}[{\bf A},\boldsymbol\Sigma] &=& \frac{1}{2}\int_M \langle\langle \boldsymbol\Sigma \wedge \boldsymbol{\mathcal{F}}\rangle\rangle = \frac{1}{2}\int_M \langle B\wedge \bar{\mathcal{F}}\rangle + \langle \bar{\mathcal{F}}^*\wedge \Sigma\rangle \nonumber \\
    &=& \frac{1}{2}\int_M \langle B\wedge \bar F\rangle + \langle\bar{F}^*\wedge \Sigma\rangle - \left[ \langle B\wedge t\Sigma\rangle + \langle  t^* B \wedge \Sigma\rangle\right]\label{monsbf}
\end{eqnarray}
on a 4-dimensional manifold $M$, where $\bar F$ and $\bar F^*$ are defined in Eq. \eqref{combo1curv}. 

Notice that the last two terms are in fact identical:
\begin{equation}
    \langle t^* B \wedge \Sigma\rangle = \langle B\wedge t \Sigma\rangle.\nonumber
\end{equation}
Moreover, by decomposing into components, each of the first two terms read explicitly
\begin{eqnarray}
    \langle B\wedge \bar F\rangle &=& \langle B\wedge F\rangle + \langle B\wedge \eta_C^*(\wedge A)\rangle,\nonumber\\
    \langle \bar F^*\wedge \Sigma\rangle &=& \langle dC + \frac{1}{2}[C\wedge C]_* + \chi_A^*(\wedge C)\wedge\Sigma\rangle,\nonumber
\end{eqnarray}
where $F = dA+\frac{1}{2}[A\wedge A]$ is the conventional 1-curvature on $\mathfrak{g}_0$. An integration by parts (neglecting the boundary term $d\langle C\wedge \Sigma\rangle$) yields
\begin{equation}
    \langle \bar{F}^*\wedge \Sigma\rangle = -\langle C\wedge \mathcal{G}\rangle + \frac{1}{2}\langle [C\wedge C]_*\wedge \Sigma\rangle\nonumber
\end{equation}
in terms of the conventional 2-curvature $\mathcal{G}=d\Sigma + A\wedge^\rhd \Sigma$ on $\mathfrak{g}_{-1}$. Thus we see that the  BF theory Eq. \eqref{monsbf} can be written as 
\begin{eqnarray}
    S_\textbf{BF}[{\bf A},\boldsymbol\Sigma] &=& \frac{1}{2}\int_M \langle B\wedge \mathcal{F}'\rangle - \langle C\wedge \mathcal{G}\rangle \nonumber \\
    &\qquad& + \frac{1}{2}\int_M \langle \frac{1}{2}[C\wedge C]_*\wedge \Sigma\rangle + \langle B\wedge \eta_C^*(\wedge A)\rangle ,\nonumber
\end{eqnarray}
which is a {\bf 2-BF theory with interaction terms}, where $\mathcal{F}' = F - 2t\Sigma$ is a kind of "doubled" fake-flatness.

\medskip 

The BF action is   invariant under  the 2-Drinfel'd double $\mathfrak{d}=\mathfrak{g}\bowtie\mathfrak{g}^*[1]$.  From Eq. \eqref{combo2gau} and from invariance of the bilinear form $\langle\langle\cdot,\cdot\rangle\rangle$ we have 
\begin{eqnarray}
    \langle \langle \boldsymbol\lambda~\mathrlap{\triangleright}{\rhd}~\boldsymbol{\Sigma}\wedge \boldsymbol{\mathcal{F}}\rangle\rangle =  \langle\langle \boldsymbol\Sigma\wedge \llbracket\boldsymbol\lambda,\boldsymbol{\mathcal{F}}\rrbracket\rangle\rangle,\qquad \langle\langle ({\bf A} \wedge^{\mathrlap{\triangleright}{\rhd}}{\bf L})\wedge \boldsymbol{\mathcal{F}}\rangle\rangle = -\langle\langle {\bf L}\wedge \llbracket{\bf A}\wedge \boldsymbol{\mathcal{F}}\rrbracket\rangle\rangle.\nonumber
\end{eqnarray}
The first equation indicates that the {\it integrand} $\langle\langle\boldsymbol\Sigma\wedge\boldsymbol{\mathcal{F}}\rangle\rangle$ is invariant under the 1-gauge. On the other hand, the second equation must be supplemented by a term $\langle\langle d{\bf L}\wedge \boldsymbol{\mathcal{F}}\rangle\rangle$, which coincides with $\langle {\bf L}\wedge d\boldsymbol{\mathcal{F}}\rangle\rangle$ modulo a boundary term. We thus have
\begin{equation}
    \langle \langle d_{\bf A}{\bf L}\wedge\boldsymbol{\mathcal{F}}\rangle\rangle = -\langle \langle{\bf L}\wedge d_{\bf A}\boldsymbol{\mathcal{F}}\rangle\rangle =0\nonumber
\end{equation}
from the Bianchi identity $d_{\bf A}{\bf F}=0$. This is essentially the same basic computation performed in Ref. \cite{Mikovic:2016xmo}.

\medskip

By varying the field $B$, we receive the equation of motion (EOM)
\begin{equation}
    \delta_BS_\textbf{BF}=0\implies \bar{\mathcal{F}}' = dA + \frac{1}{2}[A\wedge A] + \eta_C^*(\wedge A) -2t\Sigma =0\nonumber
 \end{equation}
for the doubled version of the fake-flatness Eq. \eqref{comboff} in the $\mathfrak{g}$-sector. To vary the field $C$, we recall that we may rewrite
\begin{equation}
    \langle B\wedge \eta_C^*(\wedge A)\rangle = -\langle( C\wedge^{\rhd^*}B)\wedge A\rangle = \langle C\wedge \tilde\Delta_B(\wedge A)\rangle,\label{deltaeq}
\end{equation}
then we obtain the flatness EOM
\begin{equation}
    \delta_CS_\textbf{BF} =0\implies \bar{\mathcal{G}}' = d\Sigma + A\wedge^{\rhd}\Sigma - \tilde\Delta_B(\wedge A) + \mathfrak{ad}_C^*(\wedge\Sigma) = 0 \nonumber
\end{equation}
for the $\mathfrak{g}$-sector of the combined 2-curvature Eq. \eqref{combo2c}, except the signs of the latter two terms are reversed.

On the other hand, it is clear that if we vary the fields $(A,\Sigma)$, then we obtain the dual counterparts of the above flatness EOMs
\begin{eqnarray}
    \delta_A S_\textbf{BF}=0&\Rightarrow& \bar{\mathcal{G}}'^* = dB + C\wedge^{\rhd^*}B - \Delta_\Sigma(\wedge C) + \operatorname{ad}_A^*(\wedge B) = 0\nonumber\\
    \delta_BS_\textbf{BF}=0&\Rightarrow& \bar{\mathcal{F}}'^* = dC + \frac{1}{2}[C\wedge C]_* + \chi_A^*(\wedge C) - 2\tilde tB = 0\nonumber
\end{eqnarray}
in the $\mathfrak{g}^*[1]$-sector. This means that the monster BF theory Eq. \eqref{monsbf} exhibits the 2-Drinfel'd double $\mathfrak{d}$ as symmetry, where the bilinear form Eq. \eqref{2bilin} is chosen.

\medskip

\begin{example}
As a specific example, let us consider the  skeletal crossed-module $(\R^6\xrightarrow{t=0}\so(3,1), \rhd,[\cdot,\cdot]_{\so(3,1)})$ which is the natural symmetry structure of  a $\so(3,1)$ $BF$ theory \cite{Baez:2010ya}. It can be interpreted as a 2-Drinfel'd double if one uses the Iwasawa decomposition of $\so(3,1)\approx \su(2)\bowtie\an_2$, and the decomposition $\R^6\approx \R^3\times \R^3$.  By re-ordering appropriately the components as was discussed in Section \ref{2dd}, the symmetries of the $\so(3,1)$ $BF$ theory  can also be seen as given in terms of a  pair of dual 2-gauge theories given in terms of the skeletal crossed-modules $\g=(\R^3 \xrightarrow{t=0}\su(2), \rhd, [\cdot,\cdot]_{\su(2)})$ and $\g^*[1]=(\R^3 \xrightarrow{t^*=0}\an_2, \rhd^*, [\cdot,\cdot]_{\an_2})$. Said otherwise, the $\so(3,1)$ $BF$ theory can roughly be viewed as the "Chern-Simons" formulation of a 2-gauge topological theory. Note that upon integration of the different crossed-modules, we would recover non-trivial Poisson 2-groups hence quantum 2-groups upon quantization, since the crossed-modules are equipped with non-trivial cocycles. 
\end{example}

Our construction emphasizes that there are hidden 2-gauge symmetries within the well-known 4D $BF$ theories. This point was already made in the specific case of the  2-Poincar\'e case \cite{Mikovic:2011si}. Our construction has important consequences for the discretization/quantization of the theory. Indeed, in 3D, quantizing either the Chern-Simons formulation or the $BF$ formulation means we construct our quantum states using different (quantum) groups. Of course, since it is the same theory, it means that there are relations among the different amplitudes: it is well-known that the Turaev-Viro amplitude can be related to the Chern-Simons amplitude \cite{ROBERTS, Freidel:2004nb}. If mathematically the two formulations are related/equivalent, physically some might be more relevant than others. For example, in the $BF$ formulation, one has access to the frame field hence the gravitational field. The loop quantum gravity formulation for 3D gravity relies on the $BF$ formulation \cite{Delcamp:2018sef}.  In the Chern-Simons formulation, the frame field is combined with the spin connection so that in this case it might more difficult to address physical questions specifically about the frame field. 

\medskip 

We would like to point out that something similar would happen in 4D. One can construct the quantum states of 4D BF theory as a gauge theory as it is now standard \cite{Baez:1999sr}. The new aspect we want to emphasize is that there  should also be an equivalent description of the quantum amplitudes in terms of a 2-gauge theory, presumably given in terms of the amplitude of the Yetter type \cite{Yetter:1993dh}. A first step in this direction was accomplished in the 2-Poincar\'e case, where it was shown that the theory could be discretized as a standard lattice gauge theory or instead as a 2-gauge theory \cite{Girelli:2021zmt}.  Once again, one formulation might be better suited than the other one according to the physical questions one intends to ask.

\subsection{Pairings on the 2-Drinfel'd double; the BF-BB theory.}  Similar to the 1-gauge case, the natural bilinear form Eq. \eqref{2bilin} was induced from the canonical evaluation pairing $(f+g,X+Y) = g(X) + f(Y)$, where $f\in\mathfrak{g}_{-1}^*,g\in\mathfrak{g}_0^*,Y,\in\mathfrak{g}_{-1},X\in\mathfrak{g}_0$. In this section, we explore the possibility of introducing an alternative pairing $\langle\langle\cdot,\cdot\rangle\rangle'$ on the direct sum $\mathfrak{g}\oplus\mathfrak{g}^*[1]$ of the crossed-modules.

\medskip

We proceed in direct analogy with the 1-algebra case. Given a bialgebra crossed-module $(\mathfrak{g};\delta_{-1},\delta_0)$, we can endow a 2-Manin triple structure on the direct sum $\mathfrak{d}\cong \mathfrak{g}\oplus\mathfrak{g}_{-1}^*\cong(\mathfrak{g}_{-1}\oplus\mathfrak{g}_0^*)\oplus(\mathfrak{g}_0\oplus\mathfrak{g}_{-1}^*)$ by equipping it with a non-degenerate symmetric invariant bilinear form $\langle\langle\cdot,\cdot\rangle\rangle$. In contrast to the 1-algebra case, it can be seen that the natural evaluation pairing Eq. \eqref{2bilin} was merely one of {\it four} possible pairing components on $\mathfrak{d}$.

\medskip

Specifically, given any $(Y+g)+(X+f)\in (\mathfrak{g}_{-1}\oplus\mathfrak{g}_0^*)\oplus(\mathfrak{g}_0\oplus\mathfrak{g}_{-1}^*)$, we have the following:
\begin{enumerate}
    \item {\bf Grading-inhomogeneous}, referring to pairings across sectors of distinct degrees. This case includes the natural blinear pairing Eq. \eqref{2bilin}, 
\begin{equation}
    \langle\langle  (Y+g)+(X+f) ,  (Y'+g')+(X'+f') \rangle\rangle = g(X') + g'(X) + f(Y') + f'(Y) 
\end{equation}
as well as the following
    \begin{equation}
     \langle \langle (Y+g)+(X+f),(Y'+g')+(X'+f')\rangle\rangle' = (\langle Y,X'\rangle + \langle f,g'\rangle) + (\langle X,Y'\rangle + \langle g,f'\rangle),\label{2bilin2}
    \end{equation}
    where $\langle\cdot,\cdot\rangle$ denotes an "off-diagonal" invariant bilinear form on the crossed-module $\mathfrak{g}\cong \mathfrak{g}_{-1}\oplus\mathfrak{g}_0$ and $\mathfrak{g}_0^*\oplus\mathfrak{g}_{-1}^*$. Non-degeneracy then requires $\mathfrak{g}_0$ and $\mathfrak{g}_{-1}$ to have the same dimension.
    \item {\bf Grading-homogeneous}, referring to pairings across sectors of the same degree, such as the following
    \begin{equation}
    \langle \langle (Y+g)+(X+f),(Y'+g')+(X'+f')\rangle\rangle_\text{hom} = (\langle Y,Y'\rangle_{-1} + \langle g,g'\rangle_0) + (\langle X,X'\rangle_0 + \langle f,f'\rangle_{-1})\nonumber
    \end{equation}
    induced by a "diagonal" invariant bilinear form $\langle\cdot,\cdot\rangle_{-1.0}$ on the crossed-module $\mathfrak{g}\cong \mathfrak{g}_{-1}\oplus\mathfrak{g}_0$. The other one is
    \begin{equation}
        \langle \langle (Y+g)+(X+f),(Y'+g')+(X'+f')\rangle\rangle'_\text{hom} = (\langle Y,g'\rangle_1 + \langle X,f'\rangle_0) + (\langle f,X'\rangle_0 + \langle g,Y'\rangle_1),\nonumber
    \end{equation}
    which is the degree-homogeneous version of the natural bilinear form Eq. \eqref{2bilin}.
\end{enumerate}

As usual, these alternative pairings satisfy certain invariance properties similar to the natural pairing Eq. \eqref{2bilin}. For the alternative pairing $\langle\langle\cdot,\cdot\rangle\rangle'$ given in Eq. \eqref{2bilin2}, for instance, we must have
\begin{eqnarray}
         \langle X\rhd Y,X'\rangle + \langle Y,\operatorname{ad}_XX'\rangle=0,&\quad& \langle f\rhd^*g,f'\rangle + \langle g,\mathfrak{ad}_f f'\rangle=0,\nonumber \\
        \langle \mathfrak{ad}_f^* Y,X\rangle + \langle Y, \eta_f^* X\rangle=0,&\quad& \langle \operatorname{ad}_X^*g,f\rangle + \langle g,\chi_X^* f'\rangle=0\nonumber\\
        \langle \tilde\Delta_g(X),X'\rangle + \langle X,\tilde\Delta_g(X')\rangle = 0, &\quad& \langle \Delta_Y(f),f'\rangle + \langle f, \Delta_Y(f')\rangle=0,\label{invar2manin}
\end{eqnarray}
where $\mathfrak{ad}_g,\operatorname{ad}_Y$ are induced by the second Pfeiffer identity Eq. \eqref{pfeif}.

\medskip 

The new ingredient here is, in addition to invariance, the following {\bf duality condition}\footnote{The natural pairing Eq. \eqref{2bilin} satisfies this by definition, as $(g,tY) = g(tY) = t^*g(Y) = (t^*g,Y)$.} (note $T^*=T$ is self-dual)
\begin{equation}
    \langle\langle T~\cdot~,~\cdot~\rangle\rangle = \langle\langle~ \cdot~,T~\cdot~\rangle\rangle\label{dualbigt}
\end{equation}
against the crossed-module map $T$ of the 2-Manin triple $\mathfrak{d}=\mathfrak{g}\bowtie\mathfrak{g}^*[1]$. For the alternative pairings, this puts non-trivial conditions on the crossed-module map $t$. Take the alternative pairing $\langle\langle\cdot,\cdot\rangle\rangle'$ given in Eq. \eqref{2bilin2}, say, then we must have
\begin{equation}
    \langle tY,Y'\rangle = \langle Y,tY'\rangle,\qquad \langle t^*g,g'\rangle = \langle g,t^*g'\rangle.\label{dualtmap}
\end{equation}
In other words, $t$ must be symmetric in the basis diagonalizing the bilinear form $\langle\cdot,\cdot\rangle$.

\medskip

Let consider again a BF theory type as in Eq. \eqref{monsbf}.
It is clear that the grading-homogeneous pairing types give a trivial theory: the 2-form $\boldsymbol{\Sigma}$ has degree-(-1) while the 2-form $\boldsymbol{\mathcal{F}}$ has degree-0, hence for this \textit{specific} choice of action, non-trivial pairings must be grading-inhomogeneous. 

Focusing on the alternative pairing Eq. \eqref{2bilin2}, the  BF theory then reads
\begin{eqnarray}
    S_\textbf{BF}'[{\bf A},\boldsymbol\Sigma] &=& \frac{1}{2}\int_M \langle\langle\boldsymbol\Sigma\wedge\boldsymbol{\mathcal{F}}\rangle\rangle' = \frac{1}{2}\int_M \langle B\wedge \bar{\mathcal{F}}^*\rangle + \langle \Sigma\wedge \bar{\mathcal{F}}\rangle\nonumber\\
    &=& \frac{1}{2}\int_M \langle B\wedge \bar F^*\rangle + \langle\Sigma \wedge \bar F\rangle - \left[ \langle B\wedge t^*B\rangle + \langle \Sigma \wedge t\Sigma\rangle\right]\nonumber\\
    &=& \frac{1}{2}\left(S_\text{BFBB}[A,\Sigma] + S_\text{BFBB}[C,B]\right) \nonumber \\
    &\qquad& + \frac{1}{2}\int_M \langle B\wedge \chi_A^*(\wedge C)\rangle + \langle \Sigma \wedge \eta_C^*(\wedge A)\rangle,\nonumber
\end{eqnarray}
which is a sum of the {\bf BF-BB theories}
\begin{equation}
    S_\text{BFBB}[A,\Sigma] = \int_M \langle \Sigma\wedge F\rangle - \langle \Sigma\wedge t\Sigma\rangle ,\qquad S_\text{BFBB}[C,B]=\int_M \langle B\wedge F^*\rangle - \langle B\wedge t^* B\rangle \nonumber 
\end{equation}
in each individual crossed-module sector $\mathfrak{g},\mathfrak{g}^*[1]$ respectively, plus the interaction terms
\begin{equation}
    \langle B\wedge \chi_A^*(\wedge C)\rangle=-\langle \operatorname{ad}_A^*(\wedge B)\wedge C\rangle,\qquad \langle \Sigma \wedge \eta_C^*(\wedge A)\rangle=-\langle \mathfrak{ad}_C^*(\wedge \Sigma)\wedge A\rangle.
\end{equation}
Note the invariance of $\langle\cdot,\cdot\rangle$ dualizes $\chi^*,\eta^*$ to the coadjoint $\operatorname{ad}^*,\mathfrak{ad}^*$, as opposed to the $\Delta,\tilde\Delta$ maps as in Eq. \eqref{deltaeq}. The BF-BB theory $S_\text{BFBB}$ is of central interest, in particular in the trivial case $t=\operatorname{id}$, as the 4D gravity action can be written as $S_\text{BFBB}$ plus a geometric {\it simplicity constraint} \cite{Asante:2019lki}. We shall explore this notion more in a following paper.

Similar to the case of the 1-gauge theory, we may consider the monster BF theory $S_\textbf{BF}''$ given by a linear combination of the alternative pairing with the canonical one
\begin{equation}
    \langle\langle\cdot,\cdot\rangle\rangle'' = \alpha\langle\langle \cdot,\cdot\rangle\rangle + \beta\langle\langle\cdot,\cdot\rangle\rangle',\nonumber
\end{equation}
for which the real parameters $\alpha,\beta$ satisfy $\alpha^2+ \beta^2 \neq 0$. We would then have
\begin{eqnarray}
    S_\textbf{BF}''[{\bf A},\boldsymbol\Sigma] &=& \alpha S_\textbf{BF}[{\bf A},\boldsymbol\Sigma]+ \beta S_\textbf{BF}'[{\bf A},\boldsymbol\Sigma]\nonumber\\
    &=& \frac{1}{2}\int_M \langle (\alpha B+\beta \Sigma)\wedge \bar F\rangle + \langle (\alpha\Sigma + \beta B)\wedge \bar F^*\rangle \nonumber \\
    &\qquad& - \frac{1}{2}\int_M \langle B\wedge (\alpha t\Sigma+ \beta t^*B)\rangle + \langle (\alpha t^*B+\beta t\Sigma)\wedge\Sigma\rangle.\nonumber 
\end{eqnarray}

\section{The 2-graded classical \texorpdfstring{$r$}{Lg}-matrix and the quadratic 2-Casimirs}\label{2grclrmat}

\paragraph{Classical $r$-matrix and the Yang-Baxter equation.}  
Recall that a Lie bialgebra is characterized by a Lie algebra $\mathfrak{g}$ together with a cocycle $\psi:\mathfrak{g}\rightarrow \mathfrak{g}^{2\wedge}$ satisfying Eq. \eqref{bialgcoh} \cite{Semenov1992,Chari:1994pz}, which defines the dual Lie algebra $\mathfrak{g}^*$. Moreover, this structure allows us to form the Manin triple $\mathfrak{d}=\mathfrak{g}\bowtie\mathfrak{g}^*$, subject to the compatibility condition Eq. \eqref{1manin}.

If this cocycle is exact $\psi = dr$ for some $r\in \mathfrak{g}^{2\wedge}$, then a particular algebra structure on the Manin triple is achieved. Such elements $r\in\mathfrak{g}^{2\wedge}$ is called a {\bf triangular classical $r$-matrix} \cite{Meusburger:2021cxe}. The Jacobi identity then follows from the {\bf classical Yang-Baxter equation} for $r$:
\begin{equation}
    [X^{3\otimes},[r_{12},r_{13}]+[r_{13},r_{23}]+[r_{12},r_{23}]]=0\nonumber\label{cybe}
\end{equation}
for every $X^{3\otimes} = X\otimes 1 \otimes 1 + 1\otimes X\otimes 1+ 1\otimes 1\otimes X\in \mathfrak{g}^{3\otimes}$. Here, the subscripts indicate the tensor components of $\mathfrak{g}^{3\otimes}$ upon which $r$ acts.


More generally, we call $r\in \mathfrak{g}^{2\otimes}$ a {\bf quasitriangular $r$-matrix} if it satisfies 
\begin{equation}
    [X^{2\otimes}, r + \sigma(r)]=0,\qquad X^{2\otimes}=X\otimes 1+1\otimes X \nonumber
\end{equation}
for each $X\in\mathfrak{g}$, in addition to the classical Yang-Baxter equation Eq. \eqref{cybe}, where $\sigma$ is a permutation of the tensor factors in $\mathfrak{g}^{2\otimes}$. Now if we decompose $r=r^\wedge+r^\odot$ into its skew-/symmetric components ($\wedge,\odot$ denote respectively the anti/symmetrized tensors), $r^\odot$ must be a $D(\operatorname{ad})$-invariant\footnote{Recall $[X^{2\otimes},\cdot]= \operatorname{ad}_X\otimes 1 + 1\otimes\operatorname{ad}_X = D(\operatorname{ad}_X)$.}, ie. a quadratic Casimir \cite{Beisert:2017xqx}. 

By characterizing the quadratic Casimirs of $\mathfrak{g}$, we may then solve for the skew-symmetric component $r^\wedge$ with the {\bf modified classical Yang-Baxter equation} \cite{Chari:1994pz,Beisert:2017xqx}
\begin{equation}
    \llbracket r^\wedge,r^\wedge\rrbracket = \omega = -\llbracket r^\odot,r^\odot\rrbracket,\label{modcybe}
\end{equation}
where we have introduced the shorthand (the Schouten bracket) $\llbracket r,r\rrbracket=[r_{12},r_{13}]+[r_{13},r_{23}]+[r_{12},r_{23}]$. This skew-symmetric part $r^\wedge$ then generates the 1-cocycle $\psi$ on $\mathfrak{g}$. 

In the following, we shall study the 2-graded version of the above structure by following the notions laid out in Ref. \cite{Bai_2013}, in which a definition of the "2-graded classical $r$-matrix" $R$ for a Lie bialgebra crossed-module/Lie 2-bialgebra was given. We shall pay particular attention to its symmetric component $R^\odot$, and present new results characterizing these so-called "quadratic 2-Casimirs". Some examples and applications are studied.

\subsection{The 2-graded classical \texorpdfstring{$r$}{Lg}-matrix}
Similar to the Lie 1-algebra case, a (strict) Lie 2-bialgebra can be classified in terms of a Lie algebra 2-cocycle $(\delta_{-1},\delta_0)$ satisfying Eqs. \eqref{dualcoh1}, \eqref{dualcoh2} \cite{Bai_2013,Chen:2013}, which induces a dual Lie 2-algebra $\mathfrak{g}^*[1]$. Moreover, the natural coadjoint representations Eq. \eqref{rep} gives rise to the 2-Drinfel'd double $\mathfrak{d}=\mathfrak{g}\bowtie\mathfrak{g}^*[1]$.

Similar to the 1-algebra case, we begin by considering a 2-cocycle $(\delta_{-1},\delta_0)$ that is a "2-coboundary". We in particular focus on the form of the 2-coboundary generated by certain elements $r_0\in\mathfrak{g}_0\wedge\mathfrak{g}_{-1}$ and $r_{-1}\in\mathfrak{g}_{-1}^{2\wedge}$. These elements $r_{-1},r_0$ form a {\bf triangular 2-graded classical $r$-matrix} \cite{Bai_2013} $$R=r_0-D(t)_{-1}r_{-1}=r_0- (t\otimes 1 + 1\otimes t)r_{-1}\in \mathfrak{g}_0\wedge\mathfrak{g}_{-1},$$ whence the 2-coboundary they form is given by
\begin{equation}
\delta_0(X) = [X\otimes 1+1\otimes X,R],\qquad \delta_{-1}(Y) = [Y\otimes 1+1\otimes Y,R].\label{2cocycRmat}
\end{equation}
Here we have used the graded Lie bracket $[\cdot,\cdot]=l_2$, which includes $[\cdot,\cdot]|_{\mathfrak{g}_0}$ as well as the crossed-module action $\rhd$ by Eq. \eqref{translation}.



More generally, suppose we are given $r_0\in (\mathfrak{g}_0\otimes\mathfrak{g}_{-1})\oplus(\mathfrak{g}_{-1}\otimes\mathfrak{g}_0),r_{-1}\in\mathfrak{g}_{-1}^{2\otimes}$ (namely not necessarily skew-symmetric elements). It was proven that \cite{Bai_2013}
\begin{theorem}\label{2cybe}
The 2-cochain $(\delta_{-1},\delta_0)$ Eq. \eqref{2cocycRmatcomp} makes $(\mathfrak{g};\delta_{-1},\delta_0)$ into a Lie 2-bialgebra iff for all $W\in \mathfrak{g}_0\oplus\mathfrak{g}_{-1}$,
\begin{itemize}
    \item $[W^{2\otimes},R+\sigma(R)]=0$ where $\sigma$ is an exchange of tensor factors, and
    \item the {\bf 2-graded classical Yang-Baxter equations} \cite{Bai_2013} are satisfied:
    \begin{enumerate}
        \item $D(t)_0 r_0=0$,
    \item  $[W^{3\otimes},[R_{12},R_{13}]+[R_{13},R_{23}]+[R_{12},R_{23}]]=0$
    \end{enumerate}
\end{itemize}
where $$W^{3\otimes}= W\otimes 1\otimes 1 + 1\otimes W\otimes 1 + 1\otimes 1\otimes W.$$
\end{theorem}
\noindent We call solutions $R\in (\mathfrak{g}_0\otimes\mathfrak{g}_{-1})\oplus(\mathfrak{g}_{-1}\otimes\mathfrak{g}_0)$ to the above criteria a {\bf quasitriangular 2-graded classical $r$-matrix}. 

In other words, the 2-graded classical Yang-Baxter equation implies the 2-cocycle condition \cite{Bai_2013} for 2-cochains $(\delta_{-1},\delta_0)$ defined in Eq. \eqref{2cocycRmat}. If we write out the components
\begin{equation}
    r_0 = \sum a\otimes b+ \bar a\otimes \bar b,\qquad r_{-1} = \sum c \otimes d \nonumber 
\end{equation}
for some $a,\bar b\in\mathfrak{g}_0$ and $\bar a, b,c,d\in\mathfrak{g}_{-1}$, then we have
\begin{equation}
    R = \sum \underbrace{(a\otimes b - tc\otimes d)}_{\in \mathfrak{g}_0\otimes\mathfrak{g}_{-1}} + \underbrace{(\bar a\otimes \bar b -  c\otimes td)}_{\in\mathfrak{g}_{-1}\otimes\mathfrak{g}_0}\equiv \sum \rho + \bar \rho.\label{2Rmat}
\end{equation}
By decomposing into skew-symmetric and symmetric parts $R=R^\wedge + R^\odot$, we have $\bar\rho =-\sigma\rho$ in $R^\wedge$ while $\bar\rho= \sigma\rho$ in $R^\odot$ in terms of the components defined in Eq. \eqref{2Rmat}, where $\sigma$ permutes the tensor factors.  In other words, we have
\begin{eqnarray}
    R^\wedge &=& \sum \rho -\sigma\rho =\sum a\wedge b - tc\wedge d - c\wedge td=\sum a\wedge b -D(t)_{-1}(c\wedge d),\nonumber\\
    R^\odot &=& \sum \rho+\sigma\rho= \sum a\odot b - tc\odot d - c\odot d=\sum a\odot b - D(t)_{-1}(c\odot d),\nonumber
\end{eqnarray}
If we write, using the graded Schouten bracket $\llbracket\cdot,\cdot\rrbracket$ \cite{Chen:2012gz,Bai_2013},
\begin{equation}
    \Omega = -\llbracket R^\odot,R^\odot\rrbracket =- [R_{12}^\odot,R_{13}^\odot]+[R_{13}^\odot,R_{23}^\odot]+[R_{12}^\odot,R_{23}^\odot],\nonumber
\end{equation}
then the skew-symmetric part $R^\wedge$ satisfies the {\bf modified 2-graded classical Yang-Baxter equation}
\begin{equation}
    \llbracket R^\wedge,R^\wedge \rrbracket = \Omega.\label{mod2ybe}
\end{equation}
This is an equivalent way of writing Eq. 2. in {\bf Theorem \ref{2cybe}}.

\medskip

As in the 1-algebra case, the symmetric component $R^\odot\in\mathfrak{g}_0\odot\mathfrak{g}_{-1}$ of $R$ governs the form of Eq. \eqref{mod2ybe}, while the skew-symmetric component $R^\wedge\in \mathfrak{g}_0\wedge\mathfrak{g}_{-1}$ contributes to the 2-coboundary Eq. \eqref{2cocycRmat}. Recalling $D(t)_{-1} = t\otimes 1+1\otimes t$, the 2-coboundary Eq. \eqref{2cocycRmat} is given explicitly by
\begin{eqnarray}
    \delta_0(X) &=& \sum [X,a]\wedge b + a\wedge (X\rhd b)\nonumber\\
    &\qquad& - ~ \sum [X,tc]\wedge d + tc\wedge (X\rhd d)  + (c\leftrightarrow d),\nonumber\\
    \delta_{-1}(Y)&=& \sum  c\wedge (td\rhd Y)+(c\leftrightarrow d) - \sum (a\rhd Y)\wedge b\label{2cocycRmatcomp},
\end{eqnarray}
where $c\leftrightarrow d$ indicates a swap of the elements $c,d$ from the previous term.

\medskip

One particular solution for the decomposition $R=R^\wedge+R^\odot$ is if the two quantities $r_0,r_{-1}$ can themselves be decomposed into skew-symmetric and symmetric parts:
\begin{eqnarray}
    r_0=r_0^\wedge+r_0^\odot,&\quad& r_0^\wedge \in \mathfrak{g}_0\wedge\mathfrak{g}_{-1},\qquad r_{-1}^\wedge\in \mathfrak{g}_{-1}^{2\wedge},\nonumber\\
    r_{-1}=r_{-1}^\wedge+r_{-1}^\odot,&\quad& r_0^{\odot} \in \mathfrak{g}_0\odot\mathfrak{g}_{-1},\qquad r_{-1}^\odot \in \mathfrak{g}_{-1}^{2\odot}.\nonumber
\end{eqnarray}
The 2-graded $r$-matrix then reads
\begin{eqnarray}
    R^\wedge &=& r^\wedge_0 - D(t)_{-1}r_{-1}^\wedge=\sum a\wedge b - D(t)_{-1}(c\wedge d),\nonumber\\ 
    R^\odot &=& r^\odot_{0} - D(t)_{-1}r_{-1}^\odot= \sum a\odot b - D(t)_{-1}(c\odot d).\label{2Rmatdecomp}
\end{eqnarray}
We stress that this may {\it not} be the most general form of the decomposition $R=R^\wedge+R^\odot$!

\medskip

Due to the first condition in {\bf Theorem \ref{2cybe}}, we see that the symmetric contribution $R^\odot$ must be $D(\operatorname{ad})$-invariant, where $D(\operatorname{ad}) = \operatorname{ad}\otimes 1+1\otimes\operatorname{ad}$ is defined in terms of the adjoint representation $\operatorname{ad}$ in Eq. \eqref{strictad}. In the 1-algebra case, this leads to $r^\odot$ being a quadratic Casimir, and so we shall correspondingly call $R^\odot$ a {\bf quadratic 2-Casimir} of the crossed-module $\mathfrak{g}$.

\subsection{Quadratic 2-Casimirs}\label{2cas}
Recall that quadratic 2-Casimirs, ie. the symmetric component $R^\odot$ of $R$, controls the form of the modified 2-graded classical Yang-Baxter equation Eq. \eqref{mod2ybe}. It is therefore very important towards understanding 2-graded classical $r$-matrices, and we shall dedicate the remainder of this section to characterizing it. 

\paragraph{Characterization of the Quadratic 2-Casimirs.} 
First, we note that $D(\operatorname{ad})=\operatorname{ad}\otimes 1+1\otimes\operatorname{ad}$ is the derivation on the tensor product $\mathfrak{g}^{2\otimes}$ associated to the strict adjoint representation $\operatorname{ad}$ Eq. \eqref{strictad}. The tensor product $\mathfrak{g}^{2\otimes}$ has the structure of a three-term graded complex, as we can extend $t$ to $D(t)_{-1,0}$ on the tensor product $\mathfrak{g}^{2\otimes}$. 
\begin{equation}
    \mathfrak{g}_{-1}^{2\otimes}\xrightarrow{D(t)_{-1}}(\mathfrak{g}_{-1}\otimes\mathfrak{g}_0)\oplus (\mathfrak{g}_0\otimes\mathfrak{g}_{-1})\xrightarrow{D(t)_0}\mathfrak{g}_0^{2\otimes},
\end{equation}
where
\begin{equation}
    D(t)_{-1}(Y\otimes Y') = tY\otimes Y' + Y\otimes tY',\qquad D(t)_0(Y\otimes X +X'\otimes Y') = tY\otimes X - X'\otimes tY' \nonumber
\end{equation}
for each $Y,Y'\in\mathfrak{g}_{-1},X,X'\in\mathfrak{g}_0$.

Let us thus mainly focus on the meaning of "$D(\operatorname{ad})$-invariance" in the definition of a quadratic 2-Casimir $R^\odot$ by examining in detail the action of $D(\operatorname{ad})$ on $(\mathfrak{g}_0\otimes\mathfrak{g}_{-1})\oplus(\mathfrak{g}_{-1}\otimes\mathfrak{g}_0)$. Let $Y\otimes X + X'\otimes Y'$ denote an arbitrary element in this sector, and let $X''+Y''\in \mathfrak{g}$, then
\begin{eqnarray}
    D(\operatorname{ad})_{X''+Y''} (Y\otimes X) &=&\underbrace{(X''\rhd Y) \otimes X + Y\otimes [X'',X]}_{\mathfrak{g}_{-1}\otimes\mathfrak{g}_0} + \underbrace{Y\otimes (X\rhd Y'')}_{\mathfrak{g}_{-1}^{2\otimes}},\nonumber\\
    D(\operatorname{ad})_{X''+Y''} (X'\otimes Y') &=& \underbrace{[X'',X'] \otimes Y' + X'\otimes (X''\rhd Y')}_{\mathfrak{g}_0\otimes\mathfrak{g}_{-1}}+\underbrace{(X'\rhd Y'')\otimes Y'}_{\mathfrak{g}_{-1}^{2\otimes}}.\nonumber
\end{eqnarray}
Now if we take the symmetric tensor $Y\odot X = Y\otimes X +X\otimes Y$ and sum the above contributions, then the $D(\operatorname{ad})$-invariance condition $D(\operatorname{ad})_{X''+Y''}(Y\odot X) = 0$ gives rise to the following equations
\begin{eqnarray}
   [X'',X] \odot Y + X\odot (X''\rhd Y)=0, \qquad (X\rhd Y'')\odot Y =0\nonumber
\end{eqnarray}
for all $X''+Y''\in\mathfrak{g}$. The space of solutions is the subspace 
\begin{equation}
    \Theta_\rhd = \{X\odot Y\in\mathfrak{g}_0\odot\mathfrak{g}_{-1}\mid \operatorname{ad}X\odot Y + X\odot \chi Y =0,~\chi_X =0\},\nonumber
\end{equation}
where we recall $\chi=\rhd$ is the crossed-module action. Now the condition $D(t)_0R = D(t)_0r_0=0$ in {\bf Theorem \ref{2cybe}} constrains $R^\odot$ to lie in $\operatorname{ker}D(t)_0$, whence we assemble the elements
\begin{equation}
    a\odot b \in \operatorname{2Cas}_\mathfrak{g}[0] \equiv \Theta_\rhd\cap \operatorname{ker}D(t)_0\nonumber
\end{equation}
as the quadratic 2-Casimirs of $\mathfrak{g}$.

On the other hand, for $Y\odot Y'\in\mathfrak{g}_{-1}^{2\odot}$ we have
\begin{eqnarray}
    D(\operatorname{ad})_{X''+Y''}(D(t)_{-1}(Y\odot Y'))&=& (X''\rhd Y)\odot tY' + Y\odot [X'',tY'] + Y\odot (tY'\rhd Y'') \nonumber \\
    &\qquad& + (X'' \rhd Y')\odot tY
    +Y'\odot [X'',tY]+ Y'\odot (tY\rhd Y'') \nonumber \\
    &=& (X''\rhd Y)\odot tY' + Y\odot t(X''\rhd Y') + Y\odot [Y',Y''] \nonumber \\
    &\qquad& + (X'' \rhd Y')\odot tY + Y'\odot t(X''\rhd Y)+ Y'\odot [Y, Y''] \nonumber\\
    &=& D(t)_{-1}((X''\rhd Y) \odot Y' + (X''\rhd Y')\odot Y)  \nonumber\\
    &\qquad& - ([Y'',Y']\odot Y + Y'\odot [Y'',Y]),\nonumber
\end{eqnarray}
where we have used the Pfeiffer identities Eq. \eqref{pfeif}. Note $D(t)_{-1}=t\otimes 1+1\otimes t$ on $\mathfrak{g}_{-1}^{2\odot}$, we define the subspaces
\begin{eqnarray}
    \Gamma_t &=& \{Y\odot Y'\in\mathfrak{g}_{-1}^{2\odot}\mid \chi Y\odot Y' + Y\odot \chi Y' \in\operatorname{ker}D(t)_{-1}\},\nonumber \\
    \operatorname{Cas}_{\mathfrak{g}_{-1}} &=&\{Y\odot Y'\in\mathfrak{g}_{-1}^{2\odot}\mid \operatorname{ad}Y\odot Y' + Y\odot \operatorname{ad}Y'=0\},\nonumber
\end{eqnarray}
we see that the space of solutions is given by the intersection $$c\odot d\in \operatorname{2Cas}_\mathfrak{g}[-1] \equiv \Gamma_t \cap \operatorname{Cas}_{\mathfrak{g}_{-1}}.$$ Recall the adjoint action $\operatorname{ad}$ on $\mathfrak{g}_{-1}$ is defined via the second Pfeiffer identity. If each term in $D(\operatorname{ad})_{X''+Y''}(R^\odot)=0$ vanishes, then we obtain the following characterization of quadratic 2-Casimirs:
\begin{equation}
    R^\odot = \sum a\odot b + D(t)_{-1}(c\odot d)\in \operatorname{2Cas}_\mathfrak{g}[0] \oplus D(t)_{-1}(\operatorname{2Cas}_\mathfrak{g}[-1]) \equiv \operatorname{2Cas}_\mathfrak{g},\nonumber 
\end{equation}
{\it provided} the decomposition Eq. \eqref{2Rmatdecomp} holds.



\subsection{Explicit constructions of the quadratic 2-Casimirs for some crossed-modules}
Recall from {\bf Example \ref{exp1}} that we can obtain a Lie 2-algebra from the data of a Lie algebra in two different ways: the canonical embedding and the suspension embedding. Here, let us compare the characterization of the quadratic Casimirs for the 1-algebra with that of the quadratic 2-Casimirs. We will also discuss the quadratic 2-Casimirs in the case of the 2-Dinfel'd double.

\subsubsection{2-Casimirs of the trivial canonical 2-algebra.} Consider the trivial 2-algebra $\operatorname{id}_\mathfrak{l}$ associated to a Lie algebra $\mathfrak{l}$. We have $(\operatorname{id}_\mathfrak{l})_{-1}=(\operatorname{id}_\mathfrak{l})_0=\mathfrak{l}$, and the crossed-module map $t=\operatorname{id}$ is the identity, hence
\begin{eqnarray}
    D(t)_{-1} = 1,&\qquad& \operatorname{ker}D(t)_{-1}=0,\nonumber\\
    D(t)_0 = 0, &\qquad& \operatorname{ker}D(t)_0=\mathfrak{l}\odot\mathfrak{l}.\nonumber
\end{eqnarray}
Thus, we have $\operatorname{ad}=\chi$ and $\Gamma_{t=1}  = \operatorname{Cas}_{(\operatorname{id}_{\mathfrak{l}})_{-1}}$, whence $\operatorname{2Cas}_{\operatorname{id}_\mathfrak{l}}[-1]=\Gamma_{t=1}= \operatorname{Cas}_\mathfrak{l}$.

On the other hand, the condition $\chi_{a}=\operatorname{ad}_{a} =0$ implies that $a\in \operatorname{ker}\operatorname{ad} =Z(\mathfrak{l})$ lies in the centre of $\mathfrak{l}$.  Let us take $\mathfrak{l}$  semisimple, then all Abelian ideals such as $Z(\mathfrak{l})$ vanish, whence no contributions occur from $\operatorname{2Cas}_{\operatorname{id}_\mathfrak{l}}[0]$. The only term that survives is therefore
\begin{equation}
    R^\odot =- \sum D(t)_{-1}(c\odot d) = -\sum c\odot d \in \Gamma_{t=1} = \operatorname{Cas}_\mathfrak{g};\nonumber
\end{equation}
namely the quadratic 2-Casimirs $\operatorname{2Cas}_{\operatorname{id}_\mathfrak{l}}$ of $\operatorname{id}_\mathfrak{l}$ coincide with those $\operatorname{Cas}_\mathfrak{l}$ of $\mathfrak{l}$.

In this case, it is clear that the Schouten brackets for $\mathfrak{l}$ and $\operatorname{id}_\mathfrak{l}$ coincide, hence the space of solutions of Eq. \eqref{modcybe} is isomorphic to that of Eq. \eqref{mod2ybe}. The correspondence is just $r \mapsto R$ \cite{Bai_2013}.

\begin{example}\label{exp5}
    Consider the simple Lie algebra $\mathfrak{l}=\mathfrak{su}(2)$ generated by the basis $J_1,J_2,J_3$ such that
    \begin{equation}
        [J_i,J_j] = \varepsilon^k_{ij}J_k,\qquad 1\leq i,j,k\leq 3,\nonumber
    \end{equation}
    where $\varepsilon$ are the Levi-Civita symbols. The skew-symmetric component $r^\wedge = J_1 \wedge J_2\in\mathfrak{su}(2)^{2\wedge}$ solves the modified classical Yang-Baxter equation Eq. \eqref{modcybe} associated to the following Casimir $r^\odot = -J_i\odot J_i \in \operatorname{Cas}_{\mathfrak{su}(2)}$, and generates the cocycle
    \begin{equation}
        \psi(J_1) = J_1\wedge J_3,\qquad \psi(J_2)= J_2\wedge J_3,\qquad \psi(J_3)=0.\nonumber
    \end{equation}
    This gives a bialgebra $(\mathfrak{su}(2),\psi)$ whose Drinfel'd double $\mathfrak{d}=\mathfrak{su}(2)\bowtie\mathfrak{su}(2)^* =\mathfrak{su}(2)\bowtie\mathfrak{an}_2 $ is known as {\bf the classical double of $\mathfrak{su}(2)$} \cite{Majid:1996kd, Gutierrez-Sagredo:2019ipf}. It is a deformation  of the 3D Poincar{\'e} algebra $\mathfrak{iso}(3)=\mathfrak{su}(2) \ltimes \mathbb{R}^3$.
    
    Now let us consider the trivial canonical 2-algebra $$\operatorname{id}_{\mathfrak{su}(2)}=(\mathfrak{su}(2)\xrightarrow{\operatorname{id}}\mathfrak{su}(2),\operatorname{ad},[\cdot,\cdot]).$$ The element $R^\odot_0 = -r^\odot = J_i\odot J_i$ defines an admissible 2-Casimir of $\operatorname{id}_{\mathfrak{su}(2)}$, hence the choice $R^\wedge=r_0^\wedge=J_1\wedge J_2$ solves the 2-graded modified classical Yang-Baxter equation Eq. \eqref{mod2ybe}. This is consistent with the observation made in {\it Example \ref{exp3}}, that bialgebra crossed-module $(\operatorname{id}_{\mathfrak{su}(2)};\delta_{-1},\delta_0)$ is then equivalent to two copies (ie. the canonical embedding) of the bialgebra $(\mathfrak{su}(2),\psi)$.
\end{example}

\subsubsection{2-Casimirs of the skeletal suspension 2-algebra.} Consider the semidirect product $V\rtimes\mathfrak{u}$, where $V$ is Abelian. It is known that, for $\mathfrak{u}$ semisimple and $\operatorname{dim}\mathfrak{u}=\operatorname{dim}V$, the Casimirs $\operatorname{Cas}_{V\rtimes\mathfrak{u}}$ span a 2-dimensional space \cite{Beisert:2017xqx}. Thus, the Killing form $K$ on $V\rtimes\mathfrak{u}$ can be decomposed into two parts\footnote{This decomposition will have nothing to do with {\bf Proposition \ref{bilinedecomp}}, even though we use a similar notation.}, $K=K_0+K_{-1}$, given by 
\begin{equation}
    K(X+Y,X'+Y') = K_0(X,X') + K_{-1}(X,Y') + K_{-1}(X',Y) + K_0(Y,Y')\nonumber
\end{equation}
for each $X,X'\in\mathfrak{u}$ and $Y,Y'\in V$. Indeed, $K_{-1}$ can be non-degenerate only if $\operatorname{dim}V=\operatorname{dim}\mathfrak{u}$.

These components $K_{-1},K_0$ satisfy the following invariance properties
\begin{equation}
    0=\begin{cases}K_0(\operatorname{ad}_XX',X'') + K_0(X',\operatorname{ad}_XX'')\\
    K_0(X\rhd Y,Y') + K_0(Y,X\rhd Y')\end{cases},\qquad
    0=K_{-1}(\operatorname{ad}_XX',Y) + K_{-1}(X',X\rhd Y)\nonumber
\end{equation}
for each $X,X',X''\in\mathfrak{u},Y,Y'\in V$. This gives rise to two linearly independent Casimirs \cite{Osei:2017ybk,Beisert:2017xqx}
\begin{equation}
    X^i\odot X_i+Y^i\odot Y_i,\qquad X^i\odot Y_i,\nonumber 
\end{equation}
where $\{X_i+Y_i\}_i$ is a basis of $V\rtimes\mathfrak{u}$ and $\{X^i+Y^i\}_i$ the corresponding dual basis induced by $K$.

\medskip

Now let $\mathfrak{g}^0$ denote the skeletal 2-algebra associated to the semidirect product $V\rtimes\mathfrak{u}$. Since $t=0$, we have $D(t)_{-1,0}=0$ and both of these maps have full kernel. Moreover, $V$ is Abelian, so we have $\Gamma_{t=0}=\mathfrak{g}_{-1}^{2\odot} = \operatorname{Cas}_{\mathfrak{g}_{-1}}$. However, since $D(t)_{-1}=0$, the only component of $R^\odot$ comes from 
\begin{equation}
    \operatorname{2Cas}_{\mathfrak{g}^0}[0]=\Theta_\rhd.\nonumber
\end{equation}
This space is spanned by the component $K_{-1}$ of the Killing form $K$. Interestingly, the component $K_0$ is constrained to live in the space $\Gamma_{t=0}$, but fails to contribute to $R^\odot$ as $D(t)_{-1}=0$. 

\medskip 

Since the 2-algebra structure on $\mathfrak{g}^0$ sums to the algebra structure on $V\rtimes\mathfrak{u}$ by Eqs. \eqref{seminull}, \eqref{translation}, the Schouten brackets in each case coincide --- one is the sum of the other \cite{Bai_2013}. However, the fact that $D(t)=0$ means that no contributions from $\mathfrak{g}_{-1}^{2\otimes}$ arise in $R$; in other words, the solutions for Eq. \eqref{mod2ybe} constitute those of Eq. \eqref{2cybe} lying in $(\mathfrak{u}\otimes V)\oplus(V\otimes\mathfrak{u})$. The correspondence is $r \mapsto r_0$.

\begin{example}\label{exp6}
Consider the 3D Poincar{\' e} algebra $\mathfrak{iso}(3) = \mathbb{R}^3\rtimes\mathfrak{so}(3)$, where the action $\mathfrak{so}(3)\rhd\mathbb{R}^3$ is the vector representation. Suppose $P_\mu,J_\mu,0\leq\mu\leq 2$ respectively span $\mathbb{R}^3,\mathfrak{so}(3)$, the following canonical Casimir 
\begin{equation}
    r^\odot(P,J) = K_{-1}(P,J) = P_\mu \odot J^\mu\nonumber
\end{equation}
gives rise via Eq. \eqref{mod2ybe} to the corresponding skew-symmetric part of the classical $r$-matrix
\begin{equation}
    r^\wedge = 2(P_+\wedge J_- - P_-\wedge J_+),\qquad \begin{cases}P_\pm = \frac{1}{2}(P_1\pm iP_2) \\ J_\pm = \frac{1}{2}(J_1\pm iJ_2)\end{cases}\nonumber
\end{equation}
which generates the following 1-cocycle $\psi$ on $\mathfrak{iso}(1,2)$:
\begin{eqnarray}
    \psi(P_0)=0,&\quad& \psi(P_\pm)= iP_\pm\wedge P_0,\nonumber \\
    \psi(J_0)=0,&\quad& \psi(J_\pm)= i(J_\pm\wedge P_0 + P_\pm \wedge J_0).\nonumber
\end{eqnarray}
The bialgebra $(\mathfrak{iso}(3);\psi)$ equipped with this 1-cocycle $\psi$ is known as the {\bf 3D $\kappa$-Poincar{\' e} group} \cite{Zakrzewski_1994}.
\medskip 

\begin{remark}
By appending an additional Casimir $\xi K_0(P,P)=\frac{\xi}{2}P_\mu\odot P^\mu$, where $\xi\in\mathbb{R}$, one obtains a one-parameter family of Hopf algebras that deforms the $\kappa$-Poincar{\' e} algebra.  However, this additional term is not a Drinfel'd twist \cite{Beisert:2017xqx}; moreover, the 1-cocycle $\psi$ would no longer be coisotropic $\psi(\mathfrak{so}(3)) \not\subset \mathfrak{so}(3)\wedge \mathbb{R}^3$. Interestingly, this deformation is forbidden to appear in the 2-graded classical $r$-matrix $R$.
\end{remark}

Now let us consider the skeletal suspension 2-algebra
\begin{equation}
    \mathfrak{so}(3)^0=(\mathbb{R}^{3}\xrightarrow{0}\mathfrak{so}(3),\rhd,[\cdot,\cdot]),\nonumber
\end{equation}
with $\rhd$ the vector representation. Recall, in case $t=0$, the quadratic 2-Casimir $R^\odot\in \operatorname{2Cas}_{\mathfrak{so}(3)^0}=\Theta_\rhd$ is given by the bilinear form $K_{-1}$. We recover both the Casimir $R^\odot=r^\odot$ as well as the classical $r$-matrix $R^\wedge = r^\wedge$ of the 3D Pincar{\'e} algebra $\mathfrak{iso}(3)$ \cite{Bai_2013}. The 1-cocycle $\psi$ is in fact coisotropic $\psi(\mathfrak{so}(3)) \subset \mathfrak{so}(3)\wedge \mathbb{R}^3$ \cite{Gutierrez-Sagredo:2019ipf}, hence {\bf Proposition \ref{suscocy}} states that $\psi = \delta_{-1}+\delta_0$ defines a 2-cocycle 
\begin{eqnarray}
    \delta_0(P_0)=0,&\quad& \delta_0(P_\pm) = iP_\pm\wedge P_0,\nonumber \\
    \delta_{-1}(J_0)=0,&\quad& \delta_{-1}(J_\pm)= i(J_\pm\wedge P_0- J_0\wedge P_\pm) \nonumber
\end{eqnarray}
on the suspension 2-algebra $\mathfrak{so}(3)^0$. In other words, the 2-bialgebra $(\mathfrak{so}(3)^0;\delta_{-1},\delta_0)$ is equivalent to the 3D $\kappa$-Poincar{\' e} algebra.
\end{example}

The above construction holds for the Poincar{\'e} algebra $\mathfrak{iso}(n)$ in any dimension. This result was already guessed in \cite{Girelli:2021khh}. Dimension $n=3$ is special, as the isotropy subgroup $\mathfrak{so}(3)$ has the same dimension as its quotient $\mathfrak{iso}(3)/\mathfrak{so}(3) \cong\mathbb{R}^3$. This means that deformations of $\mathfrak{iso}(3)$ can be obtained from Drinfel'd double structures on $\mathfrak{so}(3)$ \cite{Gutierrez-Sagredo:2019ipf}, as demonstrated above.

\subsubsection{2-Casimirs of the 2-Drinfel'd double}
Once one has constructed the quadratic 2-Casimirs for $\g$, we may similarly construct the quadratic 2-Casimirs $\operatorname{2Cas}_{\mathfrak{g}^*[1]}$ of the dual crossed-module $\mathfrak{g}^*[1]$. One may notice that elements in $\operatorname{2Cas}_\mathfrak{g}\subset \mathfrak{g}_0\odot\mathfrak{g}_{-1}$ and $\operatorname{2Cas}_{\mathfrak{g}^*[1]}\subset \mathfrak{g}_{-1}^*\odot\mathfrak{g}_0^*$ satisfy invariance properties constituting the first rows of Eqs. \eqref{invar2manin}. Moreover, the conditions 
\begin{equation}
    D(t)_0\operatorname{2Cas}_\mathfrak{g}=0,\qquad D(t^*)_0\operatorname{2Cas}_{\mathfrak{g}^*[1]} =0 \nonumber
\end{equation}
imply that $t,t^*$ are symmetric with respect to the basis diagonalizing the quadratic 2-Casimir elements. This is nothing but the duality condition Eq. \eqref{dualtmap}.

On the other hand, the 2-Manin triple $\mathfrak{d}=\mathfrak{g}\bowtie\mathfrak{g}^*[1]$ itself hosts a crossed-module structure given by Eq. \eqref{2maninbrac}. This allows us to use the above arguments to also construct the quadratic 2-Casimirs $\operatorname{2Cas}_\mathfrak{d}\subset \d_0\odot\d_{-1}\cong (\mathfrak{g}_0\oplus\mathfrak{g}_{-1}^*)\odot(\mathfrak{g}_{-1}\oplus\mathfrak{g}_0^*)$ of the 2-Manin triple. Here, the adjoint and coadjoint actions Eq. \eqref{rep} of $\mathfrak{g}$ and $\mathfrak{g}^*[1]$ on each other both participate in the definition of the operator $D(\operatorname{ad})$.

After a lengthy calculation, it can be explicitly shown that quadratic 2-Casimirs of $\mathfrak{d}$ satisfy the following invariance properties:
\begin{eqnarray}
    (\operatorname{ad}_{X'}X+ \chi^*_{X'}f)\odot (Y+g) + (X+f)\odot (X'\rhd Y + \operatorname{ad}_{X'}^*g) &=& 0,\nonumber \\
    (\eta_{f'}^*X+ \mathfrak{ad}_{f'}f) \odot (Y+g) + (X+f)\odot (\mathfrak{ad}_{f'}^*Y + f'\rhd^* g) &=& 0,\nonumber\\
    (X+f) \odot (\tilde\Delta_g(X') + \Delta_Y(f')) + (X'+f')\odot (\tilde\Delta_g(X) + \Delta_Y(f)) &=& 0\label{2ddinvar} 
\end{eqnarray}
for each $X,X'\in\mathfrak{g}_0,Y\in\mathfrak{g}_{-1},f,f'\in\mathfrak{g}_{-1}^*,g\in\mathfrak{g}_0^*$. By expanding each row of Eq. \eqref{2ddinvar} out, we see that these invariance properties encompass those of both the natural evaluation pairing $\langle\langle\cdot,\cdot\rangle\rangle$ of Eq. \eqref{2bilin}, as well as the grading-inhomogeneous alternative pairing $\langle\langle\cdot,\cdot\rangle\rangle'$ given in Eq. \eqref{2bilin2}. For instance, expanding out the first equation yields
\begin{eqnarray}
    0&=&(\operatorname{ad}_{X'}X \odot g + X \odot \operatorname{ad}_{X'}^*g) + (\chi^*_{X'}f \odot Y + f \odot X'\rhd Y)\nonumber \\
    &\qquad& + (\operatorname{ad}_{X'}X \odot Y + X \odot X'\rhd Y)  + (\chi^*_{X'}f\odot g + f\odot \operatorname{ad}_{X'}^*g)\nonumber\\
    &=& (g,\operatorname{ad}_{X'}X) + (\operatorname{ad}^*_{X'}g,X) + (\chi^*_{X'}f,Y) + (f,\chi_{X'}Y) \nonumber \\
    &\qquad& + \langle \operatorname{ad}_{X'}X,Y\rangle + \langle X,\chi_{X'}Y\rangle + \langle \chi^*_{X'}f,g\rangle + \langle f,\operatorname{ad}_{X'}^*g\rangle,\nonumber   
\end{eqnarray}
where $(\cdot,\cdot)$ and $\langle\cdot,\cdot\rangle$ are the components of the pairings Eqs. \eqref{2bilin}, \eqref{2bilin2}, respectively. Similar computations can be carried out for the other two equations. Moreover, the condition that $D(T)_{0}\operatorname{2Cas}_\mathfrak{d}=0$ implies the duality condition Eq. \eqref{dualbigt}. 

In other words, we have the following result.
\begin{proposition}\label{bilinedecomp}
Quadratic 2-Casimirs $\operatorname{2Cas}_\mathfrak{d}$ of the 2-Manin triple $\mathfrak{d}$ consist of the grading-inhomogeneous pairings Eqs. \eqref{2bilin}, \eqref{2bilin2}. 
\end{proposition}
\noindent Note that Eq. \eqref{2ddinvar} follows directly from the $D(\operatorname{ad})$-invariance of $R^\odot$ itself as an element of $\d_0\odot\d_{-1}$. Assumptions about its particular form, such as Eq. \eqref{2Rmatdecomp}, are not necessary.  In other words, {\bf Proposition \ref{bilinedecomp}} is a general result that applies to any 2-Drinfel'd double as defined here and in the literature \cite{Bai_2013,Chen:2013}. As we have seen previously, it is only for these pairings that the  BF theory Eq. \eqref{monsbf} is non-trivial. Here, we have provided a proof that these are in fact the only ones on $\mathfrak{d}$.

\begin{example}
Consider the case $t=1$ and the canonical 2-Manin triple $\mathfrak{1}_\mathfrak{l}$ associated to the bialgebra $(\mathfrak{l};\psi)$. Recall from {\bf Example \ref{exp3}} that we have $\mathfrak{1}_\mathfrak{l}=\operatorname{id}_\mathfrak{d}$ as crossed-modules, where $\mathfrak{d}=\mathfrak{l}\bowtie\mathfrak{l}^*$ is the Drinfel'd double.

First, as $t=1$ in the crossed-module $\operatorname{id}_\mathfrak{d}$, $\operatorname{2Cas}_{\operatorname{id}_\mathfrak{d}}$ coincides precisely with the usual quadratic Casimirs $\operatorname{Cas}_\mathfrak{d}$ (as explained in more detail below). As such, $\operatorname{2Cas}_{\operatorname{id}_\mathfrak{d}}$ consist of non-degenerate bilinear forms $K\in \mathfrak{d}\odot \mathfrak{d}$, satisfying the usual invariance properties. 

On the other hand, the quadratic 2-Casimirs of $\mathfrak{1}_{\mathfrak{l}}$ as a 2-Manin triple consists of the two grading-inhomogeneous pairings Eqs. \eqref{2bilin}, \eqref{2bilin2}, given by $K_0\in (\mathfrak{l}\odot\mathfrak{l})\oplus(\mathfrak{l}^*\odot\mathfrak{l}^*)$ and $K^0\in \mathfrak{l}\odot \mathfrak{l}^*$, as shown in {\bf Proposition \ref{bilinedecomp}}. The fact $\mathfrak{1}_\mathfrak{l}=\operatorname{id}_\mathfrak{d}$ then implies that we have a decomposition of $\operatorname{2Cas}_{\mathfrak{1}_\mathfrak{g}}$ given by
\begin{equation}
    K= K_0+K^0.\nonumber 
\end{equation}
More precisely, if $Z=(X+g)\in\mathfrak{d}$ with $X\in\mathfrak{g}$ and $g\in\mathfrak{g}^*$, then we have
\begin{equation}
    K(Z,Z') = K_0(X,X') + K_0(g,g') + K^0(g,X') + K^0(X,g').\nonumber
\end{equation}

As a consequence, the  BF theory Eq. \eqref{monsbf} reads
\begin{eqnarray}
    S_\textbf{BF}[{\bf A},\boldsymbol\Sigma] &=& \frac{1}{2}\int_M K_0(\Sigma\wedge \bar{F})+ K_0(B\wedge \bar{F}^*) + K^0(\Sigma\wedge\bar{F}^*) + K^0(B\wedge\bar{F})\nonumber\\
    &\qquad& +\frac{1}{2}\int_M K_0(\Sigma\wedge\Sigma) + K_0(B\wedge B) + K^0(\Sigma\wedge B) + K^0(B\wedge \Sigma),\nonumber
\end{eqnarray}
where $\boldsymbol\Sigma=\Sigma+B$ is valued in $(\operatorname{id}_\mathfrak{d})_{-1} = \mathfrak{l}\bowtie\mathfrak{l}^*$ and ${\bf A} = A+C$ is valued in $(\operatorname{id}_\mathfrak{d})_0 = \mathfrak{l}\bowtie\mathfrak{l}^*$. This means that the BF theory $S_\textbf{BF}[{\bf A},\boldsymbol\Sigma]$ based on $\mathfrak{d}$ naturally hosts fake-flatness equations of motion in each $\g,\g^*[1]$-sectors that mixes the 2-connections $\Sigma,B$. This is an artifact of the double structure of $\d = \g\bowtie\g^*[1]$; these "mixed" terms must appear in order to preserve the invariance of the action.

As a concrete example, consider 
the classical double $\mathfrak{d}=\mathfrak{su}(2)\bowtie\mathfrak{an}_2$. Suppose $\{\mathcal{J}^\mu\}_\mu$ is a basis for $\mathfrak{d}$ diagonalizing the non-degenerate invariant bilinear form $K\in\operatorname{2Cas}_{\operatorname{id}_{\mathfrak{d}}}=\operatorname{Cas}_{\mathfrak{d}}$. By writing $\mathcal{J}^\mu = (J^i,\tau^j)$  in terms of its sectors $J^i\in\mathfrak{su}(2),\tau^j \in\mathfrak{an}_2$, we see that
\begin{equation}
    \delta_{\mu\nu} = K(\mathcal{J}_\mu,\mathcal{J}_\nu) = K_0(J_i,J_j) + K_0(\tau_i,\tau_j) + K^0(J_i,\tau_j)+K^0(J_i,\tau_j),\nonumber
\end{equation}
However, it is not in general possible to pick the bases $\{J_i\}_i,\{\tau_i\}_i$ that diagonalize the two bilinear forms $K_0,K^0$ simultaneously --- that is, unless they commute. Generally speaking, it is important to keep in mind that a diagonal bilinear form $K$ can decompose into components that are non-diagonal. 
\end{example}

\section{Conclusions}
In this paper, we  extracted the structures of the  Drinfel'd double $\mathfrak{d}$ of a Lie bialgebra $\mathfrak{g}$, and subsequently derived the compatibility relations Eq. \eqref{1manin} from gauge-theoretic considerations. Moreover, we have generalized (ie. "categorified") this method to the case of the Lie 2-bialgebra/bialgebra crossed-module, and explicitly derived the compatibility relations Eqs. \eqref{2manin1}, \eqref{2manin2}, \eqref{2manin3} of the 2-Drinfel'd double through 2-gauge theory. 

It is important to note that our results, in particular {\bf Theorems \ref{mt}, \ref{2mt}}, agree completely with existing mathematical literature \cite{Chari:1994pz,Semenov1992,Chen:2013,Bai_2013}. Though these derived compatibility conditions are not strictly speaking new results, the new insight is that we can in fact derive them from gauge-theoretical considerations. This means that our approach opens the door for understanding the structure of higher Drinfel'd doubles in a similar manner. For instance, we could start with a 3-gauge theory \cite{Radenkovic:2019qme}, and apply our techniques to develop the notion of 3-Drinfel'd double associated to 3-bialgebras (which are   defined in \cite{bai2016}).

We had also briefly studied topological theories based on the (2-)Drinfel'd doubles. The {\it BF} theory Eq. \eqref{monsbf} is of particular interest. Indeed, it is an essential tool to construct quantum gravity amplitude in the spin foam approach \cite{Barrett:1997gw, rovelli2004}. It has been argued that it could be relevant to also include some edge decorations to have access to the frame field, which can be done using 2-gauge theories \cite{Crane:2003ep, Asante:2019lki, Girelli:2021khh}. To unlock new types of symmetries could therefore open new ways to construct spin foam models. 

In general, we believe that such a procedure would lead to a systematic understanding of the structures of a "quantum 2-group", which is very much sought after in many areas of physics \cite{Kapustin:2013uxa,Cordova:2018cvg, Benini:2018reh}.

\medskip 

We would like to highlight some  directions we find interesting to develop in a near future. 

\paragraph{\textbf{Recovering the Crane-Yetter-Broda amplitude.}} As we emphasized, the Drinfel'd double fits in a 3D Chern-Simons theory Eq. \eqref{chernsimons}, while the 2-Drinfel'd double fits in a 4D BF-theory Eq. \eqref{monsbf}. The trivial canonical embedding $\mathfrak{d}\mapsto \operatorname{id}_\mathfrak{d}$ can therefore be seen as the most natural way of "promoting" a 3D topological gauge theory to a 4D one, while keeping all of the algebra of the field contents identical. 
With this in mind,  it would interesting to see whether one can recover (more rigorously) the  Crane-Yetter-Broda  amplitude (ie. a 15-$j$ symbol)  defined in terms of  a quantum group \cite{Crane:1994fw, Baez:1999sr} as the amplitude associated to the $BF+BB$ action. This would in particular bring more strength to the quantum gravity models which implement the cosmological constant using a quantum group deformation. It would also clarify how and why the presence of a non-zero cosmological constant is also incorporated in the quantum gravity regime through a quantum group structure, in 4D.


\paragraph{\textbf{2-graded integrable systems.}}
In this paper, we studied the 2-graded classical $r$-matrices on $\mathfrak{d}$, as well as the quadratic 2-Casimirs. Particular examples in the cases $t=0,1$ were worked out explicitly, and we have found that, in these cases, the notion of duality for the 1-algebra case \cite{Semenov1992} all carry over to the Lie 2-algebra context, including the classification of the classical $r$-matrices. 

The notion of quantum/Poisson group was initially developed in the context of integrable systems. The use of the $r$-matrix is quite instrumental in understanding and classifying the notion of integral system. Now that we have Lie 2-bialgebras and the associated notion of $r$-matrix, we can explore how such formalism can be used to develop new notions of integrable systems. This is currently under investigation.

\paragraph{\textbf{Fourier transform between dual 2-groups.}} A key feature of the  Drinfel'd double is that we have a duality between Lie algebras, which can be extended to Hopf algebras \cite{Majid:1996kd}. In more physical terms, this duality pertains a notion of Fourier transform, relating possibly non-commutative structures \cite{Majid:2008iz, Guedes:2013vi, Joung:2008mr}. This notion of Fourier transform is different than (but related to) the standard Fourier transformation of group elements in terms of representations. The  representation theory is still missing for most 2-groups (except the skeletal case in \cite{Baez:2008hz}), so that if we intend to perform a Fourier transform for the notion of 2-groups, we need to resort to the duality between functions over dual 2-groups. Such duality is embedded in the definition of the 2-Drinfel'd double, hence we can expect many of the key-notions the plane-wave  need to satisfy should be recoverable from the (quantum) 2-Drinfel'd double.

\paragraph{Weak Lie 2-algebras and their 2-double.}
All along in this paper we have worked with strict Lie 2-algebras, which are equivalent to crossed-modules. However, there is in fact an additional piece of datum that characterizes a crossed-module $\g$. If we denote by $V=\operatorname{ker}t$ and $\operatorname{coker}t=\mathfrak{u}$ the kernel and cokernel of the crossed-module map $t$, then this additional datum is a Lie algebra cohomology class $\kappa\in H^3(\mathfrak{u},V)$ called the {\bf Postnikov class} \cite{Kapustin:2013uxa}, which classifies $\g$ up to elementary equivalence \cite{Wang:2016rzy}. Working explicitly with the crossed-module structure, it would be interesting to see how our above (2-)gauge-theoretic formalism extends to take the Postnikov class into account.

The motivation for doing so arises from the correspondence between crossed-modules with non-trivial Postnikov class and {\it weak 2-algebras}. In essence, the structure of a weak Lie 2-algebra introduces a skew-trilinear {\it homotopy map} $\eta: \mathfrak{g}_0^{3\wedge}\rightarrow\mathfrak{g}_{-1}$ that allows us to relax the Jacobi identity on $\mathfrak{g}$  \cite{Chen:2013}: the 2-Jacobi identities Eq. \eqref{2jacob} are modified to
\begin{eqnarray}
[X,[X',X'']]+[X',[X'',X]]+[X'',[X,X']]=t\eta(X,X',X''),\nonumber \\
X\rhd (X' \rhd Y) - X' \rhd (X\rhd Y) - [X,X']\rhd Y = \eta(X,X',tY)\label{weak2jac}
\end{eqnarray}
for each $X,X',X''\in\mathfrak{g}_0,Y\in\mathfrak{g}_{-1}$. Such structures appear most naturally in string theory, where one has the {\it string 2-algebra} \cite{Kim:2019owc}. Notably, weak 2-algebras have also appeared in the study of {\it 2-plectic geometry} \cite{BaezRogers}.

The trouble with weak 2-algebras is that they do not integrate to a 2-group. Due to the non-trivial Jacobiator, namely the right-hand sides of Eq. \eqref{weak2jac}, the associator at the group level cannot be made continuous \cite{Baez:2005sn}. The solution, given in Ref. \cite{Baez:2005sn}, is precisely to pass from the description of a weak 2-algebra to a crossed-module with non-trivial Postnikov class. As such, based on the notion of a weak 2-bialgebra already in the literature \cite{Chen:2013}, our formalism would allow one to construct, in a manner in which the gauge content is explicit, a "weak" classical 2-Drinfel'd double associated to a bialgebra crossed-module with non-trivial Postnikov class. 

Such a gauge-theoretic understanding of the structures of a weak 2-Drinfel'd double is useful, as this (or more precisely the finite group version) will allow us to pin down the 4D boundary theory of the 5D symmetry-protected topological order $\int_M w_2\cup w_3$ \cite{Thorngren2015}. This order hosts fermionic point-like and string-like quasiparticles, which was also recently discovered to be related to the {\it new} spin-$\frac{3}{2}$ $SU(2)$ anomaly \cite{JuvenWang}. This particular order has received much attention recently from the condensed matter theory community, thus it would be important to understand its symmetries.

\appendix

\section{Proofs of generalized (2-)covariance}\label{gencov}
In this section, we provide an explicit proof of the generalized covariance of the combined curvature quantities Eqs. \eqref{combocurv}, \eqref{comboff} and \eqref{combo2c} constructed in the main text. These proofs shall demonstrate the importance of the compatibility conditions Eq. \eqref{1manin} and Eqs. \eqref{2manin1}, \eqref{2manin2}, \eqref{2manin3} in the context of gauge and 2-gauge theory, respectively.

\paragraph{Covariance of the combined curvature.} Recall the combined curvature ${\bf F}$ constructed in Eq. \eqref{combocurv}. It takes values in a Drinfel'd double $\mathfrak{d}=\mathfrak{g}\bowtie\mathfrak{g}^*$ of a Lie bialgebra $\mathfrak{g}$.

\begin{theorem}
\label{cov}
Given {\bf Theorem \ref{mt}}, Eq. \eqref{combocurv} is covariant
\begin{equation}
    {\bf F} \rightarrow {\bf F}^{\boldsymbol\lambda} = {\bf F} + \pmb{[}{\bf F},\boldsymbol\lambda\pmb{]}.\nonumber
\end{equation}
\end{theorem}
\noindent With Eq. \eqref{maninbrac}, we may compute
\begin{equation}
    \pmb{[}{\bf F},\boldsymbol\lambda\pmb{]} = ([\bar F,\lambda]+\mathfrak{ad}_{\tilde\lambda}^*\bar F - \mathfrak{ad}_{\bar F^*}^*\lambda) + ([\bar F^*,\tilde\lambda]+\operatorname{ad}_\lambda^*\bar F^* - \operatorname{ad}_{\bar F}^*\tilde\lambda). \nonumber
\end{equation}
\begin{proof}
Due to the symmetry under duality 
\begin{equation}
    \mathfrak{g}\leftrightsquigarrow \mathfrak{g}^*,\qquad \operatorname{ad}^*\leftrightsquigarrow\mathfrak{ad}^*,\nonumber
\end{equation}
we may without loss of generality (WLOG) focus on one component of ${\bf F}$, say $\bar F$ that lies in the $\mathfrak{g}$-sector. This leads to the transformations
\begin{eqnarray}
    [A\wedge A]&\xrightarrow{\boldsymbol\lambda}& [A\wedge A] \nonumber \\
    &\qquad&+ 2[d_A\lambda\wedge A]  + 2[\mathfrak{ad}_{\tilde\lambda}^*A\wedge A] \nonumber \\
    &\qquad& -2[A\wedge\mathfrak{ad}_B^*\lambda] +o(\boldsymbol\lambda^2),\nonumber \\
    \mathfrak{ad}_B^*(\wedge A)&\xrightarrow{\boldsymbol\lambda}& \mathfrak{ad}_B^*(\wedge A) \nonumber \\
    &\qquad& + \mathfrak{ad}_B^*(\wedge d_A\lambda) + \mathfrak{ad}_B^*(\wedge \mathfrak{ad}_{\tilde\lambda}^*A) + \mathfrak{ad}_{d_B\tilde\lambda}^*(\wedge A) + \mathfrak{ad}_{\operatorname{ad}_\lambda^*B}^*(\wedge A)  \nonumber \\
    &\qquad&- \mathfrak{ad}_{\operatorname{ad}_A^*\tilde\lambda}^*A -\mathfrak{ad}_B^*(\wedge \mathfrak{ad}_B^*\lambda)+ o(\boldsymbol\lambda^2),\nonumber
\end{eqnarray}
where we have neglected terms of quadratic and higher order in $\boldsymbol\lambda=\lambda+\tilde\lambda$.

We first rewrite using Eq. \eqref{1manin}
\begin{equation}
    2[\mathfrak{ad}_{\tilde\lambda}^*A\wedge A] = \mathfrak{ad}_{\tilde\lambda}^*[A\wedge A] + 2\mathfrak{ad}_{\operatorname{ad}_A^*\tilde\lambda}^*(\wedge A),\nonumber
\end{equation}
and hence 
\begin{equation}
    [\mathfrak{ad}_{\tilde\lambda}^*A\wedge A] - \mathfrak{ad}_{\operatorname{ad}_A^*\tilde\lambda}^*A =\frac{1}{2} \mathfrak{ad}_{\tilde\lambda}^*[A\wedge A].\nonumber
\end{equation}
On the other hand, we have
\begin{equation}
    \mathfrak{ad}_{d_B\tilde\lambda}^*(\wedge A)= \mathfrak{ad}_{\tilde\lambda}^* dA + \mathfrak{ad}_{[B,\tilde\lambda]}^*(\wedge A),\nonumber
\end{equation}
modulo the boundary term $d(\mathfrak{ad}_{\tilde\lambda}^*A)$. By the identity\footnote{This follows from the Jacobi identity for the dual bracket $[\cdot,\cdot]_*$:
\begin{equation}
    (\mathfrak{ad}_{g'}\mathfrak{ad}_g)g''=[g',[g,g'']_*]_*= -[g,[g'',g']_*]_* - [g'',[g',g]_*]_* = (\mathfrak{ad}_g\mathfrak{ad}_{g'})g'' + \mathfrak{ad}_{[g',g]_*}g''.\nonumber
\end{equation}}
\begin{equation}
    \mathfrak{ad}_g^*\mathfrak{ad}_{g'}^* = \mathfrak{ad}_{g'}^*\mathfrak{ad}_g^* + \mathfrak{ad}_{[g',g]_*}^*\label{jacob}
\end{equation}
for each $g,g',g''\in\mathfrak{g}^*$, this leads to 
\begin{equation}
    \mathfrak{ad}_B^*(\wedge \mathfrak{ad}_{\tilde{\lambda}}^*A) = \mathfrak{ad}_{\tilde\lambda}^* \mathfrak{ad}_B^*(\wedge A) + \mathfrak{ad}_{[\tilde\lambda,B]_*}^*A=\mathfrak{ad}_{\tilde\lambda}^* \mathfrak{ad}_B^*(\wedge A) - \mathfrak{ad}_{[B,\tilde\lambda]_*}^*A,\nonumber
\end{equation}
hence 
\begin{equation}
    \mathfrak{ad}_B^*(\wedge \mathfrak{ad}_{\tilde{\lambda}}^*A)+\mathfrak{ad}_{d_B\tilde\lambda}^*(\wedge A) = \mathfrak{ad}_{\tilde\lambda}^*dA + \mathfrak{ad}_{\tilde\lambda}^*\mathfrak{ad}_B^*(\wedge A).\nonumber
\end{equation}
We may then compute the combination
\begin{eqnarray}
     [\mathfrak{ad}_{\tilde\lambda}^*A\wedge A]+\mathfrak{ad}_B^*(\wedge \mathfrak{ad}_{\tilde{\lambda}}^*A)\nonumber\\ 
     \qquad +\mathfrak{ad}_{d_B\tilde\lambda}^*(\wedge A)- \mathfrak{ad}_{\operatorname{ad}_A^*\tilde\lambda}^*A
     &=& \mathfrak{ad}_{\tilde\lambda}^*dA + \mathfrak{ad}_{\tilde\lambda}^*\mathfrak{ad}_B^*(\wedge A) + \frac{1}{2}\mathfrak{ad}_{\tilde\lambda}^*[A\wedge A] \nonumber \\
     &=& \mathfrak{ad}_{\tilde\lambda}^*\tilde F.\label{cov1}
\end{eqnarray}

Turning to the other terms, we first note the standard manipulation
\begin{equation}
    [d_A\lambda\wedge A] = -[\lambda\wedge dA] + [[A,\lambda]\wedge A] =  [dA,\lambda]+ \frac{1}{2}[[A\wedge A],\lambda],\nonumber
\end{equation}
where we have discarded a boundary term $d[\lambda,A]$ and used the Jacobi identity
\begin{equation}
     [[A,\lambda]\wedge A] - [[\lambda,A] \wedge A] + [[A\wedge A],\lambda]= 2[[A,\lambda]\wedge A] + [[A\wedge A],\lambda] =  0.\nonumber
\end{equation}
Secondly, from Eq. \eqref{1manin} we have
\begin{equation}
    \mathfrak{ad}_B^*([A,\lambda]) = [\mathfrak{ad}_B^*(\wedge A),\lambda] +[A\wedge \mathfrak{ad}_B^*\lambda] - \mathfrak{ad}_{\operatorname{ad}_\lambda^*B}^*(\wedge A) - \mathfrak{ad}_{\operatorname{ad}_A^*(\wedge B)}\lambda,\nonumber
\end{equation}
which leads to
\begin{equation}
    \mathfrak{ad}_B^*(\wedge d_A\lambda) = -\mathfrak{ad}_{dB}^*\lambda + [\mathfrak{ad}_B^*(\wedge A),\lambda] +[A\wedge \mathfrak{ad}_B^*\lambda] - \mathfrak{ad}_{\operatorname{ad}_\lambda^*B}^*(\wedge A) - \mathfrak{ad}_{\operatorname{ad}_A^*(\wedge B)}^*\lambda\nonumber
\end{equation}
modulo the boundary term $d(\mathfrak{ad}_B^*\lambda)$. As such, we see that 
\begin{equation}
    \mathfrak{ad}_B^*(\wedge d_A\lambda) + \mathfrak{ad}_{\operatorname{ad}_\lambda^*B}^*(\wedge A) =  \mathfrak{ad}_{dB}^*\lambda + [\mathfrak{ad}_B^*(\wedge A),\lambda] +[A\wedge \mathfrak{ad}_B^*\lambda]+ \mathfrak{ad}_{\operatorname{ad}_A^*(\wedge B)}^*\lambda.\nonumber
\end{equation}
Finally, we can rewrite
\begin{equation}
    2\mathfrak{ad}_B^*(\wedge\mathfrak{ad}_B^*\lambda) = \mathfrak{ad}_{[B\wedge B]_*}^\lambda \nonumber
\end{equation}
by Eq. \eqref{jacob}. Hence, combining the relevant gauge-dependent terms, we have
\begin{eqnarray}
    [d_A\lambda\wedge A] + \mathfrak{ad}_B^*(\wedge d_A\lambda) + \mathfrak{ad}_{\operatorname{ad}_\lambda^*B}^*(\wedge A) \nonumber \\
    - [A\wedge\mathfrak{ad}_B^*\lambda] - \mathfrak{ad}_B^*(\wedge\mathfrak{ad}_B^*\lambda) &=& [dA,\lambda]+ \frac{1}{2}[[A\wedge A],\lambda] + [\mathfrak{ad}_B^*(\wedge A),\lambda],\nonumber\\
    &\qquad&  -\mathfrak{ad}_{dB}^*\lambda - \mathfrak{ad}_{\operatorname{ad}_A^*(\wedge  B)}^*\lambda - \frac{1}{2}\mathfrak{ad}_{[B\wedge B]_*}^*\lambda\nonumber\\
    &=& [\tilde F,\lambda] -\mathfrak{ad}_{\tilde F^*}^*\lambda.\label{cov2}
\end{eqnarray}
Summing Eqs. \eqref{cov1} and \eqref{cov2} proves the lemma.
\end{proof}

\paragraph{2-Covriance of the combined fake-flatness.} Now let us turn to the case of the 2-Deinfel'd double $\mathfrak{d}=\mathfrak{g}\bowtie\mathfrak{g}^*[1]$. Recall that the combined fake-flatness $\boldsymbol{\mathcal{F}}$ defined in Eq. \ref{comboff}, which takes values in the degree-0 component $\mathfrak{g}_0\oplus\mathfrak{g}_{-1}^*$ of $\mathfrak{d}$.

\begin{lemma}
\label{2cov1}
Given {\bf Theorem \ref{2mt}}, the combined fake-flatness Eq. \eqref{comboff} is 2-covariant
\begin{equation}
    \boldsymbol{\mathcal{F}}\rightarrow\begin{cases} \boldsymbol{\mathcal{F}}^{\boldsymbol\lambda} = \boldsymbol{\mathcal{F}} + \llbracket \boldsymbol{\mathcal{F}},\boldsymbol\lambda\rrbracket \\ \boldsymbol{\mathcal{F}}^{\bf L} = \boldsymbol{\mathcal{F}}\end{cases}\nonumber
\end{equation}
under the 2-gauge symmetry parameterized by $\boldsymbol\lambda=\lambda+\tilde\lambda$ and ${\bf L}=L+\tilde L$.
\end{lemma}
\begin{proof}
Let us first prove the covariance of the 1-curvature ${\bf F}$. For this, we can copy the proof of {\bf Theorem \ref{cov}}, with the modification that $\operatorname{ad},\mathfrak{ad}$ are replaced by $\chi,\eta$, respectively. However, we note that the argument there relies crucially on two things: the  condition Eq. \eqref{1manin} and the Jacobi identity. These conditions are analogous in the 2-algebra context to the first  conditions Eqs. \eqref{2manin1} and the Jacobi identities Eqs. \eqref{2jacob}, \eqref{dual2jacob}. This yields the desired covariance
\begin{equation}
    {\bf F}\rightarrow {\bf F}^{\boldsymbol\lambda} = {\bf F} + \pmb{[} {\bf F},\boldsymbol\lambda\pmb{]} = {\bf F} + ([\bar F,\lambda]+ \eta_{\tilde\lambda}^*\bar F - \eta_{\bar{F}^*}^*\lambda) + ([\bar F,\tilde\lambda]_*+\chi_\lambda^*\bar{F}^* - \chi_{\bar F}^*\tilde\lambda).\nonumber
\end{equation}

Next, we examine how $T\boldsymbol\Sigma$ transforms under $\boldsymbol\lambda$. In the $\mathfrak{g}$-valued component, we have
\begin{equation}
    t\Sigma \rightarrow t\Sigma + t(\lambda\rhd \Sigma) + t\mathfrak{ad}_{\tilde\lambda}^*\Sigma - \eta_{\tilde tB}^*\lambda.\nonumber
\end{equation}
By the Pfeiffer identity on $\mathfrak{g}$, we can write $t(\lambda\rhd \Sigma) = [\lambda,t\Sigma]$. Now if moreover $\tilde t = t^*$ such that $t = \tilde t^*$, then
\begin{equation}
    t\mathfrak{ad}_{\tilde\lambda}^*\Sigma = \tilde t^*\mathfrak{ad}_{\tilde\lambda}^*\Sigma = \eta_{\tilde\lambda}^*\tilde t^*\Sigma = \eta_{\tilde\lambda}^*t\Sigma\nonumber
\end{equation}
by Eq. \eqref{comm}. This achieves the desired covariance
\begin{equation}
    \bar{\mathcal{F}} \rightarrow \bar{\mathcal{F}} + [\bar F,\lambda]+ \eta_{\tilde\lambda}^*\bar{\mathcal{F}} - \eta_{\bar{\mathcal{F}}^*}^*\lambda =  \bar{\mathcal{F}} + \pmb{[} \boldsymbol{\mathcal{F}},\boldsymbol\lambda\pmb{]}|_{\mathfrak{g}}.\nonumber
\end{equation}
A similar argument proves the covariance of $\bar{\mathcal{F}}^*$ in the dual sector $\mathfrak{g}^*[1]$.

Next, we compute that 
\begin{eqnarray}
    d{\bf A}&\rightarrow& d{\bf A} + dT{\bf L} = d{\bf A} + Td{\bf L},\nonumber \\
    \pmb{[}{\bf A}\wedge {\bf A}\pmb{]} &\rightarrow& \pmb{[}{\bf A}\wedge {\bf A}\pmb{]} + 2\pmb{[}{\bf A}\wedge T{\bf L}\pmb{]} + \pmb{[} T{\bf L}\wedge T{\bf L}\pmb{]}\nonumber\\
    &\qquad& = \pmb{[}{\bf A}\wedge{\bf A}\pmb{]} + 2T({\bf A} \wedge^{\mathrlap{\rhd}{\triangleright}} {\bf L}) + o({\bf L}^2).\nonumber
\end{eqnarray}
If we neglect quadratic terms in ${\bf L}^2$ as we have done in the main text, then we yield the expected 2-gauge transformation
\begin{equation}
    {\bf F}\rightarrow {\bf F}^{\bf L} = {\bf F} + Td{\bf L} + T{\bf A}\wedge^{\mathrlap{\rhd}{\triangleright}}{\bf L} + \frac{1}{2}T\pmb{[}{\bf L}\wedge{\bf L}\pmb{]}\equiv {\bf F} + Td_{\bf A}{\bf L}.\nonumber
\end{equation}
As we have $\boldsymbol\Sigma \rightarrow\boldsymbol\Sigma + d_{\bf A}{\bf L}$ from Eq. \eqref{combo2gau}, we see that 
\begin{equation}
    \boldsymbol{\mathcal{F}} \rightarrow {\bf F}^{\bf L} - T\boldsymbol\Sigma^{\bf L} = {\bf F} - T\boldsymbol\Sigma = \boldsymbol{\mathcal{F}}\nonumber
\end{equation}
as desired.



\end{proof}

\paragraph{2-Covariance of the combined 2-curvature.} Recall the combined 2-curvature $\boldsymbol{\mathcal{G}}$ defined in Eq. \ref{combo2c}, which takes values in the degree-(-1) component $\mathfrak{g}_{-1}\oplus\mathfrak{g}_0^*$ of the 2-Drinfel'd double $\mathfrak{d}$.

\begin{lemma}
\label{2cov2}
Given {\bf Theorem \ref{2mt}}, the combined 2-curvature in Eq. \eqref{combo2c} is 2-covariant:
\begin{equation}
    \boldsymbol{\mathcal{G}}\rightarrow\begin{cases}\boldsymbol{\mathcal{G}}^{\boldsymbol\lambda} = \boldsymbol{\mathcal{G}} + \boldsymbol\lambda~\mathrlap{\triangleright}{\rhd}~ \boldsymbol{\mathcal{G}} \\ \boldsymbol{\mathcal{G}}^{\bf L} = \boldsymbol{\mathcal{G}} + \boldsymbol{\mathcal{F}}\wedge^{\mathrlap{\triangleright}{\rhd}} {\bf L}\end{cases}.\nonumber
\end{equation}
\end{lemma}
\begin{proof}
Here we make heavy use of Eqs. \eqref{combo2gau}. We first deal with $\boldsymbol\lambda$:
\begin{eqnarray}
    d\boldsymbol\Sigma &\rightarrow& d\boldsymbol\Sigma - d\boldsymbol\lambda ~\mathrlap{\triangleright}{\rhd}~\boldsymbol\Sigma + \boldsymbol\lambda ~\mathrlap{\triangleright}{\rhd}~ d\boldsymbol\Sigma,\nonumber \\
    {\bf A}\wedge^{\mathrlap{\triangleright}{\rhd}}\boldsymbol\Sigma &\rightarrow& {\bf A}\wedge^{\mathrlap{\triangleright}{\rhd}}\boldsymbol\Sigma + d_{\bf A}\boldsymbol\lambda \wedge^{\mathrlap{\triangleright}{\rhd}}\boldsymbol\Sigma + {\bf A}\wedge^{\mathrlap{\triangleright}{\rhd}}(\boldsymbol\lambda~\mathrlap{\triangleright}{\rhd}~\boldsymbol\Sigma) \nonumber \\
    &\qquad& + o(\boldsymbol\lambda^2).\nonumber
\end{eqnarray}
By the 2-Jacobi identity for $\mathfrak{d}$, we have
\begin{equation}
    {\bf A}\wedge^{\mathrlap{\triangleright}{\rhd}}(\boldsymbol\lambda~\mathrlap{\triangleright}{\rhd}~\boldsymbol\Sigma) = \boldsymbol\lambda~\mathrlap{\triangleright}{\rhd}~ ({\bf A}\wedge^{\mathrlap{\triangleright}{\rhd}}\boldsymbol\lambda\boldsymbol\Sigma) - \pmb{[}{\bf A},\boldsymbol\lambda\pmb{]} ~\mathrlap{\triangleright}{\rhd}~ \boldsymbol\Sigma,\nonumber
\end{equation}
hence
\begin{equation}
     {\bf A}\wedge^{\mathrlap{\triangleright}{\rhd}}\boldsymbol\Sigma \rightarrow {\bf A}\wedge^{\mathrlap{\triangleright}{\rhd}}\boldsymbol\Sigma + d\boldsymbol\lambda \wedge^{\mathrlap{\triangleright}{\rhd}}\boldsymbol\Sigma + \boldsymbol\lambda~\mathrlap{\triangleright}{\rhd}~ ({\bf A}\wedge^{\mathrlap{\triangleright}{\rhd}}\boldsymbol\lambda\boldsymbol\Sigma). \nonumber
\end{equation}
The sum then transforms accordingly:
\begin{equation}
    \boldsymbol{\mathcal{G}} \rightarrow \boldsymbol{\mathcal{G}} + \boldsymbol\lambda ~\mathrlap{\triangleright}{\rhd}~ \boldsymbol{\mathcal{G}}.\nonumber
\end{equation}

Next, we deal with ${\bf L}$:
\begin{eqnarray}
    d\boldsymbol\Sigma &\rightarrow& d\boldsymbol\Sigma + d({\bf A}\wedge^{\mathrlap{\triangleright}{\rhd}}{\bf L}) + o({\bf L}^2),\nonumber \\
    {\bf A}\wedge^{\mathrlap{\triangleright}{\rhd}}\boldsymbol\Sigma &\rightarrow& {\bf A}\wedge^{\mathrlap{\triangleright}{\rhd}}\boldsymbol\Sigma + T{\bf L}\wedge^{\mathrlap{\triangleright}{\rhd}}\boldsymbol\Sigma + {\bf A}\wedge^{\mathrlap{\triangleright}{\rhd}} d_{\bf A}{\bf L} +  o({\bf L}^2),\nonumber 
    \nonumber
\end{eqnarray}
where we consider terms only up to linear order in ${\bf L}$, as always. By the Leibniz rule we have
\begin{equation}
    d({\bf A}\wedge^{\mathrlap{\triangleright}{\rhd}}{\bf L}) = d{\bf A}\wedge^{\mathrlap{\triangleright}{\rhd}}{\bf L} - {\bf A}\wedge^{\mathrlap{\triangleright}{\rhd}}d{\bf L},\nonumber
\end{equation}
hence the sum $\boldsymbol{\mathcal{G}}= d\boldsymbol\Sigma + {\bf A}\wedge^{\mathrlap{\triangleright}{\rhd}}\boldsymbol\Sigma$ transforms as
\begin{eqnarray}
    \boldsymbol{\mathcal{G}} &\rightarrow& \boldsymbol{\mathcal{G}} + d{\bf A}\wedge^{\mathrlap{\triangleright}{\rhd}}{\bf L} + T{\bf L}\wedge^{\mathrlap{\triangleright}{\rhd}}\boldsymbol\Sigma+ {\bf A}\wedge^{\mathrlap{\triangleright}{\rhd}}({\bf A}\wedge^{\mathrlap{\triangleright}{\rhd}}{\bf L}).\nonumber
\end{eqnarray}

By the second Pfeiffer identity for $\mathfrak{d}$, we have
\begin{eqnarray}
    T{\bf L}\wedge^{\mathrlap{\triangleright}{\rhd}}\boldsymbol\Sigma &=& \pmb{[} {\bf L}\wedge\boldsymbol\Sigma\pmb{]} = -\pmb{[}\boldsymbol\Sigma \wedge{\bf L}\pmb{]} = -T\boldsymbol\Sigma \wedge^{\mathrlap{\triangleright}{\rhd}}{\bf L},\nonumber
\end{eqnarray}
while by the 2-Jacobi identity for $\mathfrak{d}$ we have
\begin{equation}
    {\bf A}\wedge^{\mathrlap{\triangleright}{\rhd}}({\bf A}\wedge^{\mathrlap{\triangleright}{\rhd}}{\bf L}) = \pmb{[}{\bf A}\wedge{\bf A}\pmb{]} \wedge^{\mathrlap{\triangleright}{\rhd}} {\bf L} - {\bf A}\wedge^{\mathrlap{\triangleright}{\rhd}}({\bf A}\wedge^{\mathrlap{\triangleright}{\rhd}}{\bf L}).\nonumber
\end{equation}
These terms combine to yield
\begin{equation}
    (\frac{1}{2}\pmb{[}{\bf A}\wedge{\bf A}\pmb{]}-T\boldsymbol{\Sigma}) \wedge^{\mathrlap{\triangleright}{\rhd}} {\bf L}.\nonumber
\end{equation}

The 2-curvature then transforms expectedly:
\begin{eqnarray}
    \boldsymbol{\mathcal{G}}&\rightarrow& \boldsymbol{\mathcal{G}} + (d{\bf A} + \frac{1}{2}\pmb{[}{\bf A}\wedge{\bf A}\pmb{]} - T\boldsymbol\Sigma)\wedge^{\mathrlap{\triangleright}{\rhd}}{\bf L}\nonumber \\ 
    &=&\boldsymbol{\mathcal{G}} + \boldsymbol{\mathcal{F}}\wedge^{\mathrlap{\triangleright}{\rhd}}{\bf L}.\nonumber
\end{eqnarray}
\end{proof}

The proofs of {\bf Theorems \ref{cov}, \ref{2cov1}, \ref{2cov2}} state that the compatibility conditions for the (2-)Manin triple are in fact crucial not only for the non-ambiguity in the (2-)gauge transformations, but also for ensuring the covariance for the various curvature quantities.





\newpage
\bibliographystyle{Biblio}
\bibliography{biblio}

\end{document}